\documentclass[12pt]{article}
\usepackage{amssymb}
\usepackage{amsmath}
\usepackage{amsthm}
\usepackage{cancel}
\usepackage{slashed}
\usepackage{fullpage}
\usepackage{color}
\usepackage{braket}
\usepackage{graphicx}
\usepackage[small]{subfigure}
\usepackage{multirow}
\usepackage{verbatim}
\usepackage{cite}
\usepackage{bigints}
\usepackage{simplewick}
\usepackage{authblk}
\usepackage{hyperref}
\hypersetup{
    linktocpage,
     colorlinks,
     citecolor=darkgreen,
     linkcolor= darkgreen,
     urlcolor=darkgreen
}

\def\nn{\nonumber}

\def\cC{\mathcal{C}} 
\def\cD{\mathcal{D}}

\def\cL{\mathcal{L}}
\def\cM{\mathcal{M}}

\def\cO{\mathcal{O}}

\def\cS{\mathcal{S}}

\def\cW{\mathcal{W}}

\def\cY{\mathcal{Y}}

\def\tr{{\rm tr}}

\def\And{\quad {\rm and} \quad}

\def\For{\quad {\rm for} \quad}

\def\dg{\dagger}

\def\pie{+i\varepsilon}

\def\too{\longrightarrow}

\newcommand{\fint}[1]{\int \!\! \frac{d^4 #1}{(2\pi)^4}\,}

\newcommand{\cn}[1]{\cancel{#1}}
\newcommand{\Sl}[1]{\slashed{#1}}
\newcommand{\fd}[2]{\parbox{#1}{\includegraphics[width=#1]{#2}}}

\newcommand{\LQCD}{\Lambda_{\text{QCD}}}

\def\be{\begin{equation}}
\def\ee{\end{equation}}

\def\l{\langle}
\def\r{\rangle}

\newcommand{\adj}{ {\mathrm{adj}}}

\newcommand{\Eq}[1]{Eq.~\eqref{#1}}

\allowdisplaybreaks

\def\LPeq{\cong}
\def\LPFeq{\cong_{\text{\tiny{IR}}}}

\def\soft{ { \red \text{soft-sens.}} }
\def\nsoft{ {\blue \text{not-soft-sens.} }}

\newcommand{\notparallel}{{\,\not{\!\parallel}\,}}

\newtheorem{lemma}{Lemma}

\newtheorem*{conjecture}{Conjecture}

\newcommand{\lr}[2]{ {(#1#2)} }

\definecolor{darkred}{rgb}{0.5,0.0,0.0}
\definecolor{darkblue}{rgb}{0.0,0.0,0.9}
\definecolor{darkerblue}{rgb}{0.0,0.0,0.5}
\definecolor{darkgreen}{rgb}{0.0,0.5,0.0}
\definecolor{black}{rgb}{0.0,0.0,0.0}
\definecolor{brown}{rgb}{0.6,0.4,0.2}
\newcommand{\red}{\color{darkred}}
\newcommand{\blue}{\color{darkblue}}

\newcommand{\black}{\color{black}}

\newcommand{\ccol}{\color{darkblue}}
\newcommand{\suncol}{\color{darkgreen}}
\newcommand{\spincol}{\color{brown}}
\newcommand{\softcol}{\color{darkred}}

\newcommand{\scs}{ {{\softcol s}} }
\newcommand{\scred}{ {{\softcol{\text{red}}}}}

\newcommand{\suO}{ {{\suncol 1}} }
\newcommand{\suT}{ {{\suncol 2}} }

\newcommand{\sua}{ {{\suncol a}} }
\newcommand{\sub}{ {{\suncol b}} }
\newcommand{\suc}{ {{\suncol c}} }

\newcommand{\suj}{ {{\suncol j}} }
\newcommand{\suI}{ {{\suncol I}} }

\newcommand{\suh}{ {{\suncol h}} }
\newcommand{\suhp}{ {{\suncol h'}} }
\newcommand{\suTT}{ {{\suncol T}} }

\newcommand{\gG}{ {\spincol{\Gamma} }}

\newcommand{\gGG}{ {\spincol{\Gamma}}}
\newcommand{\gpm}{ {\spincol{\pm}}}

\newcommand{\ccF}{ {{\ccol 4}} }
\newcommand{\ccTH}{ {{\ccol 3}} }
\newcommand{\ccO}{ {{\ccol 1}} }
\newcommand{\ccT}{ {{\ccol 2}} }
\newcommand{\cci}{ {{\ccol i}} }
\newcommand{\ccj}{ {{\ccol j}} }
\newcommand{\cck}{ {{\ccol k}} }
\newcommand{\rN}{ {{\ccol N}} }
\newcommand{\ccm}{ {{\ccol m}} }
\newcommand{\ccn}{ {{\ccol n}} }
\newcommand{\ccN}{ {{\ccol N}} }
\newcommand{\ccc}{ {{\ccol c}} }

\newcommand{\eir}{ \varepsilon_{\text{IR}} }
\newcommand{\euv}{ \varepsilon_{\text{UV}} }
\newcommand{\cdt}{ \hspace{-1mm}\cdot\hspace{-1mm}}

\newcommand{\GJS}{G_{J_\ccj\otimes S}}

\newcommand{\Eik}{S}
\newcommand{\SL}{ {\red {\cS}}}

\makeatletter
\newcommand*{\hyperlinkcite}[1]{\hyper@link{cite}{cite.#1\@extra@b@citeb}}
\makeatother
\newcommand{\tree}{[\hyperlinkcite{Feige:2013zla}{FS1}]}

\newcommand{\ltag}[1]{ {\rm(}{\red\bf #1}{\rm )} }

\newcommand{\Sp}{ \mathbf{Sp} } 
\newcommand{\JJ}{ \mathbf{J} } 
\newcommand{\TT}{ \mathbf{T} }

\newcommand{\epsq}{\epsilon}

\newcommand{\Sij}{ S_{\cci \ccj}} 
\newcommand{\sij}{ s_{\cci \ccj}}

\newcommand{\IRZ}{\widehat{Z} }
\title{Hard-Soft-Collinear Factorization to All Orders}

\author{Ilya Feige\thanks{feige@physics.harvard.edu}  }

\author{Matthew D. Schwartz\thanks{schwartz@physics.harvard.edu} }

\affil{\emph{Center for the Fundamental Laws of Nature,
Harvard University, Cambridge, MA 02138, USA}}

\begin{document} 
\maketitle

\begin{abstract}
We provide a precise statement of hard-soft-collinear factorization of scattering amplitudes and prove it to all orders in perturbation theory.
Factorization is formulated as the equality at leading power of  scattering
amplitudes in QCD with other amplitudes in QCD computed from a  product of operator matrix elements. 
The equivalence is regulator independent and gauge independent.
As the formulation relates amplitudes to the same amplitudes with additional soft or collinear particles, it includes as special cases the factorization of soft currents and collinear splitting functions from generic matrix elements,
both of which are shown to be process independent to all orders.
We show that the overlapping soft-collinear region is naturally accounted for by vacuum matrix elements of kinked Wilson lines.
Although the proof is self-contained, it combines techniques developed for the study of pinch surfaces,
scattering amplitudes, and effective field theory.

\end{abstract}

\newpage
\tableofcontents

\newpage

\section{Introduction}
Factorization is at the heart of any quantitative prediction using quantum chromodynamics (QCD). 
Probably the most familiar type of factorization, which we call {\it hard factorization}, justifies the use of fixed-order
perturbation theory for sufficiently inclusive quantities.
It lets us use perturbative calculations involving partons (quarks and gluons)
to make precise predictions for experimentally measurable quantities involving color-neutral hadrons.
The intuition for hard factorization is that 
scattering has a component which freezes in at short distances and  can only incoherently influence the long-distance components. For many observables, the long-distance physics can be integrated over with essentially unit probability.
Somewhat less intuitive, but also logical after a little thought, is the factorization of infrared-sensitive physics
into soft and collinear components. 
This {\it soft-collinear factorization} can be anticipated classically, since very-long distances modes (soft physics) can only
probe the net (color) charge of a collection of particles traveling in nearly the same direction. Conversely, energetic
collinear particles cannot have their momentum changed much by low-energy soft modes.
Although the physical picture of hard-soft-collinear factorization is simple, rigorously establishing
exactly what it implies about scattering amplitudes in gauge theories is not.

Factorization has a long history, with an eclectic variety of approaches yielding a nuanced picture of when and where factorization should hold, and in what form.
In this paper, we eschew two serious complications: 1) we ignore non-perturbative effects associated with strong-coupling, discussing only power corrections associated with the kinematics of massless partons rather than corrections of order $\LQCD/Q$ and 2) we avoid configurations where final-state particles are collinear to initial state particles. 
Even within this limited scope, although much is known, a precise formulation of factorization in terms of QCD matrix elements has been lacking. It is the goal of this paper to provide such a formulation and proof.

As we will review and rederive, the essence of factorization is revealed by studying the infrared (IR) structure of gauge theories.
An obvious necessary condition for an IR divergence is that some propagators blow up. 
Sufficient conditions
are quite a bit more complicated. First, the poles associated with on-shell momenta must be {\it pinched}, so that one cannot just integrate over them~\cite{Landau:1959fi,Coleman:1965xm}. Second, the numerator structure of integrands, which is gauge-dependent, can make an integral more or less divergent than the propagator denominators alone imply. In certain gauges, such as lightcone gauge, the possible virtual momenta contributing to the IR singularities -- the so-called {\it pinch surface} -- turns out to be remarkably simple: all virtual momenta $q^\mu$ must either be exactly proportional to one of the external momenta $q^\mu = \alpha p^\mu_\cci$ with $\alpha \ge 0$ or exactly vanish, $q^\mu =0$.  A picture of such a surface is often drawn as a reduced diagram with hard, jet and soft regions~\cite{Sterman:1978bi,Libby:1978qf,Collins:1981ta}, similar for example to Eq.~\eqref{pinches2} below.

Unfortunately, understanding the singular pinch surface, that is, the topology of exactly zero momentum or exactly collinear lines, does not immediately translate to a precise statement of hard factorization or soft-collinear factorization. Indeed, descending from the pinch surface to a statement about finite amplitudes requires a whole new set of justifications. For example, one must relate the unphysical power-counting
of a pinch surface of finite phase-space volume to the physical power-counting of external momenta. 
In particular, infrared divergences associated with the soft pinch surface (where $k^\mu=0$) depend on whether that surface is approached from
a likelike (the soft region) or spacelike (the Glauber region) direction.  
Other subtleties include avoiding double-counting in the
soft-collinear region (the {\it zero bin}),
restricting the phase space for real and virtual integrations in the soft function without reintroducing dependence on the hard scale, and introducing Wilson lines
to restore gauge invariance without spoiling the leading-power factorization. 
Despite these challenges, factorization has been proven at the amplitude and amplitude-squared
level in a number of contexts~\cite{Sen:1981sd,Sen:1982bt,Collins:1981uk}. Factorization formulas for cross-sections of certain observables 
have been presented~\cite{Mukhi:1982bk,Kidonakis:1997gm, Kidonakis:1998bk,Laenen:1998qw,Berger:2003iw,Sterman:2002qn,Aybat:2006mz,Dixon:2008gr} allowing
for resummation of large logarithms associated with the pinch surface.

In  deep-inelastic lepton-hadron scattering (DIS), the pinch surface is particularly simple. In this case, factorization has been understood since the 1970s and
has  been used to compute phenomenologically important quantities, namely the DGLAP splitting functions~\cite{Gribov:1972ri,Lipatov:1974qm,Dokshitzer:1977sg,Altarelli:1977zs}.
These splitting functions
describe the leading-power behavior of certain amplitudes when an additional collinear parton is added; they also provide kernels for the renormalization group (RG) evolution of parton distribution functions (PDFs). In DIS, the splitting functions and PDF evolution can be rigorously defined through an operator product expansion (OPE)~\cite{Georgi:1951sr,Gross:1974cs}, which has led to their computation
at 2 loops~\cite{Curci:1980uw,Furmanski:1980cm} and 3 loops~\cite{Moch:2004pa}. The OPE for DIS is possible because it involves the matrix element of two currents whose analytic structure in the complex plane
is particularly simple. That the same splitting functions apply for PDF evolution in some other process, for example the Drell-Yan process,
can occasionally be shown by direct calculation~\cite{Altarelli:1979ub}. However, to show universality of the PDFs more generally requires a general
proof of hard-collinear factorization. Subtleties associated
with proton-proton scattering, where initial state partons can be collinear to final state particles, complicate factorization~\cite{Collins:1988ig,Catani:2011st,Bauer:2000yr}. 
Needless to say, showing that the same PDFs apply to any scattering process (if indeed they do) is an extremely important open question, beyond the scope of this paper.

\enlargethispage{5pt}

An alternative, more pragmatic, approach skips both the pinch surface and the OPE
and simply computes the diagrams relevant for factorization directly, usually in dimensionally regularized perturbation theory. 
Following this approach, universality of collinear splittings was
shown at 1-loop by Bern and Chalmers in 1995~\cite{Bern:1995ix} by studying collinear limits of 5-point amplitudes in QCD. Hard-collinear factorization can be written heuristically as
\be
\cM_n \overset{p_\ccO \, \parallel \cdots \parallel\, p_{\ccm}}{\LPeq} \Sp(p_\ccO, \ldots, p_{\ccm}) \cdot \cM_{n-m} \label{collspliteq} 
\ee
with $\cM_n$ an $n$-external-particle matrix-element, $p_\ccO^\mu \cdots p_{\ccm}^\mu$ the external momenta which become collinear, and $\LPeq$ indicating
the two sides agree at leading power. The important point in this formula is
that the splitting function $\Sp(p_\ccO, \ldots, p_\ccm)$ has no dependence on any of the non-collinear momenta in the process. 
Formulas like Eq.~\eqref{collspliteq} and the explicit formulas for 
$\Sp(p_\ccO, \ldots, p_\ccm)$ in $d$ dimensions are important for precision calculations in QCD. We will give more-precise operator definitions of the objects in this equation in Section~\ref{sec:SplitFunc}. In 1999, Kosower proved \Eq{collspliteq} at leading color (large $N_c$) to all orders in perturbation theory~\cite{Kosower:1999xi}. 
The factorization of IR (soft and collinear) tree-level amplitudes to all orders was shown in~\cite{Catani:1999ss}.
Ref.~\cite{Catani:2011st} has discussed difficulties with  Eq.~\eqref{collspliteq} when initial and final states are collinear. Avoiding such situations,
we will show that
Eq.~\eqref{collspliteq} holds to all orders in QCD, at finite $N_c$. Indeed, hard-collinear factorization is a corollary of the more general hard-soft-collinear factorization formula we prove
in this paper.
 
The factorization of soft emissions from generic matrix elements is also believed to satsify a formula similar to Eq.~\eqref{collspliteq}. For example, in the limit that
a single soft gluon of momentum $q^\mu$ becomes soft, tree-level amplitudes factorize as~\cite{Bassetto:1984ik}
\be
 \cM_n \overset{q~\text{soft}}{\LPeq}  \epsilon_\mu(q)\, \JJ^\mu_\sua \cdot \cM_{n-1} \label{softcurreq} 
\ee
The soft current  $\JJ_\sua^\mu$ is an operator acting in color space. 
In 2000, Catani and Grazzini proved this formula at 1-loop, with an explicit computation of  $\JJ_\sua^\mu$,
 and conjectured that the formula holds to all orders~\cite{Catani:2000pi}. In 2013, the soft current was computed at 2-loops in ~\cite{Duhr:2013msa,Li:2013lsa}.
These calculations were all done in dimensional regularization and have applications in perturbative QCD, such as to the N${}^3$LO Higgs-boson inclusive cross-section. 
As with Eq.~\eqref{collspliteq}, our general factorization formula contains the hard-soft factorization embodied in Eq.~\eqref{softcurreq} as a special case. We
prove this equation to all orders and provide regulator-independent and gauge-invariant operator definitions of the objects involved in Section~\ref{sec:SoftCurr}.

Remarkably, a factorization theorem valid at leading power to all orders in $\alpha_s$ is not strictly required for resummation to all orders in $\alpha_s$ of certain leading or next-to-leading logarithms. For example, by combining $\cO(\alpha_s)$ collinear splitting functions,   $\cO(\alpha_s)$ soft-coherence effects, and $\cO(\alpha_s^2)$ Sudakov effects (associated with the overlapping soft-collinear region), Catani, Marchesini and Webber derived a powerful coherent-branching algorithm~\cite{Catani:1990rr}. Coherent branching is the backbone of the Monte Carlo event generator approach to QCD. 
It has also been used for resummation of many observables at the next-to-leading logarithmic level~\cite{Catani:1989ne,Catani:1990rr, Catani:1992ua, Bonciani:2003nt}. A related observation is that QCD simplifies dramatically in the limit that  gluons are strongly ordered in energy~\cite{Bassetto:1982ma,Bassetto:1984ik,Fiorani:1988by}, particularly at large $N_c$. This approximation has led to the resummation of certain leading logarithms, such as non-global ones~\cite{Dasgupta:2001sh,Banfi:2002hw} which no other method has yet tamed.

A relatively recent approach to factorization is provided by Soft-Collinear Effective Theory (SCET)~\cite{Bauer:2000ew, Bauer:2000yr,Beneke:2002ph,Beneke:2002ni}.
The idea behind SCET is to hypothesize which IR modes contribute to QCD scattering processes and to write fields in QCD as sums of fields with soft or collinear quantum numbers corresponding to the hypothesized modes. Different components are assigned different scaling behavior
and the QCD Lagrangian is expanded to leading power (or beyond). 
The resulting effective theory has Feynman rules which are significantly more complicated than
those of QCD. These rules simplify somewhat after a field redefinition which moves the soft-collinear interactions from the Lagrangian into the operators.
Proofs using the effective Lagrangian are then carried out under the assumption that the only modes necessary for the proof are those in the effective theory. Therefore, proofs of factorization in SCET must be interpreted with some care.
An advantage of the SCET approach is that with operator definitions of the various objects, the hard-soft-collinear decoupling is completely transparent
and resummation of large logarithms can be done through the renormalization group. This has lead to precise predictions of jet observables at
 colliders~\cite{Schwartz:2007ib,Becher:2008cf,Becher:2009th,Chien:2012ur,Becher:2012xr,Feige:2012vc,Jouttenus:2013hs}. Another advantage is that the power counting
makes it straightforward, in principle, to go beyond leading power if desired. On the other hand, the derivation of SCET has been done in a gauge in which
the physics is quite unintuitive, for example with polarization vectors which are longitudinally polarized at leading power (see \cite{Feige:2013zla}). SCET removes the soft-collinear
double counting by simply not summing over the zero-momentum bin in the discrete sum over labels. A somewhat simpler formulation of SCET was presented recently by
Freedman and Luke in~\cite{Freedman:2011kj} and connects more directly to the current work, as discussed in Section~\ref{sec:SCET}.

In this paper, we present and prove a factorization formula for amplitudes in gauge theories, building upon
 insights from many of the approaches discussed above. All of the interesting features of this formula can be seen in the simpler case of factorization for matrix elements of the operator $\cO = \frac{1}{(N/2)!} |\phi|^N$ in scalar QED. There, our formula reads
\be
\boxed{
\bra{X}\cO\ket{0}
 \;\LPeq\; 
\cC(\Sij) \, \frac{\bra{X_\ccO} \phi^\star W_\ccO \ket{0}}{\bra{0} Y_\ccO^\dg W_\ccO \ket{0}} \,\cdots\, \frac{\bra{X_\rN} W_\rN^\dg\phi \ket{0}}{\bra{0} W_\rN^\dg Y_\rN \ket{0}}
\;\bra{X_\scs} Y_\ccO^\dg \cdots Y_\rN \ket{0}
}
\label{sQEDmain}
\ee
This formula applies to final states $\bra{X}$ which can be partitioned into $N$ regions of phase space 
such that the total momentum  $P_\ccj^\mu$ in each region has an invariant
mass which is small compared to its energy.
More explicitly, we demand $P_\ccj^2 < \lambda^2 (P_\ccj^0)^2$, where $P_\ccj^0= E_\ccj $ is the energy of the jet, for some number $\lambda \ll 1$ which is used as a power-counting parameter. 
For such states, the momentum $q^\mu$ of any particle has to be either collinear to one of $N$ lightlike
directions, $n_\ccj^\mu$, meaning $n_\ccj \cdot q < \lambda^2 q^0$, or soft, meaning $q^0 < \lambda^2 P_\ccj^0$. 
Thus we can write for the final state
$\bra{X} = \bra{X_\ccO\cdots X_\rN;X_\scs}$, where all the particles with momentum collinear to $n_\ccj$ are contained in the {\it jet} state
 $\bra{X_\ccj}$ and the particles that are soft are in $\bra{X_\scs}$. This explains the states in Eq.~\eqref{sQEDmain}. 
 The Wilson coefficient  $ \cC(\Sij)$ is a
function {\it only} of the Lorentz-invariant combinations $S_{\cci\ccj} \equiv (P_\cci + P_\ccj)^2 \LPeq 2 P_\cci \cdot P_\ccj$ of jet momenta $P_\ccj^\mu$ in each direction; it does not depend at all on the distribution of energy within the jet or on the soft momenta and, therefore, it does not depend on $\lambda$. 
 The objects
 $Y_\ccj$ are Wilson lines going from the origin to infinity in the directions of the jets, and the $W_\ccj$ are Wilson lines in directions $t_\ccj^\mu$ only restricted not to point in a direction close to that of the corresponding jet. We give more precise definitions of the Wilson lines in Section~\ref{sec:prelim}.
 The symbol $\LPeq$ in Eq.~\eqref{sQEDmain} indicates that any IR-regulated amplitude or IR-safe observable computed with the two sides will agree at leading power in $\lambda$.
 
\enlargethispage{5pt}

 Eq.~\eqref{sQEDmain} implies hard-collinear factorization (Eq.~\eqref{collspliteq})  and hard-soft factorization (Eq.~\eqref{softcurreq})
 as special cases. For example,  if a two-body final state $\bra{X}$ is modified by adding a soft photon of momentum $q^\mu$, then one can calculate the effect of this extra emission by taking the ratio of the right-hand side of  Eq.~\eqref{sQEDmain} with and without the emission. Most of the terms drop out of the product, leaving
\be
\JJ^\mu_\sua
= \frac{ 
\bra{ \epsilon^{ { \spincol \mu}} (p); \sua} 
 Y_\ccO^{\dg} Y_\ccT \ket{0}}
{ 
\bra{0} Y_\ccO^{\dg} Y_\ccT\ket{0} }
=g_s \suTT^{\sua} \left( \frac{p_\ccT^\mu}{p_\ccT \cdot q} - \frac{p_\ccO^\mu}{p_\ccO \cdot q} \right) + \cO(g_s^3)
\ee
 We will give general operator definitions for the splitting amplitude, $\Sp(p_\ccO,\cdots p_\rN)$, and the soft current, $\JJ$,
 and discuss their universality in Section~\ref{sec:applic}
 after we present the generalization of  Eq.~\eqref{sQEDmain} to QCD in Section~\ref{sec:QCD} (see \Eq{genfactQCD}). Beyond providing an all-orders proof of \Eq{sQEDmain}, as well as an operator definition and proof of universality of $\Sp$ and $\JJ$, we hope that our general method of proof will itself be useful in future discussions of formal questions on the structure of perturbative amplitudes. We also hope that our approach to factorization, and the ensuing discussion of SCET in Section~\ref{sec:SCET}, will help bridge the gap between the traditional factorization methods in the QCD literature and those of SCET, as well as provide further insight into the formulation of SCET by Freedman and Luke in ~\cite{Freedman:2011kj}.

 Eq.~\eqref{sQEDmain} was derived at tree-level in~\cite{Feige:2013zla}, a paper we will refer to often and hereafter as \tree.
At tree-level, the Wilson coefficient and the vacuum matrix elements in the denominators of \Eq{sQEDmain}
are all 1 and the factorization formula reduces to
\begin{equation}
\bra{X} \cO \ket{0} \;\overset{\text{tree}}{\LPeq}\;  	
\bra{X_\ccO} \phi^\star W_\ccO \ket{0} \cdots \bra{X_\rN} W_\rN^\dg \phi \ket{0}
	\bra{X_\scs} Y_\ccO^\dg \cdots Y_\rN \ket{0}
\label{treelevel}
\end{equation}
in agreement with the formula from \tree.

There are two differences between Eqs.~\eqref{sQEDmain} and~\eqref{treelevel}, both of which represent important physical effects. First, the nontrivial Wilson coefficient in the all-loop formula enables  the factorized expression to reproduce hard-virtual corrections. Using  Eq.~\eqref{sQEDmain}, one can isolate the Wilson coefficient using
a trivial soft sector $\bra{X_\scs} = \bra{0}$ and collinear sectors with a single particle in each $\bra{X_\ccj} = \bra{p_\ccj}$. Then $\lambda=0$ exactly, and
\begin{equation}
 \cC(\sij)  = 
 \frac{\bra{p_\ccO \cdots p_\rN}\cO\ket{0}}
 {  
 \dfrac{\bra{p_\ccO} \phi^\star W_\ccO \ket{0}}{\bra{0} Y_\ccO^\dg W_\ccO \ket{0}} \,\cdots\, 
 \dfrac{\bra{p_\rN} W_\rN^\dg\phi \ket{0}}{\bra{0} W_\rN^\dg Y_\rN \ket{0}} \; \bra{0} Y_\ccO^\dg \cdots Y_\rN \ket{0}
 }
 \label{Cdef}
\end{equation}
This is a statement of purely-virtual factorization.
 Note that, since $\lambda=0$ exactly, this is an equality, not just a leading-power equivalence. The nontrivial content in this definition is that the right-hand side is IR finite, which we shall prove. Moreover, we shall prove that the Wilson coefficient is independent of the states $\bra{X_\ccO\cdots X_\rN;X_\scs}$, so that \Eq{Cdef} unambiguously specifies $ \cC(\sij) $ at leading power.
 
 The second difference between tree-level factorization and all-orders factorization is the denominators in Eq.~\eqref{sQEDmain}. These represent a type of zero-bin subtraction for loops. 
Recall that for external states
which are both soft and collinear, one is free to put them in $\bra{X_\scs}$ or $\bra{X_\ccj}$ --- the factorization formula holds with either choice. However, since all integrals are taken over $\mathbb{R}^{1,3}$, the soft-collinear region of loop momenta 
is included in both the soft and collinear matrix elements in the factorized formula, thus their overlap must
be removed.
 The term {\it zero~bin} stems from effective theory language, where one (formally) chops up phase space into a discrete sum over soft and collinear sectors. The zero~bin is the soft-collinear overlap sector in the sum, which  must be subtracted not to double count~\cite{Manohar:2006nz}. 
 The equivalence between the zero-bin subtraction in SCET and dividing by a matrix element of Wilson lines has been shown in~\cite{Lee:2006nr}.\footnote{Conveniently (or misleadingly) when dimensional regularization is used to control both the UV and IR divergences, the vacuum matrix elements of Wilson lines are all scaleless and identically vanish. Thus, the zero-bin subtraction is easy to miss, as it was in many early SCET papers.}

Besides the salient differences between the tree-level and all-orders factorization formulas, there is an
important conceptual subtlety: starting at 1-loop, both sides of Eq.~\eqref{sQEDmain} are IR
divergent.  Declaring two infinite quantities equivalent at leading power is not as absurd as it first sounds.
With an IR regulator it is, of course,  perfectly well defined. Conceptually, one could
interpret the
 leading power equivalence $\LPeq$ in this equation as meaning that whenever an IR-safe observable is computed by
integrating over an appropriate collection of final states $\langle X|$, the two sides of  Eq.~\eqref{sQEDmain} produce the same cross section at leading power
in $\lambda$. For example, a typical IR-safe jet observable is $\tau =\frac{1}{Q^2}(\sum_i m_i^2 + QE_\text{out})$: the sum over the jet masses and the out-of-jet energy.
Then $\frac{d \sigma}{d \tau}$ will agree when computed with either side of  Eq.~\eqref{sQEDmain} up to corrections subleading in $\tau$. 
With this in mind, one can still work at the amplitude level without an explicit IR regulator.

To be clear, we do not require or expect the IR divergences on the two sides of Eq.~\eqref{sQEDmain} to exactly agree. Indeed, 
as soon as real-virtual diagrams contribute, the IR divergences will not exactly agree. 
 To see this note that the real-emission graphs computed with Eq.~\eqref{sQEDmain}  only agree at leading power and so an IR-divergent virtual graph with a subleading real emission tacked on will show up on the left-hand side of \Eq{sQEDmain} but not on the right-hand side. This implies that the IR divergences can only precisely agree when $\lambda=0$ (no emissions), as in Eq.~\eqref{Cdef}.\footnote{One can of course add subleading-power operators to the right-hand side of Eq.~\eqref{sQEDmain} so that subleading IR divergences cancel. To get all the IR divergences to cancel, one would need an infinite number of 
operators and the 
factorized expression would be identical to the full theory.}
However, subleading-power IR-divergences will contribute at subleading power to observables, so the disagreement of subleading-power IR-singularities does not invalidate the 
leading-power equivalence in Eq.~\eqref{sQEDmain}.

Regarding the power counting, our factorization theorem will be proven at leading power in $\lambda$, a small parameter that only depends on the external momentum in the state $\bra{X}$. We do not count powers of anything except the external momentum in the matrix element under consideration. When we discuss scaling of virtual momenta near IR sensitive regions, we will talk about scaling with $\kappa$ (see Section~\ref{sec:prelim}), but only to motivate dropping certain loop amplitudes completely.
Our proof actually holds at leading power in $N+1$ separate power counting parameters, $\lambda_\ccc^\cci$ and $\lambda_\scs$, one for each collinear sector and another for the soft. It will be clear that our proof does not require $\lambda_\ccc^\cci = \lambda_\scs$, and we can therefore derive the factorization theorem (at simultaneous leading power in all small parameters) for different types of soft and collinear momentum scalings. As we discuss in Section~\ref{sec:SCET} this implies that our factorization formula unifies what are considered to be two separate effective field theories in the literature, namely SCET$_\text{I}$ and SCET$_\text{II}$.

This paper attempts to give some intuition for the factorization formula rather than simply a proof. We therefore take our time with the presentation, including many examples. Section~\ref{sec:prelim} establishes
some of our notation and reviews some basic concepts. Sections~\ref{sec:1loopEx1} and~\ref{sec:1loopEx2} give examples. Although the proof does not rely on these two example sections, the special cases considered
illustrate many of the issues which come up in the proof and are useful for making some of the abstractions more concrete. Section~\ref{sec:proofoutline} outlines the proof but has no results. The proof begins
in earnest in Section~\ref{sec:coloring}. In this section we explain how Feynman diagrams can be written as sums of colored diagrams with red lines engendering soft-sensitivity and blue lines soft-insensitive. This section would
be quite short if not for the examples we include. Section~\ref{sec:Step2} proves a set of lemmas which establish the physical-gauge reduced-diagram picture manifesting hard factorization. The difference
between our reduced diagrams and reduced diagrams in the literature (see for example~\cite{Sterman:1978bi,Libby:1978qf,Collins:1981ta}) is that our diagrams correspond to specific functions of finite-external momenta computed through loop integrals over all of $\mathbb{R}^{1,3}$, 
while the traditional reduced diagrams describe only the pinch surface where all virtual momenta are either exactly zero or exactly proportional to an external momentum. To prove soft-collinear
factorization, we introduce a special gauge we call {\it factorization gauge} in Section~\ref{sec:factorizationgauge}. The soft-collinear decoupling proof is given in Section~\ref{sec:Step4}. The rest of the paper
discusses the generalization to QCD, some special cases, the QCD splitting functions and soft currents, the connection to SCET, and a brief look forward.

\section{Preliminaries}
\label{sec:prelim}
To begin, we establish in this section some of the basic features of amplitudes we will exploit for factorization. 
We first review the importance of soft and collinear momenta. We then discuss how soft and collinear regions of virtual momenta can be separated without
chopping up the loop momenta into sectors. 

Let us begin with some terminology. 
We will distinguish {\bf soft divergences} from {\bf collinear divergences}, both of which are defined in Section~\ref{sec:SCinteg}. We refer to {\bf IR divergences} as either soft or collinear. We use $\lambda$ to power-count external momenta, as discussed in Section~\ref{sec:lambdacount}. 
We use $\kappa$ to power-count loop momenta. The notation $p \parallel q$ is used to denote when
two momenta, either real or virtual, are nearly collinear according to the appropriate power counting.
The notation $p \propto q$ is reserved for when two momenta are exactly collinear, that is, when they are {\bf proportional} to each other.  
Following \tree, the symbol $\LPeq$ indicates that two expressions agree at leading power in the limit of external particles becoming soft or collinear in an amplitude. That is, it refers 
power counting in $\lambda$, not $\kappa$.
More precisely
\be
A \LPeq B \quad \Longleftrightarrow \quad 
\frac{A}{B}  = 1+ \cO(\lambda)
\ee
We also define
\be
A ~\LPFeq B \quad \Longleftrightarrow \quad \frac{A}{B} =  \cO(\lambda^0)
\ee
This less restrictive IR-equivalence will be used in Section~\ref{sec:Step4} to avoid keeping track of modifications of the hard-amplitude along the steps of soft-collinear factorization.

We are often interested not only  in whether a loop is IR divergent, but whether it would be IR divergent if two external particles were proportional, or if an external momentum were exactly zero.
 If this happens we say the loop is {\bf IR sensitive}. An IR-sensitive loop is IR divergent
when $\lambda=0$ (though it need not be for $\lambda>0$).  IR sensitivity is discussed more in Section~\ref{sec:SCinteg} with an example given in Section~\ref{sec:irsensitive}.

\subsection{Power counting for external momenta}
\label{sec:lambdacount}
A key observation which makes factorization important is that soft and collinear momenta dominate cross sections. At tree level, this is easy to see.
Consider a process with outgoing final-state momenta $p_\cci^\mu$ of zero mass. 
 At tree level, each intermediate momentum  $k^\mu$ must be a linear combination of external momenta $p_\cci^\mu$:
  $k^\mu= p_\ccO^\mu + \cdots +  p_{\ccn}^\mu$. Thus $k^2=\sum_{\cci,\ccj} p_\cci\cdot p_\ccj$. Since each $p_\cci\cdot p_\ccj$ is positive definite,
 $k^2$ can only vanish if $p_\cci^\mu$ is exactly proportional to $p_\ccj^\mu$ for each $\cci$ and $\ccj$ in the sum, or if a $p_\cci^\mu$ has zero energy. The dominant regions of phase space where the propagators are large are, therefore, the regions where momenta
are collinear: $p_\cci \parallel p_\ccj$, or soft: $E_i \ll Q$, with $Q$ the center-of-mass energy. This is discussed extensively in \tree.

We, therefore, focus on final states $\bra{X}$ partitioned into collinear sectors $\bra{X_\ccO} \cdots \bra{X_\ccN}$ and a single soft sector
 $\bra{X_\scs}$. Let $m_\cci$ and $E_\cci$ be the invariant mass and energy of the net momentum $P_\cci^\mu=\sum_{\text{sector}~\cci} p_\ccj^\mu$ in each sector,
and define $\lambda_\cci = m_\cci/E_\cci$ for the collinear sectors and $\lambda_s = E_s/Q$ for the soft sector. We assume $\lambda_\cci \ll 1$ for every sector,
so that the contribution of the state  $\bra{X}=\bra{X_\ccO} \cdots \bra{X_\ccN}\bra{X_\scs}$ to a cross section will scale like inverse powers of
all $\lambda_\cci$. It is for these states that hard-soft-collinear factorization holds.

\subsection{Power counting for virtual momenta}
\label{sec:SCinteg}

The soft and collinear regions of phase space are also important because they lead to IR divergences in loops. 
IR divergences come from virtual-particle momenta going on-shell. 
Let us call {\bf loop momenta} those being integrated over. That is, denoting the loop momenta as $k_i^\mu$,
the loop measure is $\prod_i \dfrac{d^4 k_i}{(2\pi)^4}$. Any {\bf virtual momentum} $l^\mu$ in a Feynman diagram is a linear combination of loop momenta and external momenta: $l^\mu(k_i, p_\cci)$. 
Thus, for a virtual propagator to blow up, the virtual momentum must go on-shell, which makes the loop momentum either soft or collinear to one of the jet
directions. 
Since we associate infrared divergences with virtual lines, it is convenient to route the momenta so that the virtual momentum in question is one of the loop momenta, $k^\mu$.
We say a given diagram has a {\bf soft divergence}  associated with $k^\mu$ if it is still divergent when
each component of $k^\mu$ is restricted to be smaller than some arbitrarily small scale, $\kappa^2 Q$, for any $\kappa>0$. A {\bf collinear divergence} requires the specification of a finite, non-zero lightlike momentum, $p^\mu$; the singularity is then present in any integration region containing $p^\mu$. We take {\bf infrared divergence} to mean either soft or collinear.

A shortcut to determining whether a given integral is IR divergent is through its scaling behavior, which can
be understood in lightcone coordinates.
%
 Given two distinct lightlike directions $n_a^\mu$ and $n_b^\mu$, we can uniquely decompose any 4-vector $k^\mu$ as
\be
k^\mu = k_ b \, n_a^\mu +k_a\, n_b^\mu + k^\mu_\perp
\label{lcnanb}
\ee
with $k^\mu_\perp$ defined by this equation and
\be
k_a = \frac{n_a \cdot k}{n_a \cdot n_b}, 
\qquad 
k_b = \frac{n_b \cdot k}{n_a \cdot n_b}
\ee
We can then consider rescaling the components by factors of $0<\kappa<1$ raised to various powers 
\be
k^\mu \to \kappa^b k_ b \, n_a^\mu + \kappa^a k_a\, n_b^\mu + \kappa^c k^\mu_\perp
\quad\text{with}\quad
 a,b\geq0,\;a+b>0,\; c>0
\label{resc}
\ee
We require $a,b\geq0$, $c>0$ and $a+b>0$, so that as $\kappa \to 0$ these rescalings zoom in on a possibly singular region.
For example, $a,b,c>0$ scales $k^\mu \to 0$ (the soft region), whereas $b=0$ and $a,c>0$ scales $k^\mu \to k_b \,n_a^\mu$ (the $a$-collinear region).
We say an integral is {\bf power-counting finite} if, including the measure, it scales like $\kappa$ to a positive power under a given rescaling of this form.

The purpose of these rescalings is that they are related to whether or not a diagram is infrared divergent:
\begin{conjecture}\ltag{Power-Counting Finiteness Conjecture}
\label{PCfinite}
A Feynman integral is infrared finite 
if and only if
it scales as a positive power of $\kappa$ under all  possible rescalings in Eq. \eqref{resc}.
\end{conjecture}

That an infrared-finite Feynman integral scales as a positive power of $\kappa$ for any rescaling is easy to prove: a convergent integral must have a convergent Riemann sum. The
converse, that scaling implies infrared finiteness, is also quite logical. We are certainly not aware of any counterexamples. Nor do we know of a rigorous proof. This conjecture is assumed to hold in practically every factorization proof, and we assume it too. For a discussion of a slightly stronger version of this conjecture, see  page 428 of \cite{Sterman:1994ce}.

A convenient simplification is that it is not necessary to consider all possible values of $a,b,c\ge0$. In determining the leading power of $\kappa$ with a given scaling, all that matters is which terms can be dropped with respect to which other terms -- any scaling that drops the same terms gives the same integrand with the same singularities. 
Between two power-counting regions that allow two different terms to be dropped lies a boundary where both terms must be kept.
Because more terms must be kept on the boundary, if a boundary region 
is power-counting finite then the regions it bounds 
must also be power-counting finite. This simplifies the types of power-counting we need to consider.

\begin{figure}[t]
	\centering
  	\includegraphics[width=.6\textwidth]{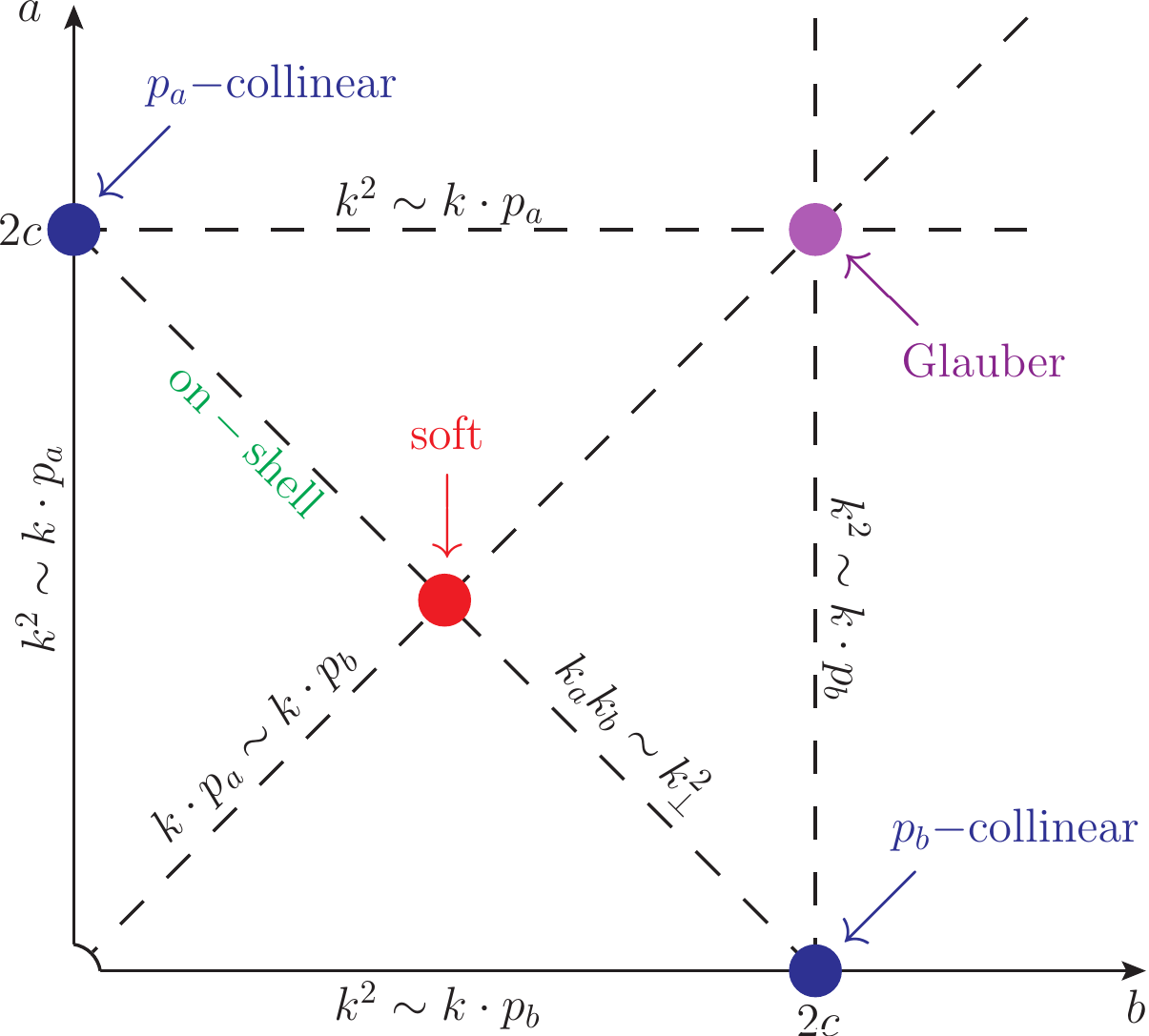}
	\caption{Scalings, $k \sim (\kappa^b, \kappa^a,\kappa^c)$, that could give power-counting IR divergences.}
\label{fig:boundaries}
\end{figure}

 In a given Feynman loop diagram, we always have one propagator whose denominator is $k^2$ (by our choice of momentum routing).
Under the rescaling in Eq. \eqref{resc}, 
\be
k^2 = 2n_a\cdot n_bk_a k_b + k_\perp^2 \to \kappa^{a+b} \,2n_a\cdot n_bk_a k_b + \kappa^{2c}\, k_\perp^2
\ee
So, if $a+b>2c$, we may drop $k_a k_b$ in place of $k_\perp^2$, and if
$a+ b < 2c$, $k_\perp^2$ can be dropped with respect to $k_a k_b$. We might also have denominators $(k-p_a)^2$ for some $p_a^\mu$. If $p_a^\mu$ is not lightlike, then $(k-p_a)^2 \sim p_a^2 \sim \kappa^0$.
A more relevant case is when $p_a^\mu$ is lightlike. Then it makes sense to choose one of our basis vectors $n_a^\mu$ to point along $p_a^\mu$. In this case, a term $k\cdot p_a \to \kappa^a\, k\cdot p_a$ may appear in a denominator. Similarly, $k \cdot p_b \to \kappa^b\, k\cdot p_b$ may appear.
Thus there are four relevant scaling behaviors:
\be
k_a k_b \sim \kappa^{a+b}, \quad  k_\perp^2 \sim \kappa^{2c}, \quad 
k\cdot p_a \sim \kappa^a \quad k\cdot p_b \sim \kappa^b
\label{allscalings}
\ee
In expanding for small $\kappa$, all we do is drop some of these when they are smaller than others.  If an integral is power-counting finite when two terms are of comparable size, it is necessarily power-counting finite when one of them is dropped. So we can restrict our considerations to scalings where two (or more) of these terms are comparable.

There are six regions where two of the scalings in \Eq{allscalings} are equal. These form the lines  in Figure~\ref{fig:boundaries}. 
For example, one of the diagonal lines has  $a+b = 2c$ so that $k_a k_b \sim k_\perp^2$ and  $k^2 \to \kappa^{2c}\, k^2$.  This scaling is special as it keeps on-shell momenta on-shell. In particular, this line shows the only relevant scalings for external momenta. 
The scalings where two lines intersect are the four solid dots. If an integral is infrared finite at all of these points, it is automatically infrared finite under any scaling. The points in the corners come from
three scalings being equal and the center point, at $a=b=c$ has $k \cdot p_a \sim k \cdot p_b$ and $k_a k_b \sim k_\perp^2$. The most overlapping region, where all four scalings are equal requires $a=b=c=0$.
This is hard scaling which does not tell us about infrared divergences since it does not zoom in on a possibly singular region. The point at the origin in Figure~\ref{fig:boundaries}, where $a=b=0$
but $c\ne0$ also cannot produce infrared divergences since for $\kappa=0$, $k^\mu$ is offshell. We are also free to choose one of $a,b,c$ arbitrarily if it is not zero; for
example, we can set $c=1$ by replacing $\kappa$ by $\kappa' = \kappa^{1/c}$.

Thus, we can restrict the discussion to the scalings listed in Table~\ref{tab:scalings}.
\begin{table}[t]
\begin{tabular}{|lclc|}
\hline
Exponents & Conditions & Momenta scaling & Name \\[1mm]
\hline
$(a,b,c) = (0,0,0)$:& $k_a k_b \sim k_\perp^2 \sim k\cdot p_a\sim k\cdot p_b $ & $\quad k^\mu \sim  (1,1,1) \qquad$  &hard \\
$(a,b,c) = (2,0,1)$:& $k_a k_b \sim k_\perp^2 \sim k\cdot p_a $ & $\quad k^\mu \sim  (1,\kappa^2,\kappa) \qquad$  &$p_a$-collinear \\
$(a,b,c) = (0,2,1)$:& $k_a k_b \sim k_\perp^2 \sim k \cdot p_b$ & $\quad k^\mu \sim (\kappa^2,1,\kappa) \qquad$  &$p_b$-collinear \\
$(a,b,c) = (2,2,2)$:& $(k_a k_b \sim k_\perp^2) ~\&~ (k\cdot p_a \sim k \cdot p_b)$ &$\quad k^\mu \sim (\kappa^2,\kappa^2,\kappa^2) \qquad$  &soft \\
$(a,b,c) = (2,2,1)$:& $k_\perp^2 \sim k\cdot p_a \sim k \cdot p_b$ &$\quad k^\mu \sim (\kappa^2,\kappa^2,\kappa) \qquad$ &Glauber \\
\hline
\end{tabular}
\caption{Scalings relevant for factorization.
\label{tab:scalings}}
\end{table}
Of these, hard scaling does not produce infrared divergences. Soft and collinear scaling both imply $k^2 \to \kappa^2 k^2$. In particular, 
timelike, spacelike and lightlike momenta stay timelike, spacelike and lightlike, respectively.
Glauber scaling, on the other hand, turns timelike and lightlike momenta into spacelike momenta as $\kappa\to 0$, preserving only the spacelike nature.

The set of scalings we need to consider is even smaller for the processes that have no collinear directions in the initial state. When there are only final state particles, for example in a decay, we know the infrared divergences must cancel among real and virtual corrections
at each order in $\alpha_s$. The reason infrared finiteness can be proven in this case is because, by unitarity, a decay is the imaginary part of a $1\to1$ total cross section
whose analytic structure is particularly simple. Not only does infrared finiteness hold, but there is a one-to-one correspondence between the momenta producing
infrared divergences in real emission contributions and the virtual contributions. This is easiest to see using old-fashioned perturbation theory (see Chapter 13 of \cite{Sterman:1994ce}).
In a real emission graph with only final state particles, all the virtual lines without loop momenta flowing through them are timelike. As we take $\lambda \to 0$ these timelike momenta approach the lightcone from within, and give rise to soft and collinear real-emission phase-space singularities. Because these phase-space divergences come from timelike momenta becoming lightlike, there cannot be any phase-space singularities with Glauber scaling, which as $\kappa \to 0$ makes timelike momenta spacelike. Then, by infrared finiteness of the total decay rate, there cannot be Glauber singularities in loop integrals either.
We conclude that, when considering only final-state collinear directions, only soft and collinear scalings can possibly produce infrared divergences. 

When there are collinear particles in the initial state, we expect that unitarity-based arguments should still hold, even if they have not yet been rigorously proven. 
The complication is that with collinear particles in the initial state, the virtual momenta in real-emission graphs can be spacelike. In particular, a virtual
particle with momentum $k = p^\mu - p'^\mu$ connecting an initial state particle of momenta $p^\mu$ to a final state particle of momentum $p'^\mu$ can be spacelike
 and have Glauber scaling if $p^\mu$ is collinear to $p'^\mu$. 
Thus Glauber scaling is important for forward scattering. In this paper, we will only have final state collinear directions, so we can ignore Glauber scaling.
A technical pinch-analysis proof of the irrelevance of Glauber scaling for decay processes can be found in Chapter 5 of~\cite{Collins:2011zzd}.

We conclude that we only need to consider soft scaling, and collinear scaling in each relevant direction.
If upon $k^\mu \to \kappa^2 k^\mu$, an integral
scales like $\kappa$ to a positive power, the integral is not soft divergent. If it scales like $\kappa^0$ (it cannot scale like $\kappa$ to a negative power, see \cite{Sterman:1978bi} or Lemma~\ref{lem:kappa0}), there might
be a soft divergence.
Collinear divergences are determined by rescaling $k^\mu$ as
\be
k^\mu \to \frac{n_b \cdot k}{n_a \cdot n_b}\, n_a^\mu + \kappa^2 \frac{n_b \cdot k}{n_a \cdot n_b}\, n_b^\mu + \kappa \,k^\mu_\perp
\label{kcollp1rescale}
\ee
If the integral scales like $\kappa$ to a non-positive power, there is a potential collinear divergence. Otherwise, the integral is collinear finite in the $n_a^\mu$
direction.

In practice, Eq.~\eqref{kcollp1rescale} implies that to find a collinear divergence associated with the direction $p^\mu$ of an external momentum,
we rescale
\begin{align}
d^4k ~~&\to~~ \kappa^4\, d^4k \notag\\
k^2 ~~&\to ~~ \kappa^2 \, k^2\\
k\cdot p ~~&\to ~~ \kappa^2 \, k\cdot p \notag
\end{align}
If $q$ is another loop momenta, then the scaling depends on whether $q$ is being consider
collinear to $p$ or not:
\begin{equation}
k\cdot q ~ \to 
k\cdot q
\times
  \begin{cases}
\kappa^2 , \quad q \parallel p\\[2mm]
1,\quad q \notparallel p
\end{cases}
\end{equation}
For collinear-sensitive power counting (see below), the same scaling rules apply (depending on whether $q\parallel p$ or not) if $q$ is a sum of external momenta.

As an example, consider the 1-loop scalar integral:
\begin{equation}
I = \fint{k} \frac{1}{\big(k^2\pie\big) \big((p_\ccO+k)^2\pie\big) \big((p_\ccT+k)^2\pie\big)}
\end{equation}
with $p_\ccO^2 = p_\ccT^2 = 0$. In the soft limit, 
\begin{align}
 \frac{d^4k}{k^2(p_\ccO+k)^2(p_\ccT+k)^2} &\overset{k~\text{soft}
 }\too
 \frac{\kappa^8 d^4k}{\kappa^4 k^2(\kappa^2 2p_\ccO\cdot k + \kappa^4 k^2)(\kappa^2 2p_\ccT\cdot k + \kappa^4 k^2)}
\nn
\\&
\hspace{15mm}=  \frac{d^4k}{k^2 (2p_\ccO\cdot k)(2p_\ccT\cdot k)} \kappa^0 + \cO(\kappa^2)
\label{softintegrand}
\end{align}
Thus there is a potential logarithmic soft divergence in this integral. In the limit where $k \parallel p_\ccO$, we choose $n_a^\mu=p_\ccO^\mu$.
Then
\begin{align}
 \frac{d^4k}{k^2(p_\ccO+k)^2(p_\ccT+k)^2} &\overset{k\, \parallel\, p_\ccO}
 \too
 \frac{\kappa^4 d^4k}{\kappa^2 k^2(\kappa^2 2p_\ccO\cdot k + \kappa^2 k^2)(2p_\ccT\cdot k + \kappa^2k^2)}
\nn
\\&
\hspace{15mm}= \frac{d^4k}{k^2(p_\ccO+k)^2 \,2p_\ccT\cdot k} \kappa^0 + \cO(\kappa^2)
\label{colintegrand}
\end{align}
Thus, there is a potential collinear divergence in the $p_\ccO^\mu$ direction. By the symmetry of the integral, there is a potential collinear divergence in the $p_\ccT^\mu$ direction as well.

In some cases, an integral does not have a divergence associated with a specific power counting despite the integrand scaling like $\kappa^0$ (for example, the Glauber scaling in decay processes).
Indeed, one can often deform the integration
contour away from the singularity. If this deformation cannot be done, the singularity is said to be pinched. 
While there is a close connection between our approach and the results of a pinch analysis, we can conveniently avoid the discussion of contour deformation all together.
Although we will use strongly that some diagrams with on-shell internal lines are not soft sensitive, we will not directly use the Landau equations~\cite{Landau:1959fi}
or their interpretation by Coleman and Norton~\cite{Coleman:1965xm} in our proof.
Instead, we will show that two
expressions agree at leading power in $\lambda$, including both infrared divergent and infrared finite contributions.
The connection between infared divergences
and the leading power in $\lambda$ is through the notion of infrared sensitivity which we discuss next.

\subsection{Infrared sensitivity}

We are often interested not in actually divergent integrals, but in integrals which would be divergent if $\lambda=0$. That is, they would scale like $\kappa$ to a non-positive power if two external collinear particles were exactly proportional, or if a soft external particle had exactly zero momenta. We generalize the concept of an IR divergence to encompass such situations by saying that a loop is {\bf IR sensitive} if it is IR divergent when $\lambda=0$. Of course,  a loop that is IR divergent 
(for any $\lambda$) is also IR sensitive. 
For a loop to be infinite at $\lambda=0$ but finite for $\lambda \ne 0$, we know $\lambda$ must be acting like an IR regulator.  For example,
\be
\int_0^1 d \kappa \frac{1}{\kappa + \lambda} = \ln \frac{\lambda+1}{\lambda} \LPeq -\ln\lambda
\ee
The equivalent in a real diagram with $p \parallel q$ might be $\ln \lambda = \ln\frac{(p+q)^2}{Q^2}$.

When computing probabilities of IR-safe physical observables we square the amplitude and integrate over phase space of the external particles. The integration over phase space encloses the region where $\lambda=0$; in fact, it is this region that cancels the IR divergences in virtual loops. Thus, to preserve IR finiteness of physical observables, we must treat loops that are IR divergent when $\lambda=0$ the same as we do loops that are IR divergent for any $\lambda$. Therefore, IR sensitivity is the appropriate concept to use when discussing loops and emissions together, rather than IR divergence.
 
When power counting IR-sensitive loops, instead of setting $\lambda = 0$ and counting powers of $\kappa$, we can simply count powers of $\kappa$ and $\lambda$ together. 
By power counting  $\lambda$ and $\kappa$ as of the same order, we ensure that all the terms are kept that are necessary for the
cancellation of IR divergences between real and virtual particles at leading power of a physical IR-safe observable.

For the power counting, we only count \emph{powers}. 
This means that we treat $\ln\lambda$ as being the same order as $\lambda^0$.
Therefore, a logarithmically divergent integral can be of the same order as a finite integral. Examples are given in Section~\ref{sec:irsensitive}, where we see that
we must treat
\be
\frac{1}{\lambda} \;\LPeq\; \frac{\ln\lambda}{\lambda} \;\LPeq\; \frac{\text{log divergent}}{\lambda}
\ee
The point is that power-suppression really requires an extra power of $\lambda$. This is consistent with the leading power of an IR-safe cumulant reproducing both the constant
term and the terms which are powers of logarithms:
\be
R(\alpha_s, \lambda) = f(\alpha_s)  + f_1(\alpha_s) \ln \lambda + f_2(\alpha_s) \ln^2 \lambda + \cdots
\ee
In a perturbative fixed-order or resummed calculation, certain terms in this expansion are reproduced, but the leading power factorization formula is {\it capable} of reproducing every term in such an expansion.

\subsection{Lightcone gauge}
\label{sec:lightconegauge}

Traditionally, lightcone gauge has been particularly useful for studying soft-collinear factorization. 
In lightcone gauge, the gluon Feynman propagator is
\be
D_F^{ab;\mu\nu}(k)
= \delta^{ab} \,
\frac{  i\Pi^{\mu\nu}(k) }{k^2+i \varepsilon}
\ee
with
\be
\Pi^{\mu\nu}(k) = -g^{\mu\nu} + \frac{r^\mu k^\nu + r^\nu k^\mu}{r\cdot k}
\label{Pimndef}
\ee
where $r^\mu$ is lightlike and its overall scale does not matter. The propagator numerator, $\Pi^{\mu\nu}(k)$, satisfies
\be
r_\mu \Pi^{\mu\nu}(k)  = 0
\ee
and
\be
k_\mu \Pi^{\mu\nu}(k) =\frac{k^2}{r\cdot k} r^\nu \label{numsupp1}
\ee
which vanishes as $k^2 \to 0$.

Eq.~\eqref{numsupp1} produces a crucial feature of lightcone gauge: if $k \propto p$ where $p^\mu$ is some lightlike direction, then $p_\mu \Pi^{\mu\nu}(k)=0$. In particular, near a collinear singularity, a numerator $p\cdot \Pi(k)$ gives a suppression factor of $\kappa$. 
To be more explicit, we will often find numerator structures from virtual gluons
of the form $p\cdot \Pi(k) \cdot q$ for some momenta $p$ and $q$. To study the limit when $k \parallel p$, 
we use Eq.~\eqref{kcollp1rescale} with $n^\mu = p^\mu$ and $r^\mu$ generic. Then
\begin{align}
p\cdot \Pi(k)\cdot q &\;=\; - p\cdot q + \frac{r\cdot p\, k\cdot q + r\cdot q\, k\cdot p}{r\cdot k}
	\notag\\
&\,\to\, - p\cdot q + \frac{ p\cdot q\, r\cdot k + \kappa^2 p\cdot k\, r\cdot q
 + \kappa r\cdot p\, k_\perp \cdot q +\kappa^2 r\cdot q \, k\cdot p}{r\cdot k}
 	\notag\\
&\;=\; \kappa \left[ \frac{r\cdot p \, k_\perp \cdot q}{{r\cdot k}}+ 2\kappa \frac{ p\cdot k \, r\cdot q}{r\cdot k}\right]
\label{numsupp2}
\end{align}
This extra factor of $\kappa$ strongly restricts the type of diagrams which are collinear sensitive in lightcone gauge; it makes many graphs finite (or collinear insensitive) which would be divergent if the numerator structure scaled like $\kappa^0$.

Lightcone gauges are sometimes called {\it physical gauges}, as the ghosts decouple and the propagator numerator is a sum over physical polarizations when the gluon goes on-shell:
\be
\Pi^{\mu\nu}(k) = -g^{\mu\nu} + \frac{r^\mu k^\nu + r^\nu k^\mu}{r\cdot k}
\quad\overset{k^2 = 0}\longrightarrow\quad
\sum_{h=\pm} \,\epsilon_h^\mu(k;r) \; \epsilon_h^{\nu\,*}(k;r)
\label{PiPolSum}
\ee
Recall that the basis of gluon polarizations $\epsilon_\pm^\mu(k;r)$ is uniquely specified by a reference vector $r^\mu$ to which the polarizations are orthogonal,
and that the polarizations satisfy $r_\mu\epsilon_\pm^\mu(k;r) = k_\mu \epsilon_\pm^\mu(k;r)=0$.  The factor of $\kappa$ coming from the numerator of the lightcone gauge propagator in Eq.~\eqref{numsupp2} is similar to the extra factor of $\lambda$ suppression of collinear-emission diagrams in generic-$r$ compared to say, their scalar field theory counterparts \tree. That is, $p\cdot \Pi(k) \sim \kappa$ when $k\parallel p$ can be thought of, via \Eq{PiPolSum}, as a consequence of the transversality of the polarization vectors, which implies that $p\cdot\epsilon(q) \sim \lambda$ when $p\parallel q$.

In \tree, the freedom to choose reference vectors for the gluon polarizations
was used extensively to prove factorization at tree level. There, it was shown that two important choices of $r$ were
\be
 \text{ \bf generic-}r: \qquad r \;\cn{\parallel}\; p_\ccj \quad  \text{for any}\; \ccj
\ee
and 
\be
\hspace{2.5mm}
\text{ \bf collinear-}r: \qquad r \parallel p_\ccj \quad \text{for some}~
\ccj
\ee
For example, choosing collinear-$r$ for the polarizations of the soft gluons and generic-$r$ for the polarizations of the collinear gluons simplified the disentangling of soft and collinear radiation.

For loops, we can of course choose $r$ generic (not parallel to any $p_\ccj$), which we call a {\bf generic-lightcone gauge}, or we can choose $r \parallel p_\ccj$ for
some $p_\ccj$, which we call {\bf collinear-lightcone gauge}.  To prove factorization at loop level, however, it will be helpful to be able to choose lightcone gauges for the soft-virtual gluons and collinear-virtual gluons separately. We introduce a gauge called {\bf factorization gauge} in Section~\ref{sec:factorizationgauge} which provides this flexibility. We will refer
to either lightcone gauge with generic choice of $r$ or factorization gauge with generic choice of $r_\ccc$ as {\bf physical gauges}. This is not quite a standard usage since 1) all lightcone gauges
are usually considered physical and 2) ghosts do not completely decouple
in factorization gauge (see Section~\ref{sec:ghosts}). Since our definition is morally equivalent to the usual definition, we do not feel a new term is needed.

\subsection{Wilson Lines}
\label{sec:WilsonLines}

Wilson lines describe the radiation produced by a charged particle moving along a given path in the semi-classical limit. The semi-classical limit applies when the back reaction of the radiation on the particle can be neglected, so that the particle behaves like a source of charge. In particular, this limit holds when the particle is much more energetic than any of the radiation, that is, when the radiation is soft. The physical picture of how Wilson lines arise in the soft and collinear limits of Yang-Mills theories is discussed in \tree.

We define a soft Wilson line in the $n_\ccj^\mu$ by
\begin{equation}
  Y_\ccj^\dag(x) = P \left\{ \exp \left[ i g 
  \int_0^{\infty} d s \,n_\ccj\cdot A ( x^{\nu} + s n_\ccj^{\nu} )\, e^{-
  \varepsilon s} \right] \right\} 
\label{Ydef}
\end{equation}
where $P$ denotes path-ordering and $A_\mu = A_\mu^\sua \suTT^\sua$ is the gauge field in the fundamental representation (Wilson lines in other representations are a straightforward generalization).
This Wilson line is outgoing because the position where the gauge field $A_\mu(x)$ is evaluated goes from
$x$ to $\infty$ along the $n_\ccj^\mu$ direction. We write $Y_\ccj^\dag$ for Wilson lines for outgoing particles,
and $Y_\ccj$ for outgoing antiparticles (as $\bar{\psi}$ creates outgoing quarks and $\psi$ creates outgoing antiquarks). 
Explicitly,
\begin{equation}
  Y_\ccj(x) = \overline{P} \left\{ \exp \left[- i g
  \int_0^{\infty} d s\, n_\ccj\cdot A ( x^{\nu} + s n_\ccj^{\nu} ) \, e^{-
  \varepsilon s} \right] \right\}
\end{equation}
where $\overline{P}$ denotes anti-path ordering. We will not bother to discuss incoming Wilson lines in this paper; they are defined in \tree.

Wilson lines can be in any representation. For example, an adjoint Wilson line can be written as
\be
\cY_\ccj^\dg(x) = P\bigg\{ \exp\bigg[ ig \int_0^\infty ds\, n_\ccj\cdot A_\mu^a (x + s\, n_\ccj) \suTT^\sua_\adj \, e^{-\epsilon s} \bigg] \bigg\}
\label{Yadjdef}
\ee
where $(\suTT^\sua_\adj)^{\sub \suc} = if^{\sub \sua \suc}$ are the adjoint-representation group generators. 
Since
\be
(\suTT_\adj^\suc)^{\sua\sub}\suTT^\sub = [\suTT^\sua,\suTT^\suc] 
\, ,
\ee
fundamental and adjoint Wilson lines are related as
\be
 Y_\ccj^\dg \suTT^\sua \, Y_\ccj = \,\cY_\ccj^{\sua \sub}\, \suTT^\sub
\label{YYrel}
\ee
This identity is occasionally useful to write all of the Wilson lines for QCD in terms of fundamental and antifundamental Wilson lines.

From a practical perspective, the most important facts about Wilson lines for this paper are their Feynman rules and their gauge-transformation properties. Their Feynman rules are exactly the eikonal rules, coming from the soft limit of a QCD interaction:
\be
\fd{3.5cm}{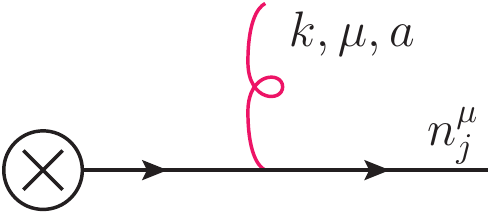} \;\;\overset{k\to \text{soft}}= \; -g
\suTT^\sua\, \frac{n_\ccj^\mu}{n_\ccj\cdot k + i\varepsilon}
\;=\; \bra{k,\mu;\sua} Y_\ccj^\dg \ket{0}
\ee
with the correct $i \epsilon$ prescription. Here $\bra{k,\mu;\sua} Y_\ccj^\dg \ket{0}$ means the off-shell matrix element for a gluon with polarization  $\epsilon^\mu(k)$ and color $\sua$ with the polarization vector
stripped off.  
That $Y_n^\dg$ gives the eikonal Feynman rules persist at any order \tree.
The $e^{\pm \varepsilon s}$ factors in the Wilson lines are required to produce the correct $i \varepsilon$ prescription for the Feynman rules (see \tree).

We denote collinear Wilson lines as $W_\ccj^\dg$. They are mathematically identical to soft Wilson lines but the path is different. While soft Wilson lines point in the direction of the particle they represent, collinear Wilson lines point in some other direction $t_\ccj^\mu$:
\begin{equation}
  W_\ccj^\dag(x) = P \left\{ \exp \left[ i g 
  \int_0^{\infty} d s \,t_\ccj\cdot A ( x^{\nu} + s t_\ccj^{\nu} )\, e^{-
  \varepsilon s} \right] \right\} 
\label{Wdef}
\end{equation}
We always take $t_\ccj^\mu$ to \emph{not} be collinear to $n_\ccj^\mu$, that is,  $t_\ccj \;\cn\parallel\; n_\ccj$. 
As discussed in \tree~and as we will see here, while soft Wilson lines account for the soft radiation of a particle, collinear Wilson lines account for the collinear radiation from all the {\it other} particles.

\section{Example 1:  one-loop Wilson coefficient}
\label{sec:1loopEx1}

The general proof of factorization will be presented starting in Section~\ref{sec:proofoutline}. To understand this proof, we first provide two examples.
 For the first example,  in this section we discuss factorization for $\bra{p_\ccO,p_\ccT} \phi^\star  \phi \ket{0}$  
at 1-loop order.  This  is perhaps the simplest 1-loop amplitude for which factorization holds. What we will show here at 1-loop order is that
\be
\bra{p_\ccO, p_\ccT} \phi^\star  \phi \ket{0} 
=  \cC(s_{\ccO\ccT})\, 
\frac{\bra{p_\ccO} \phi^\star  W_\ccO \ket{0}}{ \bra{0}  Y_\ccO^\dg W_\ccO\ket{0}}\,
\frac{\bra{p_\ccT} W_\ccT^\dg \phi \ket{0}}{ \bra{0} W_\ccT^\dg Y_\ccT \ket{0}}\, \bra{0} Y_\ccO^\dg Y_\ccT \ket{0}
\label{factEx1}
\ee
where $s_{\ccO\ccT} = (p_\ccO+p_\ccT)^2$. Note that \Eq{factEx1} is an exact equality, not a leading power equivalence, because there are no particles collinear to each other and no soft particles, so $\lambda=0$. It is also somewhat trivial: it is just
a definition of $\cC(s_{\ccO\ccT})$. The nontrivial part is showing that  $\cC(s_{\ccO\ccT})$ is IR finite. The next example,  in 
Section~\ref{sec:1loopEx2}, discusses what happens when one of the sectors has two collinear particles and provides a nontrivial check on the universality of $\cC(s_{\ccO\ccT})$.

\subsection{Overview of graphs}
There are five graphs contributing to the left-hand side of Eq.~\eqref{factEx1} at 1-loop order. Four of them involve only one leg
\be
G^\lr{1}{1}_a = \fd{1.6cm}{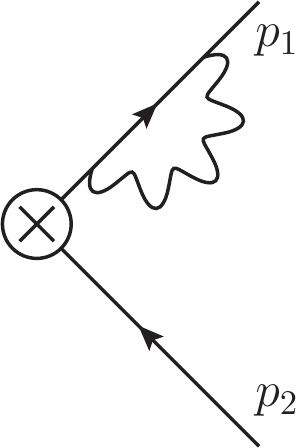} \;,\quad\;
G^\lr{1}{1}_b = \fd{1.6cm}{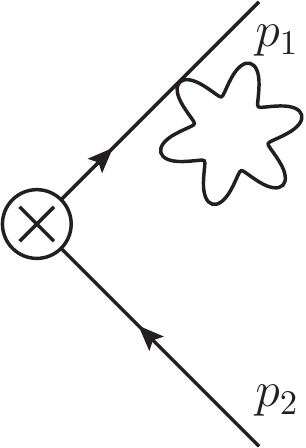} \;,\quad\;
G^\lr{2}{2}_a = \fd{1.6cm}{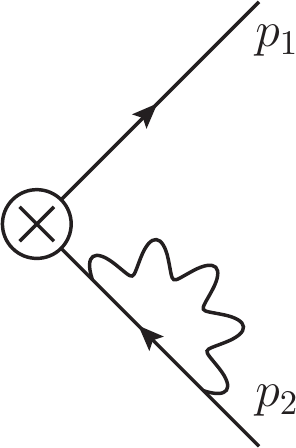} \;,\quad\;
G^\lr{2}{2}_b = \fd{1.6cm}{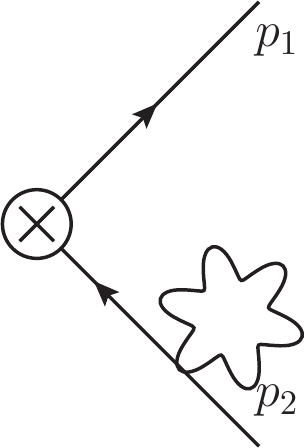}
\label{otherfourEx1}
\ee
and the final diagram connects both legs.
\be
G^\lr{1}{2}=\fd{1.6cm}{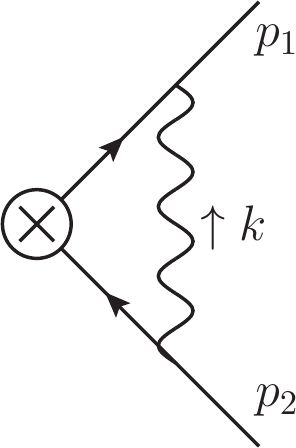}
\label{G12Glaub}
\ee

For the right-hand side of Eq.~\eqref{factEx1}, there are a number of graphs involving emissions from the collinear Wilson lines $W_i$. Recall from Eq.~\eqref{Wdef}
that the Wilson lines are defined with a certain direction $t_\cci^\mu$. 
For simplicity, let us choose $t_\ccO = t_\ccT = r$ to be some random direction not collinear to either $p_\ccO$ or $p_\ccT$. Then, if we work in a generic-lightcone gauge with the same reference vector, $r^\mu$, all of the graphs involving $W_i$
precisely vanish. The remaining non-vanishing diagrams are
\be
\langle p_\ccO | \phi^\star |0\rangle \;=\;  \fd{1.6cm}{SQEDcoll1pinch1.pdf}
\;+\;\;
\fd{1.6cm}{SQEDcoll1pinch2.pdf} \;=\; G^\lr{1}{1}_a + G^\lr{1}{1}_b 
\label{g221}
\ee
and
\be
\langle p_\ccT | \phi |0\rangle \;=\; \fd{1.6cm}{SQEDcoll2pinch1.pdf}
\;+\;\;
\fd{1.6cm}{SQEDcoll2pinch2.pdf} \;=\; G^\lr{2}{2}_a + G^\lr{2}{2}_b
\label{g222}
\ee
and those involving soft Wilson lines $Y_\cci$. 
The diagrams in Eqs.~\eqref{g221} and~\eqref{g222}  precisely agree with those in Eq.~\eqref{otherfourEx1}. 
Let us denote the diagrams coming from soft Wilson lines with the subscript {\it $\soft$}
So the remaining terms are
\begin{align}
\cC(s_{\ccO\ccT})\,\frac{ \bra{0} Y_\ccO^\dg Y_\ccT \ket{0}} { \bra{0} Y_\ccO^\dg \ket{0}\bra{0 } Y_\ccT  \ket{0}}
 &= \cC(s_{\ccO\ccT}) \,\frac{1 + G^\lr{1}{1}_\soft + G^\lr{2}{2}_\soft + G^\lr{1}{2}_\soft  }
{ (1+ G^\lr{1}{1}_\soft)(1 + G^\lr{2}{2}_\soft)  } + \cO(\alpha^2) \nn\\
&=\cC(s_{\ccO\ccT})\,[ 1 + G^\lr{1}{2}_{\soft} + \cO(\alpha^2)]
\label{Gsoftonly}
\end{align}
where $G^\lr{i}{j}_{\soft}$ is the graph found by contracting $Y_\cci$ with $Y_\ccj$.
Note that the Feynman rules from the soft Wilson line are eikonal, so there are no 4-point vertices, and therefore, no $G_{b}$-type graphs.
Solving for $\cC(s_{\ccO\ccT})$ we find
\be
\cC(s_{\ccO\ccT})=1 + G^\lr{1}{2}_\nsoft  + \cO(\alpha^2)
\ee
where
\be
G^\lr{1}{2}_\nsoft  \equiv G^\lr{1}{2} - G^\lr{1}{2}_\soft \label{Gnss}
\ee
Thus, to verify Eq.~\eqref{factEx1} at 1-loop order all we need to show is that $G^\lr{1}{2}_\nsoft$   is IR finite.

\subsection{IR finiteness}
\label{sec:G12}
The graph of interest is
\be
G^\lr{1}{2}=\fd{1.6cm}{SQEDloop1.pdf}  = -g^2\fint{k} \frac{(2p_\ccO-k)^\alpha}{(p_\ccO-k)^2\pie}
\,\frac{i\Pi_{\alpha\beta}(k)}{k^2\pie} \,
\frac{(2p_\ccT+k)^\beta}{(p_\ccT+k)^2\pie}
\label{G0full}
\ee
where $\Pi^{\mu\nu}$ is given in Eq.~\eqref{Pimndef} in lightcone gauge. The soft graph, from the matrix element of Wilson lines is
\be
G^\lr{1}{2}_\soft = \fint{k} \frac{-ig^2\,p_\ccO\cdot \Pi(k)\cdot p_\ccT}{\big(-p_\ccO\cdot k\pie\big) \big(k^2\pie\big) \big(p_\ccT\cdot k\pie\big)}
\label{G0soft}
\ee
Note that Eq.~\eqref{G0soft} can be obtained from Eq.~\eqref{G0full} with the eikonal approximation. More precisely, we can use the identity
\be
\frac{1}{(p+k)^2\pie} = \frac{1}{2p\cdot k\pie} \bigg(1 - \frac{k^2}{(p+k)^2\pie} \bigg)
\label{softid}
\ee
which holds at $p^2=0$. This identity lets us
replace propagators in the full graph with a sum of eikonal propagators, plus a correction proportional to $k^2$. 
It is similar to the Grammar-Yennie decomposition~\cite{Grammer:1973db} used in many factorization proofs in QCD~\cite{Libby:1978qf,Collins:1981ta,Collins:1989gx}.
Since the original graph was
logarithmically divergent in the soft limit ($k\to 0$), the $k^2$ factors will make the remainder soft finite. That is
$G^\lr{1}{2}_\nsoft =
 G^\lr{1}{2} -G^\lr{1}{2}_\soft$ 
is soft finite.

To see collinear finiteness, we will show that in a generic-lightcone gauge, both $G^\lr{1}{2}$ and $G^\lr{1}{2}_\soft$ are separately collinear finite.
Consider the case $k^\mu \parallel p_\ccO^\mu$.  Then under collinear rescaling $k^2 \to \kappa^2 k^2$ and $k\cdot p_\ccO \to \kappa^2 k\cdot p_\ccO$. If
we ignore the numerator in Eq.~\eqref{G0full}, the diagram would scale like $\kappa^0$ and be logarithmically divergent. For the scaling of
the numerator, we note that we are exactly in the situation where Eq.~\eqref{numsupp2} applies. That is,
\be
p_\ccO\cdot \Pi(k)\cdot p_\ccT 
= \kappa  \frac{r\cdot p_\ccO k_\perp \cdot p_\ccT}{{r\cdot k}}+ \cO(\kappa^2)
\ee
for a generic choice of lightcone gauge reference vector $r^\mu$.
This extra factor of $\kappa$ makes the $G^\lr{1}{2}$ convergent when $k \parallel p_\ccO$. A similar analysis for $k \parallel p_\ccT$ shows that $G^\lr{1}{2}$ is completely
collinear finite. The same argument shows that $G^\lr{1}{2}_\soft$ is collinear finite, and therefore $G^\lr{1}{2}_\nsoft$ has no IR singularities and \Eq{factEx1} is verified at 1-loop order.

For the IR-finite contribution from $G^\lr{1}{2}_\nsoft$, which contributes to the Wilson coefficient, we introduce the diagrammatic notation 
\be
\fd{1.2cm}{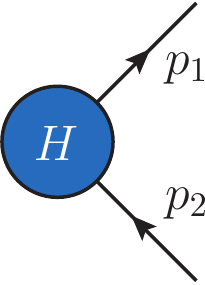} \;=\; 1 + G^{\lr{1}{2}}_\nsoft + \cO(\alpha^2)
\label{Wils1}
\ee
This is a type of reduced diagram we call {\bf hard}.  A hard diagram is IR finite, but relevant at leading power.

\subsection{Explicit result and ${\mathbf t_\ccj^\mu}$-independence}
To calculate the Wilson coefficient, rather than scalar QED, we consider the  more phenomenologically relevant case of a vector current decaying to a $q \bar{q}$ pair, where $\cO = \bar{\psi} \gamma^\mu \psi$. 
For this case, the factorization formula states
\be
\bra{p_\ccO;p_\ccT} \bar{\psi}\gamma^\mu \psi \ket{0} 
\LPeq  \cC(s_{\ccO\ccT})\, \gamma^\mu_{{\spincol \alpha}{\spincol \beta}}
\frac{\bra{p_\ccO}\bar\psi W_\ccO \ket{0}^{{\spincol \alpha}}}{ \bra{0}  Y_\ccO^\dg W_\ccO \ket{0}}\,
\frac{\bra{p_\ccT} W_\ccT^\dg \psi \ket{0}^{{\spincol \beta}}}{ \bra{0} W_\ccT^\dg Y_\ccT \ket{0}}\,
\bra{0} Y_\ccO^\dg  Y_\ccT \ket{0}
\label{C2fact}
\ee
where ${\spincol \alpha}$ and ${\spincol \beta}$ are Dirac spin indices.
To calculate the Wilson coefficient, it is easiest to use Feynman gauge rather than lightcone gauge, where all of the Wilson-line self-interactions vanish. In pure dimensional regularization, all of the diagrams from the factorized expression are scaleless and exactly vanish. The Wilson coefficient is therefore given by $G^\lr{1}{2}$ with the $\frac{1}{\varepsilon}$ and $\frac{1}{\varepsilon^2}$ terms dropped (the UV divergences are removed with $\overline{\text{MS}}$ counterterms and the IR cancel in the matching). 
The Wilson coefficient then comes out to~\cite{Manohar:2003vb,Bauer:2006mk,Bauer:2006qp,Schwartz:2007ib}
\begin{align}
\cC
(s_{\ccO\ccT}) 
&= 1 - \frac{\alpha}{4\pi} \left(8 - \frac{\pi^2}{6} + \ln^2 \frac{-\mu^2}{s_{\ccO\ccT}} + 3 \ln \frac{-\mu^2}{s_{\ccO\ccT}}\right)+ \cO(\alpha^2)
\label{dimregC2}
\end{align}
The Wilson coefficient result is independent of both the IR regulator and the collinear Wilson line directions $t_\ccO^\mu$ and $t_\ccT^\mu$.

To see the $t_\ccO^\mu$ and $t_\ccT^\mu$ independence more nontrivially and the importance of the zero-bin subtraction, one must use an IR regulator other the dimensional regularization. Following~\cite{Manohar:2006nz} on the zero-bin subtraction in SCET (where more details are given) we consider adding an off-shellness regulator.
 The differences between our approach and SCET  are that 1) we use an operator definition of the zero-bin subtraction; 2) we do not have separate soft and collinear modes: all interactions are those in full QCD; and 3) we allow for the collinear Wilson lines to point in arbitrary directions, $t_\ccj^\mu$. These differences are all minor, and the results can essentially be drawn from Eqs. (65)-(70) of~\cite{Manohar:2006nz} with small modifications.

 We can decompose any momentum into lightcone coordinates using the directions
in the soft and collinear Wilson lines, $n_\ccO^\mu$ and $t_\ccO^\mu$:
\be
p^\mu = \frac{ p \cdot t_\ccO}{n_\ccO \cdt t_\ccO} n_\ccO^\mu +  \frac{ p\cdot n_\ccO }{n_\ccO \cdt t_\ccO} t_\ccO^\mu + p_\perp^\mu
\ee
The off-shellness regulator keeps $n_\ccO \cdt p_\ccO >0$ even if $p_\ccO \propto n_\ccO$ as in the external state. Thus
\be
p_\ccO^2 = \frac{2}{n_\ccO \cdt\, t_\ccO} (n_\ccO \cdt p_\ccO)(t_\ccO \cdt p_\ccO) > 0 
\label{psquared}
\ee
We could also have decomposed with respect to $n_\ccT^\mu$ and $t_\ccT^\mu$. If we perform the calculation in 
$4-2\varepsilon$ dimensions, $\varepsilon$ will regulate the UV and soft divergences, with the collinear divergences cut off by the off-shellness.

First, consider the self-energy graphs on the external legs. These are trivially identical on both sides of Eq.\eqref{C2fact} (with any regulator) thus they can be ignored in the matching. Although this is also true in label SCET, it is not trivially true, since the Feynman rules for collinear fields are different from full theory fields.

For the remaining graphs, we present  only the double-logarithmic terms for simplicity, since these manifest all the interesting cancellation.
On the left-hand side of \Eq{C2fact}, the only full-theory graph needed is
\be
G^\lr{1}{2} = \fd{1.5cm}{SQEDloop1.pdf} 
\; \overset{\text{DL}}=\; -\bar{v} \gamma^\mu u \, C_F \frac{\alpha_s}{2\pi}  \ln \frac{p_\ccO^2}{s_{\ccO\ccT}}\ln \frac{p_\ccT^2}{s_{\ccO\ccT}} 
\label{FullDoubleLogs}
\ee
where $\overset{\text{DL}}=$ means equal at double-logarithmic order. 

The graphs needed in the factorized expression are the soft Wilson line graph:
\be
\bra{0}Y_\ccO Y_\ccT^\dg\ket{0} \overset{\text{DL}}= - C_F\frac{\alpha_s}{4\pi}
 \left\{ \frac{2}{\euv^2} 
+ \frac{2}{\euv} \ln\frac{- \mu^2 \, s_{\ccO\ccT} }{ p_\ccO^2 p_\ccT^2} 
+ \ln^2 \frac{-\mu^2\, s_{\ccO\ccT}}{p_\ccT^2 \, p_\ccO^2}
\right\}
\ee
the collinear graphs, without the leg corrections:
\begin{align}
\bra{p_\ccO} \bar \psi W_\ccO \ket{0}
 &\overset{\text{DL}}= - \bar u\, C_F\frac{\alpha_s}{4\pi} \left\{ -\frac{2}{\euv \eir} - \frac{2}{\eir} \ln\frac{\mu^2}{- p_\ccO^2} - \ln^2 \frac{\mu^2}{-p_\ccO^2} + \left( \frac{2}{\eir} - \frac{2}{\euv}\right) \ln \frac{\mu}{t_\ccO \cdt p_\ccO}  
\right\}
\\
\bra{p_\ccT} W_\ccT^\dg  \psi \ket{0}
 &\overset{\text{DL}}= -C_F\frac{\alpha_s}{4\pi} \left\{ -\frac{2}{\euv \eir} - \frac{2}{\eir} \ln\frac{\mu^2}{- p_\ccT^2} - \ln^2 \frac{\mu^2}{-p_\ccT^2} + \left( \frac{2}{\eir} - \frac{2}{\euv}\right) \ln \frac{\mu}{t_\ccT \cdt p_\ccT}  
\right\} v
\end{align}
and the zero-bin subtractions:
\begin{align}
\IRZ_\ccO=\frac{1}{N_c} \tr \bra{0} Y_\ccO^\dg W_\ccO \ket{0}
 &\overset{\text{DL}}= C_F\frac{\alpha_s}{4\pi} 
\left( \frac{2}{\eir} - \frac{2}{\euv}\right) \left(\frac{1}{\euv} + \ln \frac{\mu^2}{-p_\ccO^2} - \ln\frac{\mu}{t_\ccO\cdt p_\ccO} \right)
\\
\IRZ_\ccT=\frac{1}{N_c} \tr\bra{0} W_\ccT^\dg Y_\ccT \ket{0}
 &\overset{\text{DL}}= C_F\frac{\alpha_s}{4\pi} 
\left( \frac{2}{\eir} - \frac{2}{\euv}\right) \left(\frac{1}{\euv} + \ln \frac{\mu^2}{-p_\ccT^2} - \ln\frac{\mu}{t_\ccT\cdt p_\ccT} \right)
\end{align}
This notation and normalization for the zero bin subtraction will be explained in Sections~\ref{sec:QCD} and~\ref{sec:SCET}.
Note that the appearance of the hard scales $t_\ccO \cdt p_\ccO$ and $t_\ccT \cdt p_\ccT$ is illusory --- using \Eq{psquared}, one can express $\IRZ_\ccO$ and $\IRZ_\ccT$ in terms
of the off-shellnesses $n_\ccO \cdt p_\ccO$ and $n_\ccT \cdt p_\ccT$ alone.

Therefore,
\begin{align}
\frac{\bra{p_\ccO} \bar \psi W_\ccO \ket{0}}{\IRZ_\ccO}
 &\overset{\text{DL}}= -\bar u\, C_F\frac{\alpha_s}{4\pi} \left\{ -\frac{2}{\euv^2} - \frac{2}{\euv} \ln\frac{\mu^2}{- p_\ccO^2} - \ln^2 \frac{\mu^2}{-p_\ccO^2} 
\right\}
\\
\frac{\bra{p_\ccT} W_\ccT^\dg  \psi \ket{0}}{\IRZ_\ccT}
 &\overset{\text{DL}}= -C_F\frac{\alpha_s}{4\pi} \left\{ -\frac{2}{\euv^2} - \frac{2}{\euv} \ln\frac{\mu^2}{- p_\ccT^2} - \ln^2 \frac{\mu^2}{-p_\ccT^2} 
\right\}\, v
\end{align}
These equations show that each collinear sector is independent of the Wilson-line directions, $t^\mu_\ccj$, and is only $p_\ccj$-collinear sensitive as evidenced by the cancellation of the 
$\eir$ poles.

Putting everything together up to 1-loop we find:
\begin{multline}
\gamma^\mu_{ {\spincol \alpha\beta}}
\frac{\bra{p_\ccO}{W_\ccO\bar\psi} \ket{0}^{{\spincol \alpha}}}{\IRZ_\ccO}
\frac{\bra{p_\ccT} {W_\ccT^\dg \psi} \ket{0}^{ {\spincol \beta}}}{\IRZ_\ccT}
\bra{0} Y_\ccO^\dg  Y_\ccT \ket{0}\\
\overset{\text{DL}}= -\bar{v} \gamma^\mu u \,
 C_F\frac{\alpha_s}{4\pi}\bigg\{
-\frac{2}{\euv^2} 
- \frac{2}{\euv} \ln \frac{ \mu^2}{-s_{\ccO\ccT}}  
+  \ln^2\frac{ \mu^2}{-s_{\ccO\ccT}}  
+ 2\ln \frac{-p_\ccT^2}{\mu^2} \ln\frac{-p_\ccO^2}{\mu^2}
\bigg\}
\end{multline}
Comparing to  the full-QCD matrix element shown in \Eq{FullDoubleLogs}, we see that, to double-logarithmic order, the IR-divergences in the full theory and factorized expression exactly agree.

\section{Example 2: two collinear particles}
\label{sec:1loopEx2}

As the next illustrative example, we consider a state with two particles in one jet.
That is we consider
$\bra{p_\ccO, q; p_\ccT} \phi^\star \phi \ket{0}$, for which the factorization formula reads
\be
\bra{p_\ccO, q;p_\ccT} \phi^\star  \phi \ket{0} 
\;\LPeq\;  \cC(S_{\ccO\ccT})\, 
\frac{\bra{p_\ccO,q} \phi^\star  W_\ccO\ket{0}}{ \bra{0}Y_\ccO^\dg W_\ccO  \ket{0}}\,
\frac{\bra{p_\ccT} W_\ccT^\dg \phi \ket{0}}{ \bra{0} W_\ccT^\dg Y_\ccT \ket{0}}\, \bra{0} Y_\ccO^\dg Y_\ccT \ket{0}
\label{factEx2}
\ee
where $P_\ccO^\mu = p_\ccO^\mu+q^\mu$,  $P_\ccT^\mu = p_\ccT^\mu$ and $S_{\ccO\ccT} \LPeq (P_\ccO+P_\ccT)^2 \equiv Q^2$.
In this case, the two sides are not equal, but equal at leading power in $\lambda$, where
$\lambda = P_\ccO^2/Q^2$. We also must show that the Wilson coefficient
$\cC(S_{\ccO\ccT})$ is the same function computed with minimal collinear sectors, as in the previous section.
 This example will illustrate the role played by real-emission and IR-sensitive graphs in  factorization.

\subsection{Overview of graphs}

In this example, since we have an external photon, we must choose a reference vector for its polarization.
It is natural to choose the same generic-$r$ reference vector as in the lightcone-gauge photon propagator. 
So 
$r_\mu \epsilon^\mu(q)=q_\mu \epsilon^\mu(q)=0$.
 These constraints define the polarization vectors that are consistent with generic-lightcone gauge completely:
\be
\epsilon^-(q;r) = \sqrt{2}\, \frac{q \r [r}{[qr]} 
\quad\And\quad
\epsilon^+(q;r) = \sqrt{2}\, \frac{r \r [q}{\l rq \r}
\ee 
where we use the spinor-helicity formalism to ease the discussion of the dependence on the reference vector, $r$, of amplitudes. Our conventions for the spinor-helicity formalism are given in \tree, however, we will not need any details of the spinor-helicity formalism in this paper as everything we need concerning polarization vectors will be taken from \tree.
We also choose $t_\ccO=t_\ccT=r$ for the collinear Wilson lines to decouple them completely. Thus we can set $W_\ccO=W_\ccT=1$ in this example.

As in the previous example, many graphs contribute to both the left-hand side and right-hand side of 
Eq.~\eqref{factEx2}. In particular, all graphs involving one leg only in the full theory matrix element, such as
\be
\fd{1.4cm}{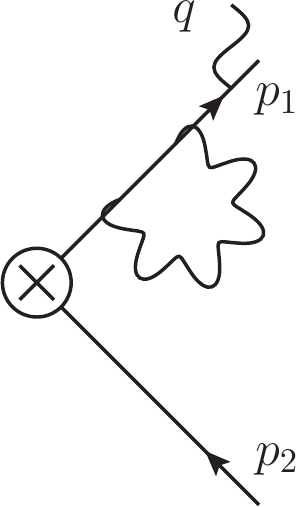} \;,\quad
\fd{1.4cm}{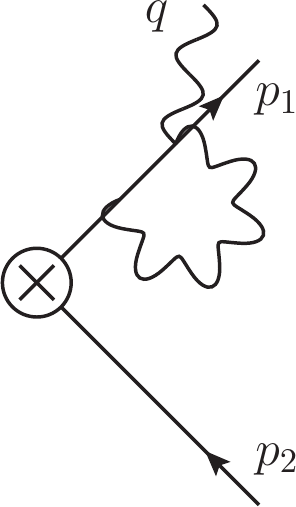} \;,\quad
\fd{1.4cm}{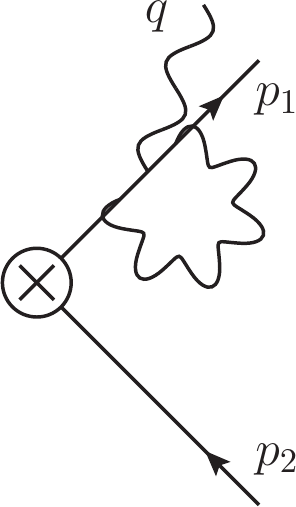} \;,\quad
\fd{1.4cm}{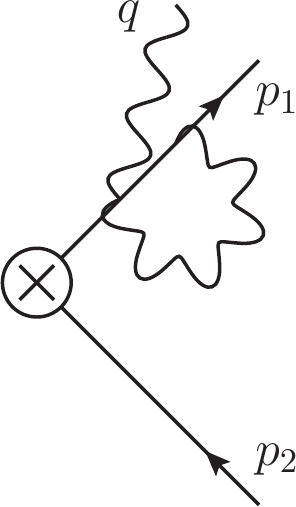} \;,\quad
\fd{1.4cm}{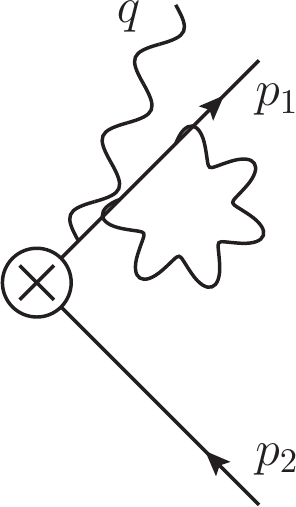} \;,\quad
\fd{1.4cm}{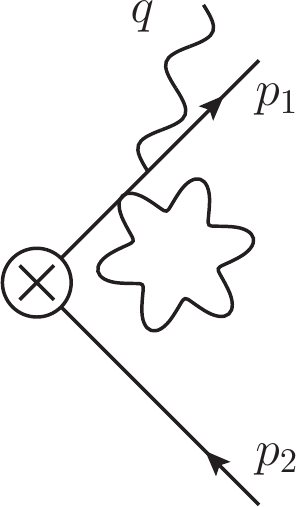} \;,\quad
\fd{1.5cm}{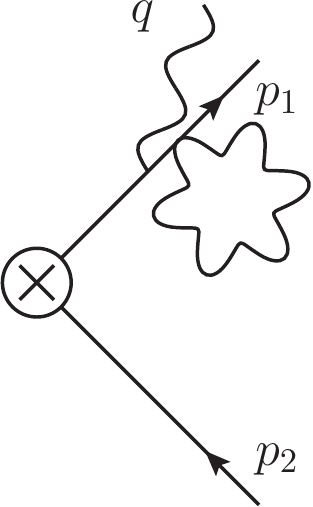} 
\label{SQEDselfenergies}
\ee
contribute to the right-hand side through $\bra{p_\ccO,q}\phi^\star \ket{0}$. Also trivially-factorizing cross terms,
such as 
\be
\bra{p_\ccO,q}\phi^\star \ket{0}_\text{tree}
\bra{p_\ccT}\phi\ket{0}_\text{1-loop}
\;=\;
\fd{1.4cm}{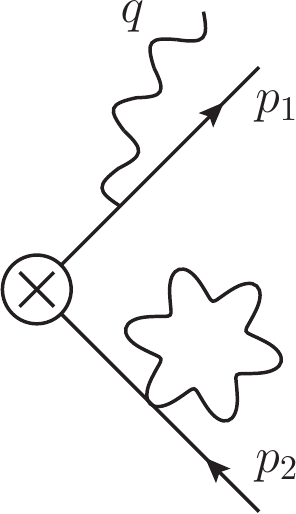} \;+\;\; \fd{1.4cm}{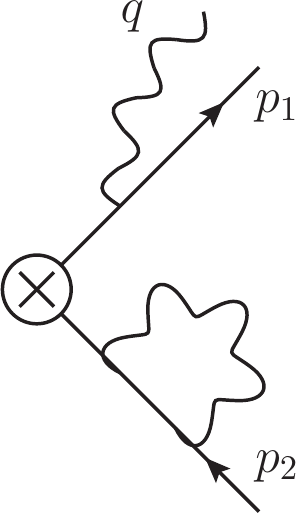}
\label{crossdiagrams}
\ee
contribute identically on both sides of  Eq.~\eqref{factEx2}.

The remaining graphs from the left-hand side of \Eq{factEx2} either have a loop connecting the two legs and the emission coming off either the $p_\ccO$ leg:
\be
G^{\lr{1}{2},a} \;\equiv\; \fd{1.5cm}{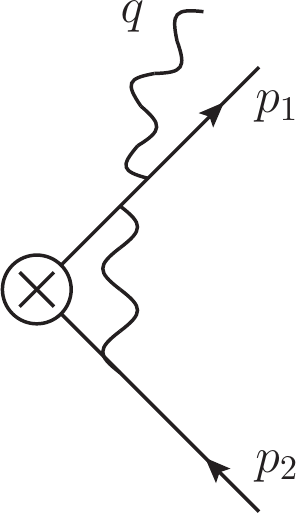} 
\;,\qquad
G^{\lr{1}{2},b} \;\equiv\; \fd{1.5cm}{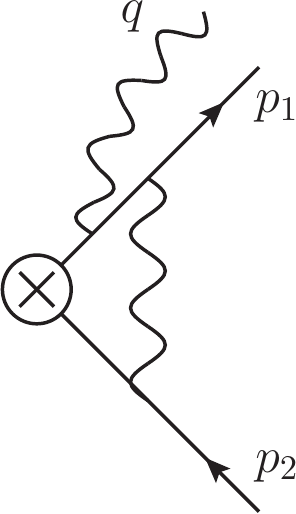} 
\;,\qquad 
G^{\lr{1}{2},c}\;\equiv\; \fd{1.5cm}{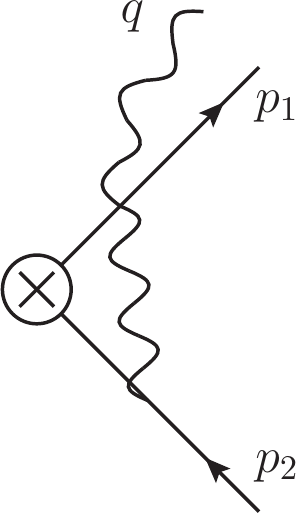}
\label{1loopEx2CollCoupGraphs}
\ee
or they have the emission coming off of the $p_\ccT$ leg with the loop anywhere:
\be
\fd{1.8cm}{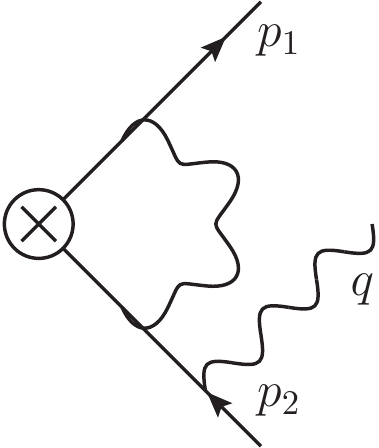}
\;,\qquad
\fd{1.7cm}{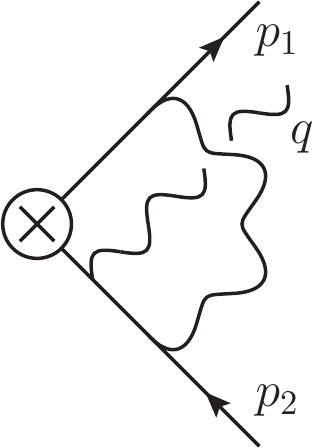}
\;,\qquad
\fd{1.7cm}{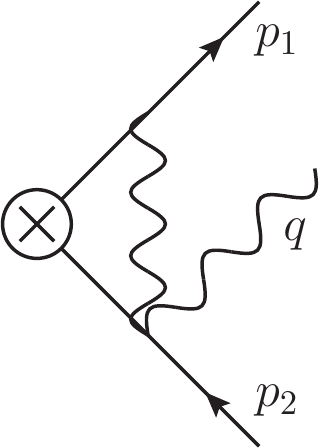} 
\;,\qquad
\fd{1.7cm}{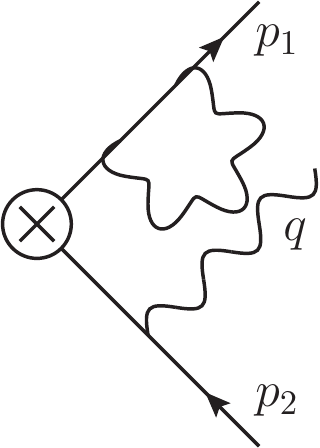} 
\;,\;\ldots
\label{RedSoft3}
\ee

With generic reference vectors, the twelve graphs in \Eq{RedSoft3} are power suppressed compared to the graphs where the emission comes off of the $p_\ccO$ leg. Indeed, graphs which contribute at leading power must have a factor of $\frac{1}{q\cdot p_\ccO} \sim \lambda^{-2}$, as does $G^{\lr{1}{2},a}$. The graphs with the emission coming from the $p_\ccT$ leg have instead  
$\frac{1}{q\cdot p_\ccT}\sim \lambda^0$ factors which are subleading power. The fact that non-self-collinear emissions are power suppressed in generic-lightcone gauge was discussed elaborately in \tree. This result holds at loop level as well, simply because in generic-lightcone gauge a non-self-collinear emission can never have an enhanced propagator. We will come back to the general discussion in the next section and focus, for now, on the 1-loop example at hand. The result is that we do not need to consider the graphs in Eq.~\eqref{RedSoft3} at leading power.

 Note that the power suppression in $\lambda$ holds whether or not the graphs are IR finite. Although power counting something infinite may seem bizarre, one should keep in mind that the IR divergences in loops are always ultimately canceled by phase-space integrals in computing IR-safe observables. Thus, power-suppressed IR divergences translate to power-suppressed finite contributions, which is why we can drop them.

The remaining graphs contributing to the right-hand side of Eq.~\eqref{factEx2} come from the tree-level real emission multiplied by the Wilson coefficient and soft-Wilson-line terms at 1-loop order:
\be
\bra{p_\ccO,q}\phi^\star \ket{0}_\text{tree}
\times
\bigg\{
\cC(S_{\ccO\ccT})\,
\frac{\bra{0} Y_\ccO^\dg Y_\ccT \ket{0}}{ \bra{0} Y_\ccO^\dg \ket{0}\bra{0} Y_\ccT \ket{0} }
\;\overset{\text{1-loop}}=\;
\contraction[.7ex]{\bra{0}}{Y_\ccO^\dg}{}{Y_\ccT}\bra{0} Y_\ccO^\dg Y_\ccT \ket{0} + G^\lr{1}{2}_\nsoft \bigg\}
\label{softWilsfracEx2}
\ee
where $G^\lr{1}{2}_\nsoft$, defined in Eq.~\eqref{Gnss}, comes from the calculation of the 1-loop Wilson coefficient in the previous section.

 What we will now show is that the $\contraction[.7ex]{\bra{0}}{Y_\ccO^\dg}{}{Y_\ccT}\bra{0} Y_\ccO^\dg Y_\ccT \ket{0}$ term  in Eq.~\eqref{softWilsfracEx2} reproduces the
 sum of the soft limits of  $G^{\lr{1}{2},a}$ or $G^{\lr{1}{2},b}$ at leading power, the
  $G^\lr{1}{2}_\nsoft$ term reproduces the
  non-soft part of $G^{\lr{1}{2},a}$ at leading power, and both $G^{\lr{1}{2},c}$ and the non-soft part of $G^{\lr{1}{2},b}$ are power suppressed, hence proving \Eq{factEx2} at 1-loop order.

\subsection{The graph $G^{\lr{1}{2},a}$}
\label{sec:irsensitive}
Writing out the Feynman rules, we find
\be
G^{\lr{1}{2},a} =\; g\,\frac{p_\ccO\cdot\epsq}{p_\ccO\cdot q}
\times ig^2 \fint{k} \frac{(2p_\ccO+2q-k)\cdot\Pi(k)\cdot(2p_\ccT+k)}{k^2\,(p_\ccT+k)^2\,(p_\ccO+q-k)^2}
\ee
As in the previous example, we will write this graph as
\be
G^{\lr{1}{2},a} =G^{\lr{1}{2},a}_\soft+ G^{\lr{1}{2},a}_\nsoft 
\ee
where the soft-sensitive part is found by dropping terms which are subleading in $\kappa$ after the rescaling $k^\mu\to \kappa^2 k^\mu$. We draw the soft limit with the soft photon colored red and with a long wavelength. That is,
\be
\fd{1.8cm}{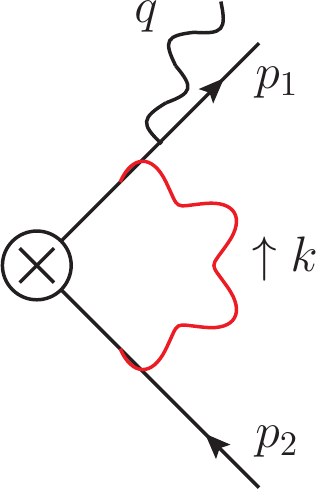} 
\; = G^{\lr{1}{2},a}_\soft= g\,\frac{p_\ccO\cdot\epsq}{p_\ccO\cdot q}
\times \fint{k} \frac{2ig^2\, (p_\ccO+q)\cdot\Pi(k)\cdot p_\ccT}{k^2\,\big(p_\ccT\cdot k \pie \big)\,\big((p_\ccO+q)^2-2(p_\ccO+q)\cdot k \pie)}
\label{RedSoft1}
\ee
This graph is not IR divergent, but it is IR sensitive.
Because $(p_\ccO+q)^2\sim \lambda^2$, in taking the soft limit, we did not drop $2(p_\ccO+q)\cdot k$ in favor of $(p_\ccO+q)^2$. Doing so would have assumed a certain order of limits, essentially $\kappa \ll \lambda$, which would lead to inconsistent results. More precisely, if we were to integrate over the phase space of $q$ to produce an IR-safe cross section, the region where $q\cdot p_\ccO \approx 0$ must be treated independently of the region of $k^\mu\approx0$ in the loop integral. That is, the only way for the order of integration of the loop and phase-space integrals to not matter is if we keep both terms.

Now, since we keep $(p_\ccO+q)^2 > 0$ the loop integral is not soft-divergent. This is clear from counting powers of $\kappa$ as $k^\mu\to\kappa^2 k^\mu$, which gives $G^{\lr{1}{2},a}_\soft \to \kappa G^{\lr{1}{2},a}_\soft$. However, if $(p_\ccO+q)^2=0$, the loop scales like $\kappa^0$ and is logarithmically soft divergent.
Thus, for $(p_\ccO+q)^2 \sim \lambda^2$ with $\lambda$ small, $\lambda$ acts like an IR cutoff. We, therefore, have that 
\be
G^{\lr{1}{2},a}_\soft \,\sim\, g\,\frac{p_\ccO\cdot\epsq}{p_\ccO\cdot q} \, g^2 \ln \lambda 
\ee
This singular-$\lambda$ dependence must be reproduced by the factorized expression, as the Wilson coefficient
is $\lambda$ independent.
On the other hand, the non-soft part of the loop, $G^{\lr{1}{2},a}_\nsoft =G^{\lr{1}{2},a}-G^{\lr{1}{2},a}_\soft$ is
free of soft divergences, even at $\lambda=0$ (except for the prefactor, of course).  This follows from the eikonal substitution in Eq.~\eqref{softid} which adds
additional powers of $k^2$ to the non-soft part. 

Both the soft and non-soft parts of the loop are also
collinear finite in generic-lightcone gauge.
 This holds for the exact same reason that $G^{\lr{1}{2}}_\nsoft$ was collinear-finite in the previous section: in generic--lightcone gauge, the numerator of $G^{\lr{1}{2},a}$ is suppressed when $k$ becomes collinear to $p_\ccO$ or $p_\ccT$ as in \Eq{numsupp1} and \Eq{numsupp2}. Thus, $G^{\lr{1}{2},a}$ is collinear-finite (even when $(p_\ccO+q)^2=0$), implying that $G^{\lr{1}{2},a}_\nsoft$ is IR-insensitive (collinear and soft insensitive) since $G^{\lr{1}{2},a}_\nsoft$ has the soft sensitivity subtracted off. 
 
Because the loop integral in $G^{\lr{1}{2},a}_\nsoft$ is IR-finite even when $(p_\ccO+q)^2=0$, we can expand it in powers of $\lambda$ in the integrand, and only keep the leading term. The leading term in this 
expansion corresponds to treating $P_\ccO^\mu = p_\ccO^\mu+q^\mu$ as being lightlike. Performing this
expansion on $G^{\lr{1}{2},a}$ and $G^{\lr{1}{2},a}_\soft$ shows that they reduce to the integrals in $G^{\lr{1}{2}}$ and 
$G^{\lr{1}{2}}_\soft$, respectively, from the previous section. Since both loops are the same, so is their difference, $G^{\lr{1}{2},a}_\nsoft$. That is,
\be
G^{\lr{1}{2},a}_\nsoft \big(P_\ccO,p_\ccT\big) \;\LPeq\; -g\,\frac{p_\ccO\cdot\epsq}{p_\ccO\cdot q}
\times  G^{\lr{1}{2}}_\nsoft \big(P_\ccO, p_\ccT\big) 
 \;\overset{\text{1-loop}}\LPeq\;
 \fd{1.7cm}{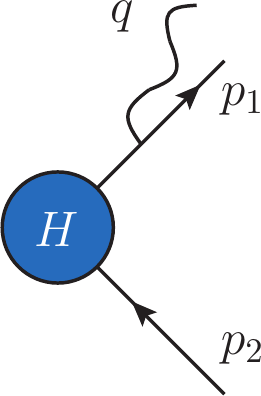} 
\label{Wils2}
\ee
where $G^{\lr{1}{2}}_\nsoft(p_\ccO,p_\ccT)$ was the IR-finite and $\lambda$-independent 1-loop contribution to the Wilson coefficient found in the previous section. 

Therefore, the graph $G^{\lr{1}{2},a}_\nsoft$ from the left-hand side of \Eq{factEx2} is reproduced by the factorized expression in last term in brackets in Eq.~\eqref{softWilsfracEx2}.

\subsection{The graph $G^{\lr{1}{2},b}$}

We now analyze the second diagram that seems to break collinear factorization in \Eq{1loopEx2CollCoupGraphs}, namely
\be
G^{\lr{1}{2},b} = 2ig^3 \fint{k} \frac{(2p_\ccO-k)\cdot\Pi(k)\cdot(2p_\ccT+k) \;(p_\ccO-k)\cdot\epsq}{k^2(p_\ccT+k)^2(p_\ccO-k)^2(p_\ccO+q-k)^2}
\label{G12bdef}
\ee
The soft limit of this graph, again keeping the IR-sensitive parts, is
\be
\fd{1.9cm}{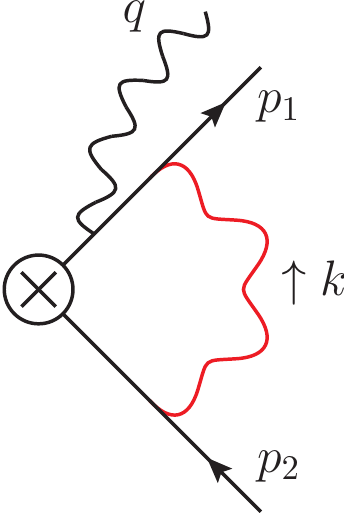}
\; \equiv G^{\lr{1}{2},b}_\soft = -2g\,p_\ccO\cdot\epsq \fint{k} \frac{ig^2\; p_\ccO\cdot\Pi(k)\cdot p_\ccT}
	{k^2\, p_\ccT\cdot k\, p_\ccO\cdot k\, \big((p_\ccO+q)^2 - 2(p_\ccO+q)\cdot k \big)}
\ee
This graph is soft divergent, scaling as $\kappa^0$ even with $(p_\ccO+q)^2\ne0$, thus it must be reproduced in the factorized expression.

Next, we will show that $G^{\lr{1}{2},b}_\nsoft$ is collinear sensitive, but power suppressed compared to
$G^{\lr{1}{2},a}_\nsoft$. First, to see that  $G^{\lr{1}{2},b}$ is collinear finite at finite $(p_\ccO+q)^2$,
we note that 
for $(p_\ccO+q)^2$ positive and fixed, the $(p_\ccO+q-k)^2$ propagator cannot go on-shell when other propagators do, so the loop is not more singular than $G^{\lr{1}{2},a}$. As with  $G^{\lr{1}{2},a}$, it would be collinear divergent for $k\parallel p_\ccO$ or $k\parallel p_\ccT$ but for the fact that the numerator vanishes by \Eq{numsupp1} and \Eq{numsupp2} which causes the integral to be collinear finite for $(p_\ccO+q)^2\ne 0$.

Now, if $(p_\ccO+q)^2=0$, then the integral would be $p_\ccO$-collinear divergent (though it remains $p_\ccT$-collinear finite). This can be seen by 
taking $p_\ccO\propto q$ in which case $k^\mu$ scales like
\be
k^\mu \sim \kappa^0 \, p_\ccO^\mu + \kappa^2 \, p_\ccT^\mu + \kappa\, k^\mu_{\perp} 
\ee
and so $G^{\lr{1}{2},b}$ in \Eq{G12bdef} scales like
\be
G^{\lr{1}{2},b} 
\sim \int d^4k  \kappa^4 \; \frac{\kappa \, \kappa}{ \kappa^2 \, \kappa^0 \, \kappa^2 \, \kappa^2}
\sim \kappa^0
\ee
where we used that $d^4k \sim \kappa^4$, $(2p_\ccO-k)\cdot\Pi(k)\cdot(2p_\ccT+k) \sim \kappa$,  $k\cdot\epsq\sim\kappa$, and $(p_\ccT+k)^2 \sim \kappa$.
We thus see that $G^{\lr{1}{2},b}$ is logarithmically $p_\ccO$-collinear divergent. We have made all of these arguments for $G^{\lr{1}{2},b}$, but they apply also to  $G^{\lr{1}{2},b}_\soft$ and hence to $G^{\lr{1}{2},b}_\nsoft$. Then, given that $G^{\lr{1}{2},b}_\nsoft$ is completely IR-finite when $(p_\ccO+q)^2\ne 0$ but logarithmically $p_\ccO$-collinear divergent when $(p_\ccO+q)^2=0$, we must have that it scales like
\be
G^{\lr{1}{2},b}_\nsoft \,\sim\, g^3\ln\left[(p_\ccO+q)^2\right] \,\sim\, g^3\ln\lambda
\label{loglambda}
\ee
for small $\lambda$. This is power suppressed compared to say Eq.~\eqref{Wils2} which scales like $\lambda^{-1}$. Thus, we can drop $G^{\lr{1}{2},b}_\nsoft$ at leading power.

\subsection{The graph $G^{\lr{1}{2},c}$}

Finally, we have the graph with the scalar-QED 4-point vertex
\be
G^{\lr{1}{2},c}\;\equiv\; \fd{1.2cm}{RedDiagEx2.pdf}
= -2ig^3\fint{k} \frac{\epsq\cdot\Pi(k)\cdot(2p_\ccT+k)}{k^2(p_\ccT+k)^2(p_\ccO+q-k)^2}
\ee
We will show that this graph is completely power suppressed. 

To see if there are soft divergences, we look at the soft limit of $G^{\lr{1}{2},c}$.
First, note that if $(p_\ccO+q)^2 \ne 0$ then $G^{\lr{1}{2},c}$ would be finite in the soft limit, as can be seen by counting powers of the soft momentum in the integrand which gives $d^4k \big/ k^3$. On the other hand, for $(p_\ccO+q)^2 = 0$, the integrand of $G^{\lr{1}{2},c}$ becomes $d^4k \big/ k^4$ signaling a logarithmic divergence. Thus, we must have that, in the soft region of the integral,
\be
G^{\lr{1}{2},c} \,\overset{\text{soft}}\sim\, g^3 \ln (p_\ccO+q)^2 \,\sim\, g^3 \ln \lambda^2 \,\ll\, \frac{g^3}{\lambda}
\ee
Hence, in the soft limit, $G^{\lr{1}{2},c}$ is power suppressed.

We have seen that $G^{\lr{1}{2},c}$ is power suppressed in the soft limit.  Next, we will now show that the same is true for the collinear limits of the integral, meaning that the entire graph $G^{\lr{1}{2},c}$ is a power correction in our factorization formula. We start by showing that $G^{\lr{1}{2},c}$ is $p_\ccT$-collinear finite in generic-lightcone gauge. This holds for the same reason as for the other
collinear-finite graphs: were it not for the numerator, $G^{\lr{1}{2},c}$ would be logarithmically $p_\ccT$-collinear divergent. However, when $k$ becomes collinear to $p_\ccT$, $\Pi(k)$ becomes the polarization sum of photons in the $p_\ccT$ direction which is transverse to $p_\ccT$. Hence $\Pi(k)\cdot(2p_\ccT+k) \to 0$ when $k\parallel p_\ccT$. These are the words that describe \Eq{numsupp1} and \Eq{numsupp2}. Hence, $G^{\lr{1}{2},c}$ is $p_\ccT$-collinear finite. 

$G^{\lr{1}{2},c}$ is also $p_\ccO$-collinear finite, but only when $(p_\ccO+q)^2\ne0$. This can be seen by power counting the denominator, as $k$ 
becomes collinear to $p_\ccO$. For $(p_\ccO+q)^2 = 0$, the denominator of $G^{\lr{1}{2},c}$ causes it to be logarithmically divergent, but in this case the numerator does not vanish as $k\parallel p_\ccO$ since $\Pi(k)$ is not transverse to $\epsq$. That is,
\be
\epsq^\mu \, \Pi_{\mu\nu}(k) \;=\; -\epsq_\nu + \frac{k\cdot\epsq\, r_\nu}{r\cdot k}
\;\too\; -\epsq_\nu  
\qquad \text{for } k\parallel p_\ccO \parallel q 
\ee
where we used that $r\cdot\epsq = 0$. Thus, when $k\parallel p_\ccO$ the numerator of $G^{\lr{1}{2},c}$ looks like $p_\ccT\cdot \epsq$ which does not vanish. Since $G^{\lr{1}{2},c}$ is collinear finite for $(p_\ccO+q)^2\ne0$ and has a logarithmic divergence for $k\parallel p_\ccO$ when $(p_\ccO+q)^2=0$, we conclude that in the $k\parallel p_\ccO$ region of the integral
\be
G^{\lr{1}{2},c} \,\overset{p_\ccO\text{-coll}}\sim\, g^3 \ln (p_\ccO+q)^2 \,\sim\, g^3 \ln \lambda^2 \,\ll\, \frac{g^3}{\lambda}
\ee
Thus, the entire integral in $G^{\lr{1}{2},c}$ is power suppressed compared to the leading-power matrix element, $\dfrac{p_\ccO\cdot\epsq }{ p_\ccO\cdot q} \sim \lambda^{-1}$.

\subsection{Putting it together}
We have shown that most of the contributions to Eq.~\eqref{factEx2} agree identically on both sides. The ones that do not are $G^{\lr{1}{2},a}$, $G^{\lr{1}{2},b}$ and $G^{\lr{1}{2},c}$ in Eq.~\eqref{1loopEx2CollCoupGraphs} for the left hand side and Eq.~\eqref{softWilsfracEx2} for the right-hand side. Of these,
$G^{\lr{1}{2},c}$ is power suppressed, as is the non-soft part of $G^{\lr{1}{2},b}$. Thus the nontrivial
leading-power diagrams are
\be
G^{\lr{1}{2},a}_\nsoft  \LPeq\; \fd{1.6cm}{RedDiagEx3red.pdf} \;,
\qquad
G^{\lr{1}{2},a}_\soft  \LPeq\;  \fd{2cm}{RedDiagEx3s.pdf} \;,
\qquad
G^{\lr{1}{2},b}_\soft  \LPeq\;  \fd{2cm}{RedDiagEx1s.pdf} 
\label{Ex2RedDiags1}
\ee
We also showed that $G^{\lr{1}{2},a}_\nsoft$ reproduces the contribution from the Wilson coefficient
in  Eq.~\eqref{softWilsfracEx2}. Thus what remains is to show that the contribution connecting the two soft Wilson lines in the factorized expression agrees with $G^{\lr{1}{2},a}_\soft+ G^{\lr{1}{2},b}_\soft$ at leading power. We do this by direct calculation.

Let us define a lightlike directions $n_\cci^\mu=(1,\vec{n}_\cci)$, such that $p_\cci^\mu = \frac12 \bar n_\cci \cdot p_\cci \, n_\cci^\mu$, then
\begin{align}
& G^{\lr{1}{2},a}_\soft + G^{\lr{1}{2},b}_\soft 
= \fint{k} \frac{2ig^3\,p_\ccO\cdot\epsq}{k^2\,p_\ccT\cdot k\,\big((p_\ccO+q)^2-2(p_\ccO+q)\cdot k)} 
\notag
\\&\hspace{55mm}\times \bigg[
\frac{(p_\ccO+q)\cdot\Pi(k)\cdot p_\ccT}{p_\ccO\cdot q}
- \frac{p_\ccO\cdot\Pi(k)\cdot p_\ccT }{p_\ccO\cdot k} \bigg] 
\notag
\\
&\LPeq \fint{k} \frac{\frac12 ig^3\,n_\ccO\cdot\Pi(k)\cdot p_\ccT\;p_\ccO\cdot\epsq}{k^2\,p_\ccT\cdot k\,\big(p_\ccO\cdot q-(p_\ccO+q)\cdot k)} 
\bigg[
\frac{\bar n_\ccO\cdot(p_\ccO+q)}{p_\ccO\cdot q} \frac{\frac12 \bar n_\ccO\cdot p_\ccO \,n_\ccO\cdot k}{p_\ccO\cdot k}
- \frac{\bar n_\ccO\cdot p_\ccO}{p_\ccO\cdot k} \frac{p_\ccO\cdot q}{p_\ccO\cdot q} \bigg]
\notag
\\&\LPeq \fint{k} \frac{ ig^3\, p_\ccO\cdot\Pi(k)\cdot p_\ccT\;p_\ccO\cdot\epsq}{k^2\,p_\ccT\cdot k\,\big(p_\ccO\cdot q-(p_\ccO+q)\cdot k)} 
\bigg[
\frac{(p_\ccO+q)\cdot k}{p_\ccO\cdot q\, p_\ccO\cdot k} - \frac{p_\ccO\cdot q}{p_\ccO\cdot q\,p_\ccO\cdot k} \bigg]
\notag
\\&= -g \frac{p_\ccO\cdot\epsq}{p_\ccO\cdot q} \times ig^2\fint{k} \frac{n_\ccO\cdot\Pi(k)\cdot n_\ccT}{k^2\,n_\ccO\cdot k\,n_\ccT\cdot k} 
\label{SCfactEx2}
\end{align}
The first term is the tree-level term in $\bra{p_\ccO,q}\phi^\star \ket{0}$ and the second term is the loop integral, $\contraction[.7ex]{\bra{0}}{Y_\ccO^\dg}{}{Y_\ccT}\bra{0} Y_\ccO^\dg Y_\ccT \ket{0}$, where the photon propagates between the Wilson lines. This is exactly equal to the rest of the factorized expression by \Eq{softWilsfracEx2}.

This completes the check that the sum of the 1-loop diagrams on both sides of \Eq{factEx2} agree at leading power and that the Wilson coefficients are the same and IR insensitive.

\section{Outline of all-orders proof}
\label{sec:proofoutline}

In the previous two sections, we checked special cases of the factorization formula at 1-loop order by matching diagrams. This approach is not sustainable for an all-orders proof. Moreover, even when two diagrams are identical on both sides, dropping them from consideration somewhat obscures 
the physics of factorization. For example, the loops in \Eq{SQEDselfenergies} have both soft and non-soft parts, but it was easier not to separate them when matching them loop-for-loop with those in $\bra{p_\ccO,q}\phi^\star \ket{0}\bra{p_\ccT}\phi\ket{0}$. If we had separated the soft and non-soft parts, we would have found that the sum of the non-soft parts of the graphs in \Eq{SQEDselfenergies} is exactly $\bra{p_\ccO,q}\phi^\star \ket{0} \big/ \bra{0} Y_\ccO^\dg \ket{0} $ and the soft parts are exactly $\bra{p_\ccO,q}\phi^\star \ket{0}_\text{tree}\contraction[.5ex]{\bra{0}}{\hspace{-4mm}Y_\ccO^\dg}{\hspace{0mm}}{Y_\ccT}\bra{0} Y_\ccO^\dg Y_\ccT \ket{0}$, where the contraction indicates the the photon connects only to $Y_\ccO^\dg$. Both these approaches are equivalent, but in the latter we see that all of the soft physics is contained in $\bra{0} Y_\ccO^\dg Y_\ccT \ket{0}$; $\bra{p_\ccO,q}\phi^\star \ket{0} \big/ \bra{0} Y_\ccO^\dg \ket{0} $ is soft-inensitive.

Proving soft-collinear factorization in general, will involve 4 steps
\begin{enumerate}

\item Write each diagram contributing to the matrix element in the full theory as a sum of colored diagrams where each virtual gluon can either contribute to a soft singularity, in which case we call it soft sensitive (and draw it with a long-wavelength red line), or it cannot, in which case we call it soft insensitive (and draw it with a blue line).

\item Drop diagrams which cannot contribute at leading power and identify finite diagrams. Doing this in physical gauges lets us write the full-theory matrix element as the sum of
colored diagrams with a restricted topology in the following way
\be
\bra{X_\ccO\cdots X_\rN;X_\scs} \cO \ket{0}
 \;\overset{\substack{\text{physical} \\[0.5mm] \text{gauges} }}\LPeq\; 
 \sum_\text{diagrams}
\fd{6cm}{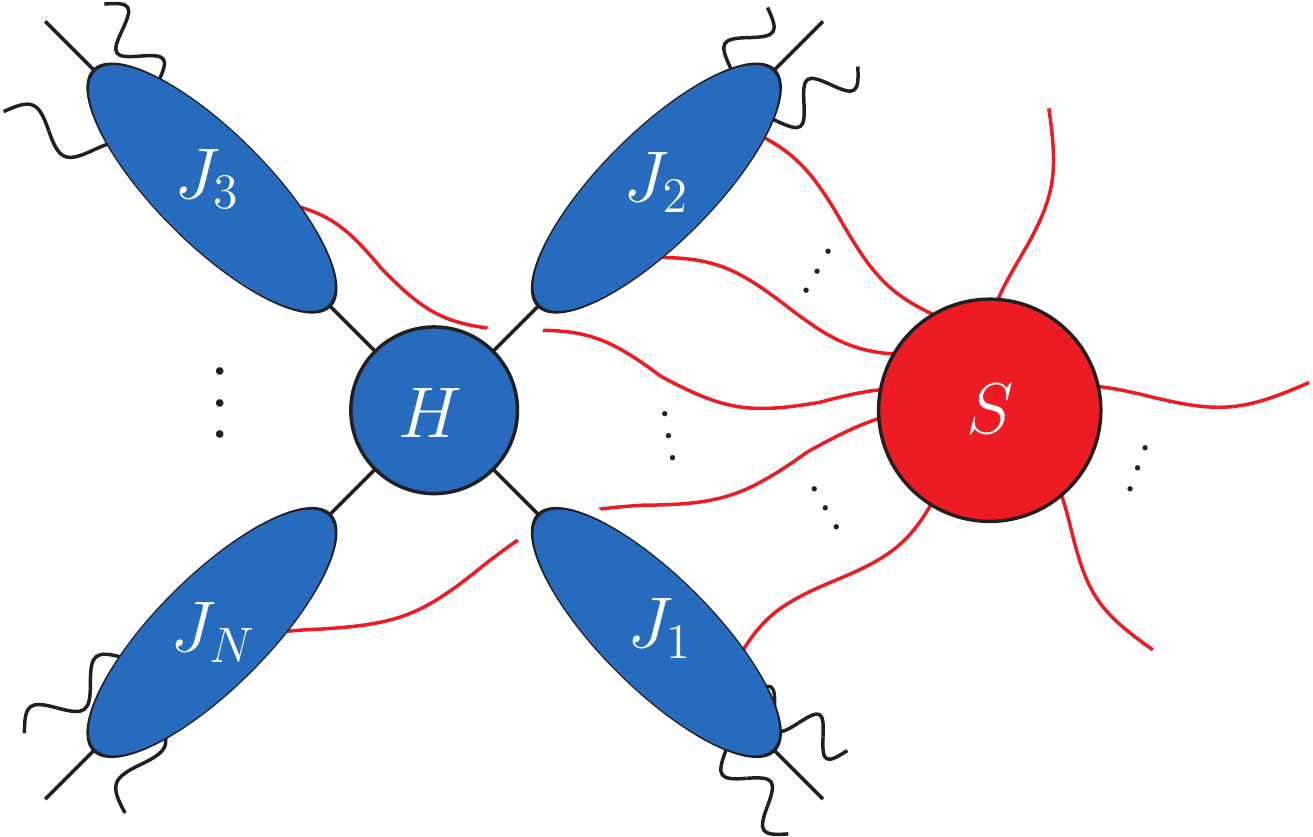} 
\label{pinches1}
\ee
We call the toplogy indicated on the right-hand side the {\bf reduced diagram}. It has the following properties:

\begin{itemize}

\item Each colored diagram in the sum  corresponds to a precise Feynman integral, with loop momenta integrated over all of $\mathbb{R}^{1,3}$. Note that our reduced diagrams are different from those used 
in~\cite{Sterman:1978bi,Libby:1978qf,Collins:1981ta}, which are pictures representing the pinch surface, not computable functions.

\item The ``jet'' amplitudes, labeled $J_\ccj$ are soft insensitive and collinear sensitive only in their own, $p_\ccj$ directions. That is, there are no $p_\ccj$-collinear sensitivities in the $J_\cci$ jet amplitudes for $\cci\neq \ccj$.

\item All soft sensitivity comes from virtual gluons in (or connecting to) the ``soft'' amplitude.

\item The  blue ball in the center is called the ``hard'' amplitude.  It is infrared insensitive (IR finite for any $\lambda$, and hence, independent of $\lambda$ at leading power). It only depends on the net collinear momenta coming in from each direction and no soft particles or red lines  connect to it. This property will establish that the Wilson coefficient in the factorization theorem is independent of the external state, as is expected in an operator product expansion.
\end{itemize}

\item Examine factorization gauge, which gives the flexibility needed for an efficient proof of soft-collinear decoupling. Although ghosts do not decouple completely, we show that they do not contribute new IR sensitivities and  do not affect the reduced diagram in \Eq{pinches1}.

\item Using factorization gauge, show that the soft gluons can be disentangled from the non-soft gluons. This step follows quite naturally from the proof of tree-level disentangling in~\tree. In the process, show that the factorized reduced diagrams are exactly reproduced by gauge-invariant matrix elements in the factorization formula.

\end{enumerate}

As with the 1-loop examples above, we will prove these steps in a more-or-less gauge-theory independent way, using QCD and scalar QED for examples. In this approach, technical details specific to QCD, such as color structures, become mostly notational. These are discussed in Section~\ref{sec:QCD}.

\section{Step 1: Coloring (separating soft sensitivities) }
\label{sec:coloring}

The first step is to separate the soft-sensitive physics from that which is soft-insensitive. As in the examples, we define soft-sensitive to mean either that a loop has a power-counting soft-divergence or that it would
have one for kinematic configurations corresponding to $\lambda=0$. 

Soft sensitivity is a property that
each virtual particle may have. We want to write each Feynman diagram as the sum of what we call {\bf colored diagrams} where the color of each virtual line in a colored diagram indicates if it is
soft sensitive or not. 
We have already seen examples of this separation at 1-loop: in Section~\ref{sec:1loopEx1} the soft-sensitive version of the graph $G^{(12)}$ in \Eq{G0full} was explicitly given as $G^{(12)}_\soft$ in \Eq{G0soft}, and it was shown that the not-soft-singular part, $G^{(12)}_\nsoft = G^{(12)} - G^{(12)}_\soft$, was  soft finite. The same was done with $G^{(12)a,b}$ in Section~\ref{sec:1loopEx2}. 

Beyond 1-loop,  it is not possible to split each diagram into one soft-sensitive and one soft-insensitive piece, since all of the loops are tangled up in a generic graph. More generally, we would like to expand in each virtual momenta. The only complication is that all the virtual momenta are not independent and
so the expansion has to be done iteratively. These iterations can be done algorithmically, starting from the
most soft-sensitive graphs, as we now explain. Section~\ref{sec:genalg} gives the algorithm,
which is perhaps easiest to understand through the examples in Sections~\ref{sec:algex1},~\ref{sec:algex2} and~\ref{sec:algex3}.

\subsection{Decomposition into colored diagrams}
\label{sec:genalg}
Consider sets $\Omega=\{\ell_1^\mu,\ell_2^\mu, \ldots\}$ of virtual momenta in a particular Feynman diagram $G$ which can all go to $\ell_i^\mu= 0$ simultaneously. 
For a given set $\Omega$, we can expand the integrand to leading order around $\ell_i^\mu=0$ for all the $\ell_i^\mu \in \Omega$ simultaneously. We want to
do this very carefully, dropping only terms which {\it must} be small when  $\ell_i^\mu=0$. For example, if $p^\mu$ is an external collinear momentum, then we can drop $l_i^2$ compared
to $l_i\cdot p$. We do not want to drop $l_i^\mu$ compared to any external soft momentum, or to any other virtual momentum $\ell_k^\mu$
which go soft simultaneously with $\ell_i^\mu$. 
We also drop $l_i \cdot p_\ccj$ compared to $(p_\ccO + p_\ccT)^2$ for
two collinear momenta $p_\ccO^\mu$ and $p_\ccT^\mu$ if and only if $p_\ccO^\mu$ and $p_\ccT^\mu$ are in different collinear sectors.
If they are in the same sector then we allow that $(p_\ccO + p_\ccT)^2\sim \lambda^2$ can be 
arbitrarily small.

 Let us call the leading term in the expansion according to this procedure the {\bf soft limit} of the set $\Omega$ in $G$ and
denote it by $G_{\SL(\Omega)}$.  The soft limit defined in this way allows us to see if a set $\Omega$ is soft-sensitive simply by looking at the scaling of $G_{\SL(\Omega)}$ (or equivalently of $G$)  under $\ell_i^\mu \to \kappa^2 \ell_i^\mu$ for all $\ell_i^\mu \in \Omega$. By not dropping soft momenta compared to terms which could possibly vanish for certain external momenta, we are effectively taking the leading power of $\kappa$ at $\lambda=0$.
Taking the soft limit in this way implies that 
\be
\lim_{\Omega \to \text{soft}} \, G \;=\;  \lim_{\Omega \to \text{soft}} \, G_{\SL(\Omega)}
\label{LimOmegaSoft}
\ee
so that $G-G_{\SL(\Omega)}$ is automatically less-singular than $G$ in the limit that all the $\ell_i^\mu \in \Omega$ go soft. The limit in \Eq{LimOmegaSoft} means restricting the integration regions to balls around the point where each momenta in $\Omega$ vanish and taking the limit where those balls have vanishing size. The point of taking the soft limit $\SL(\Omega)$ is that, since infrared divergences in gauge theories are at most logarithmic (at least in physical gauges, as we will show in the \emph{Log Lemma} (Lemma~\ref{lem:kappa0})), the difference $G-G_{\SL(\Omega)}$ cannot be soft sensitive in this $\Omega\to\,$soft limit.

That all the momenta in a set $\Omega$ can go soft together does not imply that $G$ is soft sensitive in this limit.
Let $\{\Omega_i\}$ enumerate all the possible sets $\Omega$ which {\it do} have a soft sensitivity in their simultaneous soft limit. Note that which sets are in
$\{\Omega_i\}$ is gauge-dependent, and we will be concerned primarily with $\Omega_i$ in generic-lightcone gauge.
Consider first the largest sets $\{\Omega^i_\text{max}\}$, defined as those sets, $\Omega_i$, which are not proper subsets of any other $\Omega_i$'s.
Now take the soft limit and define
\be
G_{\Omega^i_\text{max}}
\equiv 
G_{\SL(\Omega^i_\text{max})}
\ee
Here, $G_{\Omega^i_\text{max}}$ refers to a particular integral, for each $i$, derived form an expansion of the integrand of the original Feynman diagram integral, $G$. We represent it as a diagram with the same topology as $G$ in which we color all the lines in $\Omega^i_\text{max}$ red and color blue all the lines not in $\Omega^i_\text{max}$. The blue lines cannot give rise to a soft singularity because we have already taken the maximal soft limit in $G_{\Omega^i_\text{max}}$ by construction (this will be shown in Lemma~\ref{SoftInsens} below).

Next, take the sets, $\{\Omega^j_\text{next}\}$, defined as being the next largest proper
 subsets of any of the $\Omega^i_\text{max}$'s  whose simultaneous soft limit engenders a soft sensitivity. Each $\Omega^j_\text{next}$ may be a subset of multiple $\Omega^i_\text{max}$.  
Then subtract off from the soft limit of $\Omega^j_\text{next}$ all
 of the $G_{\Omega^i_\text{max}}$ for which it is a subset:
\be
G_{\Omega^j_\text{next}} \equiv \bigg( G - 
\sum_{\{ i\,;\, \Omega^i_\text{max} \supsetneq \Omega^j_\text{next} \}} 
G_{\Omega^i_\text{max}} 
\bigg)_{\SL({\Omega_\text{next}^j})}
\label{SoftSubtract}
\ee
As before, we represent $G_{\Omega^j_\text{next}} $ as a diagram with the lines in $\Omega^j_\text{next}$ colored red, and all other lines colored
blue to show that they cannot give rise to a soft sensitivity due to the subtraction.

This procedure can be iterated, with subsets of $\Omega^j_\text{next}$ and so on. In each step, we take subsets, $\Omega^j_\text{step}$, of the $\Omega^i_\text{max}$'s of a given size and subtract off $G_\Omega$ for every subset, $\Omega$, of the $\Omega^i_\text{max}$'s for which $\Omega^j_\text{step}$ is a subset:
\be
G_{\Omega^j_\text{step}} \equiv \bigg(G - \sum_{ \Omega \supsetneq \Omega^j_\text{step}} G_\Omega \bigg)_{\SL(\Omega^j_\text{step})}
\ee
Eventually, all of the possible sets of soft-singular lines are exhausted. In particular, in the last step, $\Omega_\text{last}$ is the empty set. This is a subset of all the other sets, so we have
\be
G = G_{\text{last}} + \sum_{\Omega} G_{\Omega} 
\ee
At every stage $G_\Omega$ is drawn as the graph $G$ but with the lines in $\Omega$ colored red and those not in $\Omega$ 
 colored blue. 
Thus the full graph becomes the sum of colored graphs.

After this procedure, each colored graph represents a particular integral which can have a soft singularity or soft-sensitivity \emph{only} when any of the red lines become soft, but never when any of the blue 
lines become soft. In other words:

\vspace{2mm}
\begin{lemma} \ltag{Soft-insensitivity Lemma}
Soft sensitivities cannot come from the soft region of any set of blue lines.
\label{SoftInsens}
\end{lemma}

\begin{proof}

We prove this by induction on the number of blue lines in a colored graph, $G_\Omega$. 
The first step is to show the result for graphs with the fewest number of blue lines, namely $G_{\Omega_\text{max}}$. Indeed, the only way for a line, $\ell_\text{blue} \notin \Omega_{\text{max}}$, to be able to give a soft sensitivity in $G$ but not in the simultaneous limit $\Omega_{\text{max}}\cup\ell_\text{blue} \to \text{soft}$ is if the limit is forbidden by momentum conservation. But then $\lim_{\ell_\text{blue} \to \text{soft}} G_{\Omega_{\text{max}}}$ will vanish since the limit where $\Omega_{\text{max}} \to \text{soft}$ has already been taken. So the lemma holds for graphs with the least number of blue lines, $G_{\Omega_{\text{max}}}$.

Now, suppose it is true for any colored graph with $n$ or fewer blue lines and consider a colored graph with $n+1$ blue lines, $G_\Omega$. Now consider the most general limit where some subset, $\omega$, of blue lines goes soft. We must show that $\lim_{\omega \to \text{soft}} G_\Omega$ is finite.

By definition
\be
G_\Omega = \bigg(G - \sum_{\Upsilon \supsetneq \Omega} G_{\Upsilon} \bigg)_{\SL(\Omega)}
= G_{\SL(\Omega)} 
- \sum_{\Upsilon \supsetneq \Omega, \; \omega\subseteq\Upsilon} 
		\big(G_{\Upsilon}\big)_{\SL(\Omega)}
- \sum_{\Upsilon \supsetneq \Omega, \; \omega\Sl\subseteq\Upsilon} 
		\big(G_{\Upsilon}\big)_{\SL(\Omega)}
\ee
where the sets $\Upsilon$ are soft-sensitive sets.
In the $\omega \to \text{soft}$ limit, the last term would involve the soft limit of at least one blue line in a colored graph with $n$ or fewer blue lines which must be finite by the induction hypothesis combined with the fact that
\be
\lim_{\omega \to \text{soft}}\big(G_{\Upsilon}\big)_{\SL(\Omega)}
=
\Big(\big(G_{\Upsilon}\big)_{\SL(\Omega)}\Big)_{\SL(\omega)}
=
\Big(\big(G_{\Upsilon}\big)_{\SL(\Omega\cup\omega)}\Big)_{\SL(\omega)}
\label{CombiningLimits}
\ee
Therefore, the soft limit we are interested in simplifies to
\be
\lim_{\omega \to \text{soft}} G_\Omega = 
\lim_{\omega \to \text{soft}} \bigg[
G_{\SL(\Omega)}
- \sum_{\Upsilon \supseteq \Omega\cup\omega} \big(G_{\Upsilon}\big)_{\SL(\Omega)}
\bigg]
+ \text{finite}
\label{blueproof1}
\ee

Now, if $\Omega\cup\omega \,\Sl{\subseteq}\, \Omega^i_\text{max}$ for some $i$, the term in square brackets in \Eq{blueproof1} is finite because, in that case, the sum is empty and the soft limit of $\Omega$ followed by $\omega$ does not give rise to a soft sensitivity in the first term by momentum conservation (the same argument given in the first-induction step). If \Eq{blueproof1} is finite, we are done the proof, so assume $\Omega\cup\omega \subseteq \Omega^i_\text{max}$ for some $i$. Consequently, there exists a soft-sensitive set $\Gamma$ that is the next smallest set containing $\Omega\cup\omega$ for which $\SL(\Omega\cup\omega) = \SL(\Gamma)$. 
Therefore, using \Eq{CombiningLimits}, we have
\begin{align}
\lim_{\omega \to \text{soft}} G_\Omega &= \lim_{\omega \to \text{soft}}
\bigg(
G
 - G_{\Gamma}
- \sum_{\Upsilon \supseteq \Omega\cup\omega, \; \Upsilon \neq \Gamma} G_{\Upsilon}
\bigg)_{\SL(\Gamma)}
+ \text{finite} \\
&\overset{\text{def}}= \lim_{\omega \to \text{soft}} \bigg[
G_{\SL(\Gamma)} 
- G_{\SL(\Gamma)} + \sum_{\Upsilon\supsetneq\Gamma} \big(G_\Upsilon\big)_{\SL(\Gamma)}
- \sum_{\Upsilon \supseteq \Omega\cup\omega, \; \Upsilon \neq \Gamma} \big(G_{\Upsilon}\big)_{\SL(\Gamma)}\bigg]
+ \text{finite}
\end{align}
Now we can split the last sum into
\be
\sum_{\Upsilon \supseteq \Omega\cup\omega, \; \Upsilon \neq \Gamma} G_{\Upsilon}
=
\sum_{\Upsilon \supsetneq \Gamma} G_{\Upsilon}
+
\sum_{\Upsilon \supseteq \Omega\cup\omega, \; \Upsilon \Sl\supseteq \Gamma} G_{\Upsilon}
\ee
Then, canceling the first four terms we are left with
\be
\lim_{\omega \to \text{soft}}  G_\Omega =
-\lim_{\omega \to \text{soft}}  \sum_{\Upsilon \supseteq \Omega\cup\omega, \; \Upsilon \Sl\supseteq \Gamma} \big(G_{\Upsilon}\big)_{\SL(\Gamma)}
+ \text{finite}
\ee

Finally, either $\Gamma = \Omega\cup\omega$ in which case the above sum is empty and $\lim_{\omega \to \text{soft}}  G_\Omega $ is finite, or the $\omega\to\text{soft}$ limit  forces other lines in $\Gamma\setminus (\Omega\cup\omega)$ to go soft along with those in $\omega$. The latter case means that for every term in the above sum, $\lim_{\omega \to \text{soft}}G_{\Upsilon}$ involves taking a blue line soft which gives a finite result by the induction hypothesis. Thus, $\lim_{\omega \to \text{soft}}G_\Omega$ is always finite.
\end{proof}

This algorithm may make more sense after a few explicit examples. We have already seen how to separate the soft-sensitive and soft-insensitive parts of graphs at 1-loop order in Sections~\ref{sec:1loopEx1} and \ref{sec:1loopEx2}, so we move directly to the more complicated 2-loop examples. The first two examples in Sections~\ref{sec:algex1} and \ref{sec:algex2} outline the basics of the coloring algorithm, having only a single maximal soft-sensitive set. The example in Section~\ref{sec:algex3} has multiple $\Omega^i_\text{max}$'s as well as a discussion about symmetry factors of the colored graphs.

It is also worth pointing out that this separation into red and blue lines is similar to the zero-bin subtraction discussed in \cite{Manohar:2006nz}. Our blue lines correspond to the propagation of degrees of freedom that can be collinear sensitive but cannot be soft sensitive. This is implemented by recursively subtracting off the soft-sensitive limits from the full-theory graphs. In SCET, collinear fields are defined by summing over discrete labels on momentum space with the label pointing to zero momentum -- known as the zero bin -- removed. In practice the discrete sum is always turned into an integral and the zero bin is subtracted off. This procedure calls for a soft subtraction for every single collinear line, irrespective of whether or not the line is soft sensitive, but otherwise is similar to our subtraction for the blue lines. Therefore, the SCET-familiar reader could think of our blue lines as a cleaner version of the collinear lines of SCET. In any case, our blue lines are still too complicated 
to use in practice; by the end, our factorization theorem will be formulated entirely in terms of full-theory Feynman rules with the subtraction procedure implemented by dividing by simple matrix elements of Wilson lines.

In a colored diagram, every line is either soft sensitive (red) or soft insensitive (blue). We sometimes draw soft-insensitive lines as black lines if no expansion is done (for example with external lines).
All black lines in the following should technically be drawn blue.

\subsection{Example one: Tangled 2-loop}
\label{sec:algex1}

Consider the following graph in scalar QED:
\be
\fd{1.9cm}{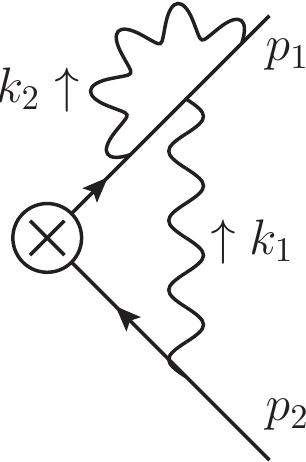} \equiv G= \int 
\frac{(2p_\ccT+k_1)\cdot\Pi(k_1)\cdot (2p_\ccO-2k_2-k_1) \, (2p_\ccO-2k_1-k_2)\cdot\Pi(k_2)\cdot(2p_\ccO-k_2)}
	{k_1^2(p_\ccT+k_1)^2(p_\ccO-k_1)^2\;k_2^2(p_\ccO-k_2)^2(p_\ccO-k_1-k_2)^2}
\ee
where we have dropped constant prefactors and the integration measure, $d^4 k_1 d^4 k_2$ is left implicit.
In Feynman gauge (or other covariant gauges), the gauge-dependent $\Pi(k_i)$ factors count as order 1. Then, this graph has a soft singularity when both photons go soft, or when either one goes soft and the other goes collinear. Note that the virtual scalars can never give rise to a soft sensitivity by helicity conservation, which can easily be checked by power counting, say, the $(p_\ccO-k_2) \to $ soft limit.

Our first step is to write down the soft-singular graph with the most soft lines. This is done by expanding the integrand as if both virtual-photon momenta $k_1$ and $k_2$ were soft, giving: 
\be
G_{\Omega_\text{max}}  = \fd{1.5cm}{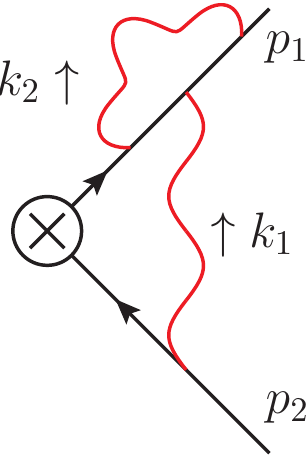}= \int 
\frac{p_\ccT\cdot\Pi(k_1)\cdot p_\ccO}{k_1^2 p_\ccT\cdot k_1 (-p_\ccO\cdot k_1)}
\times
\frac{ p_\ccO\cdot\Pi(k_2)\cdot p_\ccO}{ k_2^2 (-p_\ccO\cdot k_2) (-p_\ccO\cdot (k_1+k_2))}
\label{GSmax}
\ee
Note that we have not dropped either soft momentum with respect to the other. Also, $G_{\Omega_\text{max}}$ is clearly soft divergent when both $k_1$ and $k_2$ vanish.

Now we would like to write down the part of $G$ that is soft divergent when only one of the photons goes soft  (and the other goes collinear). To do this, we expand one of the virtual momentum as if it were soft and leave the other one general. That is, for $k_1$ soft we have
\be
G_{\Omega_\text{next}^1} = \int 
\frac{p_\ccT\cdot\Pi(k_1)\cdot (p_\ccO-k_2)}{k_1^2 p_\ccT\cdot k_1 (- p_\ccO\cdot k_1)}
\times
\frac{ (2p_\ccO-k_2)\cdot\Pi(k_2)\cdot(2p_\ccO-k_2)}{ k_2^2(p_\ccO-k_2)^2(p_\ccO-k_2)^2}
	- \big(G_{\Omega_\text{max}}\big)_{\SL(k_1)}
\label{GS1}
\ee
With this definition, $G_{\Omega_\text{next}^1}$ is clearly finite when $k_2$ goes soft because we have subtracted that limit off in the form of $\big(G_{\Omega_\text{max}}\big)_{\SL(k_1)}$. Similarly, we define the $k_2$-soft-singular graph as
\be	
G_{\Omega_\text{next}^2} = \int 
\frac{(2p_\ccT+k_1)\cdot\Pi(k_1)\cdot (2p_\ccO-k_1)}{k_1^2(p_\ccT+k_1)^2(p_\ccO-k_1)^2}
\times
\frac{(p_\ccO-k_1)\cdot\Pi(k_2)\cdot 2p_\ccO}
	{k_2^2(-p_\ccO\cdot k_2)  (p_\ccO-k_1)^2}
	- \big(G_{\Omega_\text{max}}\big)_{\SL(k_2)}
\label{GS2}
\ee
which is, again, finite in the limit where $k_1$ goes soft because of the subtraction.

Finally, we have the remainder of the graph, given by 
\be
G_\text{last} = G - G_{\Omega_\text{max}} - G_{\Omega_\text{next}^1} - G_{\Omega_\text{next}^2}
\label{Gsoftreg}
\ee
It is easy to see that $G_\text{last}$ is finite in any limit $\omega\to$ soft for $\omega \subseteq \{k_1,k_2\}$, for example
\be
\lim_{k_1\to \text{ soft}} G_\text{last} = \big(G - G_{\Omega_\text{max}}\big)_{\SL(k_1)} - \big(G_{\Omega_\text{next}^1}\big)_{\SL(k_1)} + \text{finite}
= G_{\Omega_\text{next}^1} - G_{\Omega_\text{next}^1} + \text{finite}= \text{finite}
\ee
where we used the definition of $G_{\Omega_\text{next}^1}$, that $(G_{\Omega_\text{next}^1})_{\SL(k_1)} = G_{\Omega_\text{next}^1}$ and that $(G_{\Omega_\text{next}^2})_{\SL(k_1)}$ is finite.

We can now draw these four integrals as separate graphs by denoting which internal lines are taken soft by a longer-wavelength red line and the other lines that are made soft-insensitive by the subtraction are drawn blue. That is, 
\be
G_{\Omega_\text{max}} = \fd{1.6cm}{EikEx1_1.pdf}
\quad
G_{\Omega_\text{next}^1} = \fd{1.6cm}{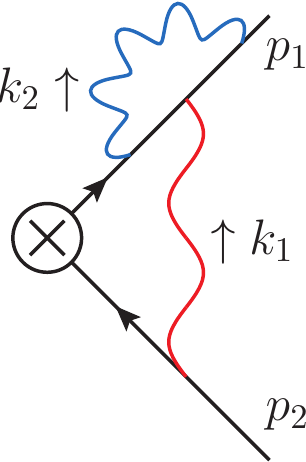}
\quad
G_{\Omega_\text{next}^2} = \fd{1.6cm}{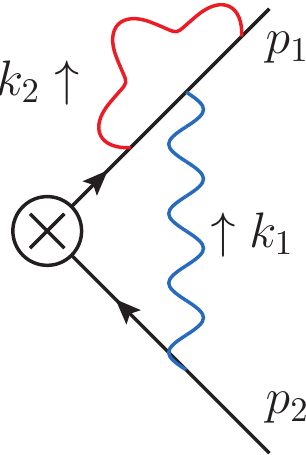}
\quad
G_\text{last} = \fd{1.6cm}{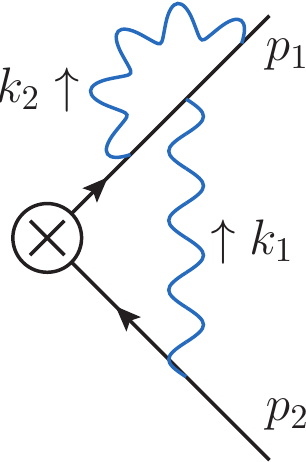} 
\quad
\binom{ \text{covariant} }{ \text{gauges} }
\label{covdec1}
\ee
and the sum of these four graphs is trivially equal to the original graph, $G$.

We reiterate that in these modified graphs, only the red, long-wavelength lines can have soft singularities. Each blue line is made soft insensitive by subtracting from the original graph all of the graphs with that line red. In our example, $G_{\Omega_\text{max}}$ was subtracted off in \Eq{GS1} and \Eq{GS2} to ensure that the blue line in both $G_{\Omega_1}$ and $G_{\Omega_2}$ is soft insensitive and all three of $G_{\Omega_\text{max}}$, $G_{\Omega_1}$ and $G_{\Omega_2}$ were subtracted off in \Eq{Gsoftreg} in order to make both of the blue lines in $G_\text{last}$ soft insensitive.

In deriving the decomposition in \Eq{covdec1}, no scaling of the numerators was used. Thus this decomposition holds in covariant gauges, such as Feynman gauge, where there is no extra numerator suppression.
In physical gauges, such as generic-lightcone gauge, the set of colored graphs is different. As will be discussed in detail in Section \ref{sec:Step2} in a physical gauge, there is no singularity when $k_2$ goes soft and $k_1$ does not, so $\Omega_2$ is not a possible set with a soft sensitivity. Thus, in a physical gauge, $G_{\Omega_\text{max}}$ and $G_{\Omega_1}$ are defined as above and $G_\text{last} = G - G_{\Omega_\text{max}} - G_{\Omega_1}$. So, the colored-graph decomposition of $G$ in a physical gauge is given by the sum of only three graphs:
\be
G_{\Omega_\text{max}} = \fd{1.6cm}{EikEx1_1.pdf}
\qquad
G_{\Omega_\text{next}^1} = \fd{1.6cm}{EikEx1_2.pdf}
\qquad
G_\text{last} = \fd{1.6cm}{EikEx1_4.pdf} 
\qquad
\binom{ \text{physical} }{ \text{gauges} }
\ee

\subsection{Example two: 2 loops, 3 gluons}
\label{sec:algex2}

Consider now a slightly more complicated example, the QCD graph:
\be
\fd{2cm}{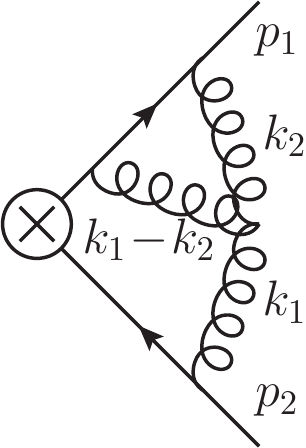} \equiv H
\ee
For this graph, 
 when all three gluons go soft, there are 9 powers of soft momenta in the denominator from the propagators, 1 in the numerator from the 3-point vertex, and 8 from the $d^4 k_1 d^4 k_2$  integration measure. The result is an overall logarithmic divergence (in covariant or physical gauges). This is the soft singularity with the highest number propagators that are simultaneously going soft.

Thus the soft-singular graph with the largest number of soft propagators in it is
\be
H_{\Omega_\text{max}} =H_{{\SL (\{k_1,k_2\} ) }}=\; \fd{1.6cm}{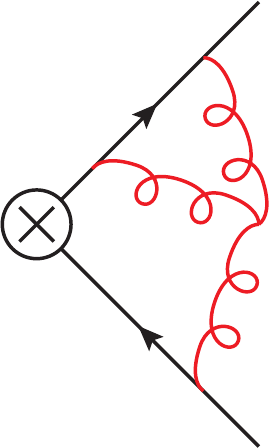}
\ee
The algebraic expression for $H_{\Omega_\text{max}}$ is found by taking the integrand of $H$ and expanding as if $k_1$ and $k_2$ were soft but of the same order, as was done in \Eq{GSmax}.

There are no singularities with only two gluons going soft since momentum conservation will not allow two of the gluons to go soft without the third being soft as well. Thus, the soft-singular configurations with the next largest number of soft internal lines are those with one of the gluons going soft.
In covariant gauges there is a singularity when any of the gluons go soft
\be
H_{\Omega_\text{next}^1} =\; \fd{1.6cm}{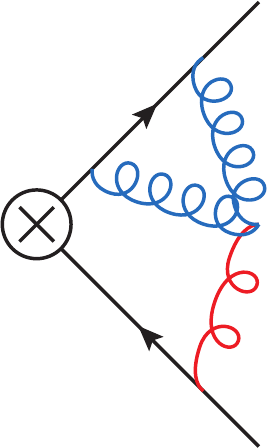}
\qquad
H_{\Omega_\text{next}^2} =\; \fd{1.6cm}{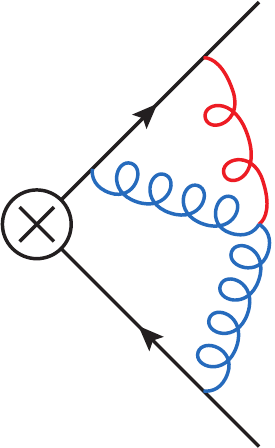}
\qquad
H_{\Omega_\text{next}^3} =\; \fd{1.6cm}{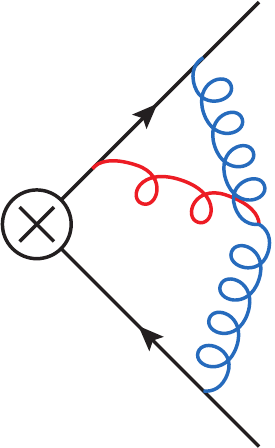}
\quad
\quad
\binom{ \text{covariant} }{ \text{gauges} }
\ee
and their algebraic expressions are given by taking the soft limit of one of the gluons and subtracting off $(H_{\Omega_\text{max}})_{\SL(k_i)}$ to ensure that the other gluons cannot be soft singular. That is
\be
H_{\Omega_\text{next}^i} = \big(H - H_{\Omega_\text{max}}\big) _{\SL(k_i)}
, \quad i=1,2
\quad\And\quad
H_{\Omega_\text{next}^3} = \big(H - H_{\Omega_\text{max}}\big)_{\SL(k_1-k_2)}
\ee
Finally, the soft-insensitive graph is given by
\be
\fd{1.6cm}{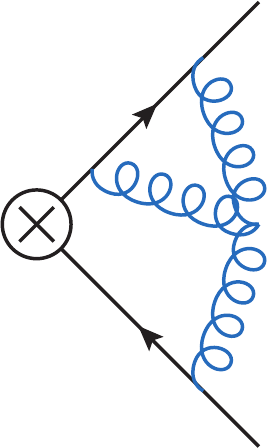} \;= H_\text{last} = H - H_{\Omega_\text{max}} - H_{\Omega_\text{next}^1} - H_{\Omega_\text{next}^2} - H_{\Omega_\text{next}^3}
\quad
\binom{ \text{covariant} }{ \text{gauges} }
\ee
Thus we have the decomposition
\be
\fd{1.8cm}{EikEx2.pdf} \;=\; \fd{1.5cm}{EikEx2_4.pdf} \;+\; \fd{1.5cm}{EikEx2_1.pdf} \;+\; \fd{1.5cm}{EikEx2_2.pdf} \;+\; \fd{1.5cm}{EikEx2_3.pdf} \;+\; \fd{1.5cm}{EikEx2_5.pdf}
\qquad
\binom{ \text{covariant} }{ \text{gauges} }
\ee
Every graph has its soft sensitivities manifest, since none of the blue lines admit a soft sensitivity by construction.

We will see in Section \ref{sec:Step2} that in physical gauges
$\Omega_\text{next}^2$ and $\Omega_\text{next}^3$ are soft insensitive. Thus, $H_{\Omega_\text{max}}$ and $H_{\Omega_\text{next}^1}$ are defined as above, but $H_{\Omega_\text{next}^2}$ and $H_{\Omega_\text{next}^3}$ do not exist, thereby modifying the definition of $H_\text{last}$ to $H_\text{last} = H - H_{\Omega_\text{max}} - H_{\Omega_\text{next}^1}$. The colored diagram expansion in physical gauges is then:
\be
\fd{1.8cm}{EikEx2.pdf} \;=\; \fd{1.5cm}{EikEx2_4.pdf} \;+\; \fd{1.5cm}{EikEx2_1.pdf} \;+\; \fd{1.5cm}{EikEx2_5.pdf}
\qquad
\binom{ \text{physical} }{ \text{gauges} }
\ee

\subsection{Example three: soft-gluon decoherence}
\label{sec:algex3}

For our final example, we consider a graph that does not have a unique maximal set of soft lines that contribute to a soft sensitivity:
\be
\fd{4cm}{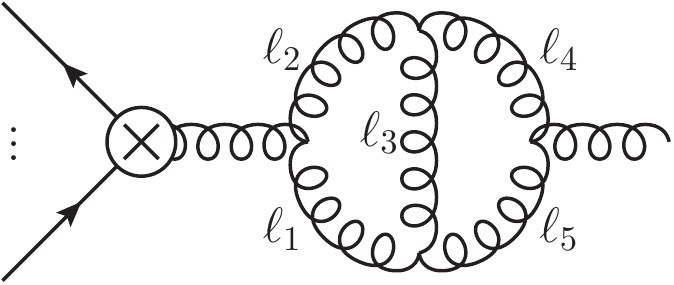} \; \equiv I
\label{eggbeater}
\ee
Due to momentum conservation, there is no way for all the gluons to go soft in the loops; at least a single continuous line of non-soft momentum must flow through the graph. This means that there are multiple maximally soft-sensitive sets of different sizes.

First we define the soft graphs with the maximal sets of soft-sensitive lines:
\begin{align}
I_{\Omega_\text{max}^1} &=\; \fd{3.4cm}{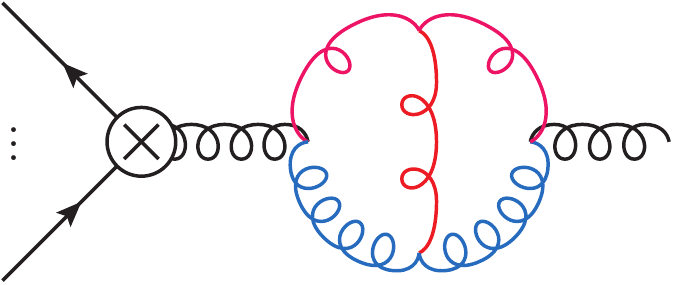} \; = I_{\SL( \{\ell_2,\ell_3,\ell_4\}) }\\
I_{\Omega_\text{max}^2} &=\; \fd{3.4cm}{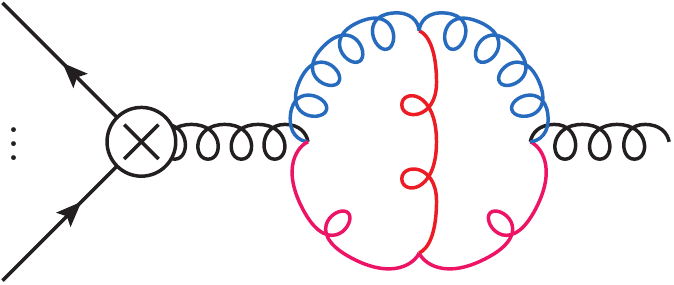} \; = I_{\SL( \{\ell_1,\ell_3,\ell_5\}) }\\
I_{\Omega_\text{max}^3} &=\; \fd{3.4cm}{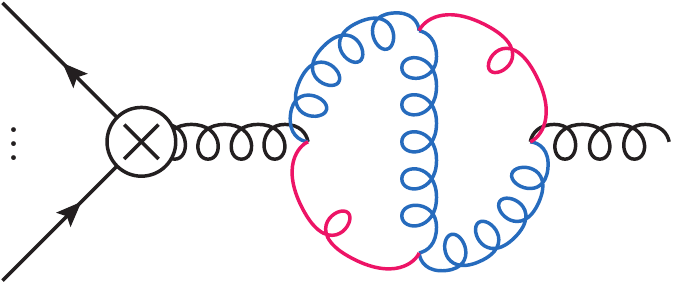} \; = I_{\SL( \{\ell_1,\ell_4\}) } \\
I_{\Omega_\text{max}^4} &=\; \fd{3.4cm}{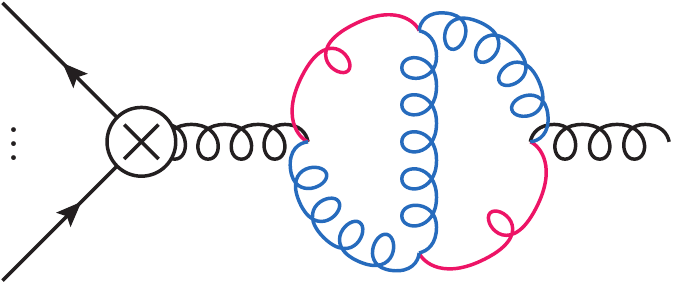} \; = I_{\SL( \{\ell_2,\ell_5\}) }	
\end{align}
The algebraic expressions for these graphs are found by taking the soft limit of the relevant virtual momenta in $I$.
Note that although no subtraction is performed,
none of the blue lines can give rise to soft sensitivities due to momentum conservation.
Although $I_{\Omega_\text{max}^1} = I_{\Omega_\text{max}^2}$ and $I_{\Omega_\text{max}^3} = I_{\Omega_\text{max}^4}$, these graphs are generated by
expanding in different non-overlapping regions of the virtual momentum phase space in the original integral, $I$. Thus they correspond to separate colored graphs.
This separation foreshadows the separation of QCD gluons into soft (red) and collinear (blue) gluons in the factorized expression.

Now, take the next largest subsets that admit a soft sensitivity, $\Omega_\text{next}^j$, and define the corresponding colored graph via the subtraction procedure. In every case, the sets $\Omega_\text{next}^j = \{\ell_j\}$ have a single soft line: 
\be
I_{\Omega^j_\text{next}} = \bigg(I - \sum_{i \,;\, \ell_j\in\Omega_\text{max}^i } I_{\Omega_\text{max}^i}\bigg)_{\SL(\ell_j)}
\ee
and define the last graph as
\be
I_{\text{last}} = I - \sum_{j=1}^5I_{\Omega^j_\text{next}}  - \sum_{i=1}^4I_{\Omega_\text{max}^i}
\ee
We draw these graphs by coloring every line that has a soft limit taken red and the other lines 
blue:
\begin{align}
I_{\Omega^1_\text{next}} &= \fd{3.4cm}{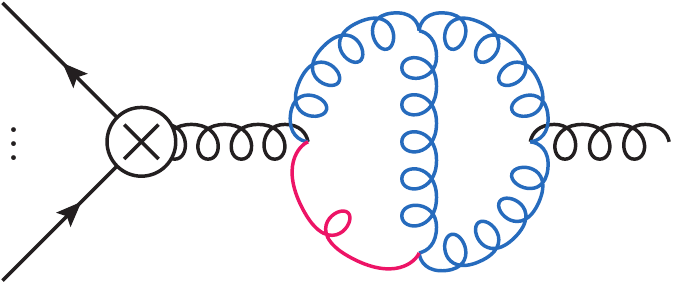} 
&I_{\Omega^2_\text{next}} = \fd{3.4cm}{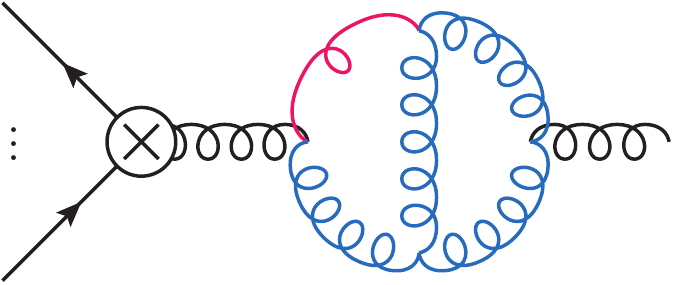} \\
I_{\Omega^5_\text{next}} &= \fd{3.4cm}{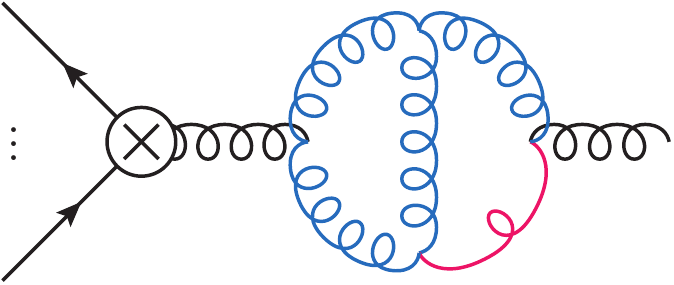} 
&I_{\Omega^4_\text{next}} = \fd{3.4cm}{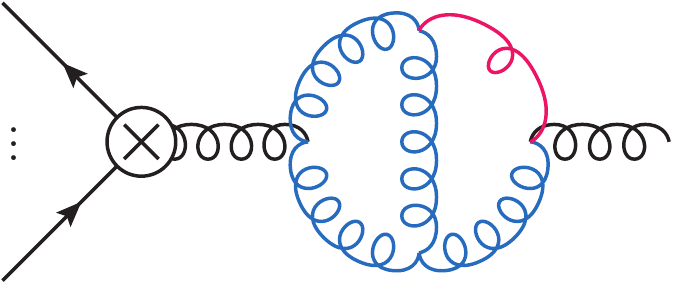} \\
I_{\Omega^3_\text{next}} &= \fd{3.4cm}{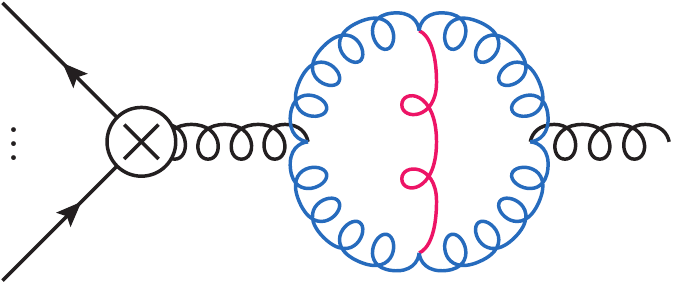} 
&I_{\text{last}} = \fd{3.4cm}{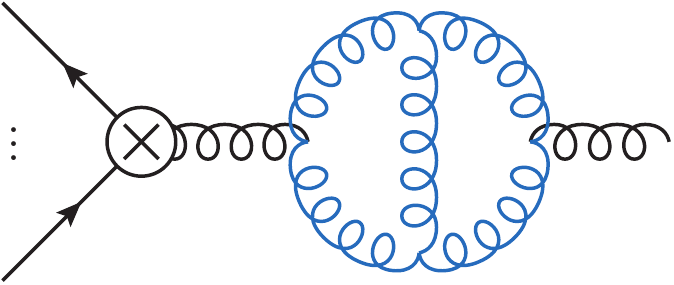}
\end{align}
The blue lines either have a soft subtraction or are soft finite by momentum conservation.

It is easy to check that no blue lines can give rise to a soft sensitivity. To be explicit, we check that this is the case for $I_\text{last}$ in the limit where $\ell_1$ goes soft. First note that only $I_{\Omega^1_\text{next}}$, $ I_{\Omega_\text{max}^2} $ and $I_{\Omega_\text{max}^3} $ can have a soft singularity in the $\ell_1\to0$ limit because only these graphs have a red $\ell_1$-line. Thus,
\be
\begin{aligned}
\lim_{\ell_1\to0} I_\text{last} &= \lim_{\ell_1\to0} \Big[  
I - I_{\Omega^1_\text{next}} - I_{\Omega_\text{max}^2} - I_{\Omega_\text{max}^3}
\Big] + \text{finite}
\nn
\\&=  
I_{\SL(\ell_1)} - \big(I - I_{\Omega_\text{max}^2} - I_{\Omega_\text{max}^3} \big)_{\SL(\ell_1)} - \big(I_{\Omega_\text{max}^2}\big)_{\SL(\ell_1)} - \big(I_{\Omega_\text{max}^3}\big)_{\SL(\ell_1)}
 + \text{finite}
\nn
\\&=\text{finite} 
\end{aligned}
\ee

Finally, note that all of the colored graphs in the decomposition of $I$ are equal to another colored graph except for $I_{\Omega^3_\text{next}}$ and $I_\text{last}$. That is,
\begin{align}
I &= \sum_{i=1}^4I_{\Omega_\text{max}^i} + \sum_{j=1}^5I_{\Omega^j_\text{next}} + I_\text{last} 
\nn
\\
&= \,2\times 
\fd{3.2cm}{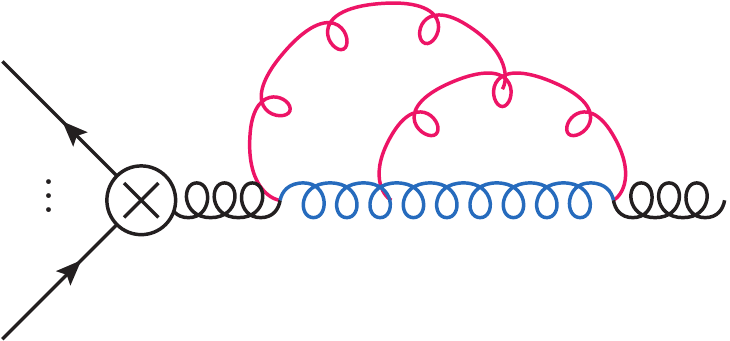}
\;+\;
2\times \fd{3.2cm}{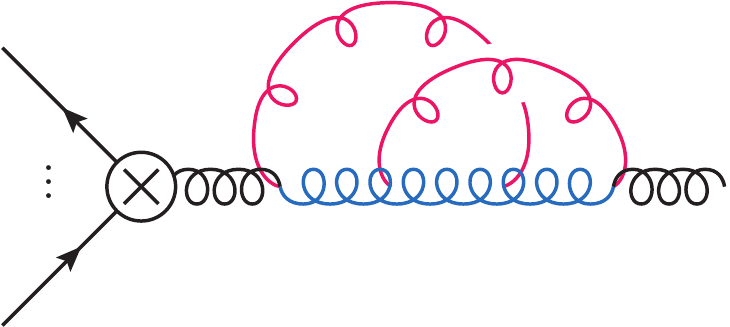}
\;+\;
2\times \fd{3.2cm}{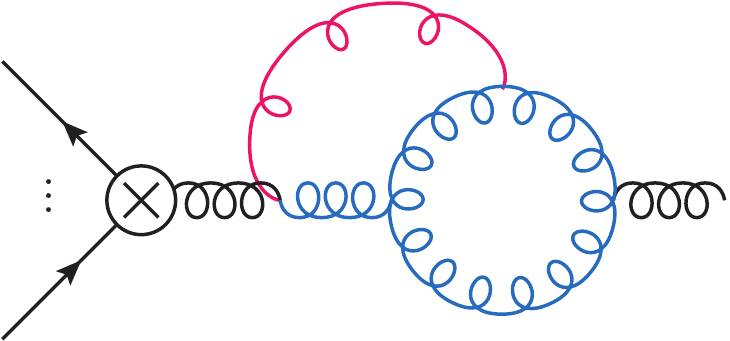}
\nn
\\
&\;+\;
2\times \fd{3.2cm}{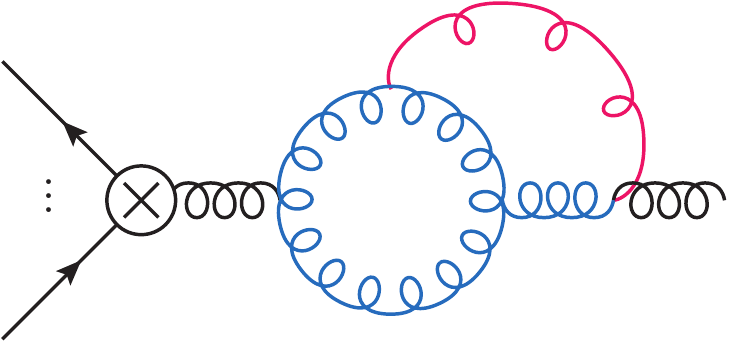} 
\;+\;
\fd{3.1cm}{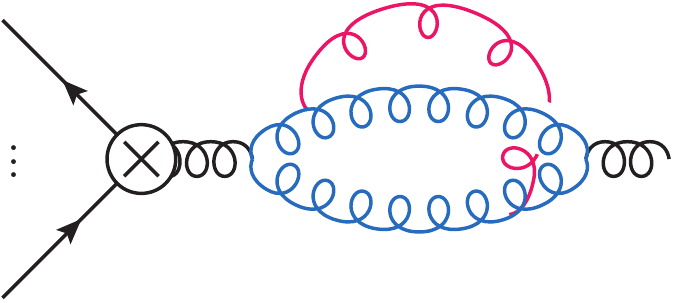} 
\;+\; 
\fd{3.1cm}{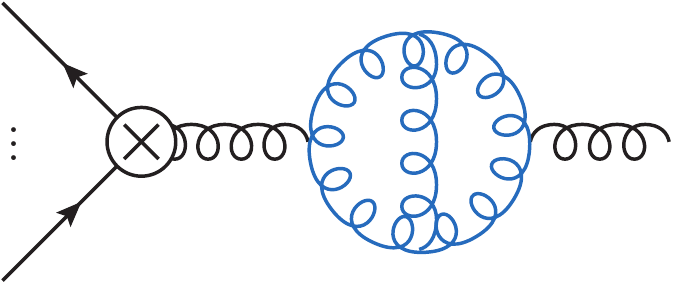}
\end{align}
In the graphs that are doubled, the coloring breaks the $\mathbb{Z}_2$ symmetry of the original graph, $I$. Because of this symmetry
$I$ gets a symmetry factor of $1/2$. In the graphs where the coloring breaks the symmetry, the factors of 2 directly cancel this factor
of $1/2$. In the graphs where the coloring preserves the symmetry, no factor of $2$ results and the original symmetry factor of $I$ is preserved. 
Thus, the final integrals have exactly the symmetry factor corresponding to the symmetries of the colored graphs. It is easy to see that this happens quite generally,
as expected in an effective theory where the red and blue lines are distinguishable particles.

\section{Step 2:  Reduced diagrams}
\label{sec:Step2}
At this point, we have a procedure for writing any Feynman graph as a sum of graphs each of which has all its lines marked as either soft-sensitive (red) or soft-insensitive (blue). As discussed in some of the
examples, the coloring is gauge-dependent. The coloring also does not indicate if a graph is collinear-sensitive. 
In this section we prove a set of lemmas that determine which graphs can be soft or collinear sensitive.
The lemmas in Section~\ref{sec:lemmas} are very general. They apply to QCD Feynman diagrams, independent of the coloring. Conclusions about collinear sensitivity, for example, apply equally well to soft-sensitive and soft-insensitive lines. The lemmas in Sections~\ref{sec:colorlemmas} and~\ref{sec:powersupress} are more specific to the colored diagrams. 
 Taken together, the lemmas imply a simplified reduced-diagram structure which encapsulates hard factorization and facilitates soft-collinear factorization. 

Our reduced diagrams are very similar to the reduced diagrams describing the pinch surfaces~\cite{Sterman:1978bi,Libby:1978qf,Collins:1981ta}.
Indeed, our reduced diagrams include the singular momenta  defining this surface ($k^\mu =0$ or $k^\mu = \alpha p^\mu$ for some external $p^\mu$), 
but also have a precise expression as integrals (with singular and nonsingular parts) derived from the full Feynman diagrams as described in the previous section.

Recall that we define {\bf physical gauges} as either lightcone gauge, with a generic choice of reference vector, or factorization gauge (see Section~\ref{sec:factorizationgauge}) with generic $r_\ccc$. Our physical gauges also have generic reference vectors for the
polarizations of external collinear particles.
 In the literature, {\it physical gauges} often refers more generally to any gauge whose propagator-numerator corresponds to a sum over physical polarizations, including axial gauges. We
will not need to consider such a generality.

To be clear, although we do not say so explicitly in the formulation of each lemma, all the lemmas in this section are only proven to hold in physical gauges. Most of them in fact do not hold in Feynman gauge, which plays no role in our proof.

\subsection{Finding the IR sensitivities}
\label{sec:lemmas}

We now discuss how to locate the IR sensitivities in graphs. IR sensitivity is a delicate thing. One IR-insensitive line can contaminate a whole subdiagram, removing its IR sensitivity. 
This fact formalized in the {\it Zombie Lemma} (Lemma~\ref{lem:zombie}). However, Lemma~\ref{lem:zombie} requires the proof of
the {\it Log Lemma} (Lemma~\ref{lem:kappa0}), which states that IR sensitivities in graphs are at most logarithmic. Other facts that will be necessary to determine where IR sensitivities lie in QCD graphs are also proven in the process of showing Lemma~\ref{lem:kappa0}.

Our first step is to prove that in physical gauges, IR sensitivities are at most logarithmic:

\begin{lemma}\ltag{Log Lemma}
\label{lem:kappa0}
According to the power counting discussed in Section~\ref{sec:prelim}, in physical gauges any Feynman diagram in QCD (or any other renormalizable theory with only gauge interactions) scales at worst like $\kappa^a$ with $a\geq0$. 
Thus IR divergences are at most logarithmic.
\end{lemma}

\noindent This fact has been known for decades~\cite{Sterman:1978bi}. We reproduce the proof here for completeness and to facilitate
 the proofs of Lemmas~\ref{lem:NoSC4pts} through~\ref{lem:OnlySelfColl}.

 Although we will not discuss covariant
gauges much, it is also known that in Feynman gauge, individual diagrams can have divergences more severe than logarithmic~\cite{Collins:1989gx}. These power divergences provide an obstruction to using
reduced diagrams for a transparent picture of hard factorization. Of course, the power divergences cancel in a gauge-invariant sum over diagrams, but this cancellation is of little use in
a diagram-by-diagram analysis. Lightcone gauge with non-generic choices of reference vectors also do not lead to the same simple reduced-diagram picture.

The two lemmas that will be proven during the proof of Lemma~\ref{lem:kappa0} are:
\begin{lemma}\ltag{Collinear Lemma}
\label{lem:Collinear}
Consider two lines of a given diagram. If the lines cannot become collinear due to momentum conservation or if they give rise to a $\kappa$ suppression when they do become collinear,
then a virtual particle connecting between them cannot be collinear sensitive.
\end{lemma}

\begin{lemma}\ltag{4-point Lemma}
\label{lem:NoSC4pts}
There are no diagrams with soft-sensitive gluons attaching to soft-insensitive lines through a 4-point vertex.
\end{lemma}

\begin{proof}[Proof of Lemmas~\ref{lem:kappa0}, \ref{lem:Collinear} and \ref{lem:NoSC4pts}]
We will focus on proving the Log Lemma (Lemma~\ref{lem:kappa0}), and mention the other two lemmas as they come up. 

Before getting into the proof, we will need to establish the form of the various vertices in the theory in the limit where all of the particles involved are soft or collinear. First, as discussed in \tree, the 3-point vertex involving a soft gauge boson has the following limiting behavior:
\be
\fd{2.5cm}{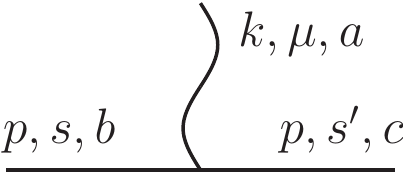} \;\;\LPeq\; -i2g_s \, \suTT_{\sub \suc}^\sua\, \delta_{ {\spincol s} {\spincol s'}}\; p^\mu, 
	\qquad {\rm for} \; k \; {\rm soft \; and }\; p \;{\text{on-shell and not soft}}
\label{SoftVert}
\ee
where 
$\sua,\sub$ and $\suc$ are color indices and ${\spincol s}$ and ${\spincol s'}$ are helicities (the wave functions of the non-soft particles are included). This result holds if the non-soft lines represent particles of any spin~\cite{Weinberg:1964ew}, in particular, these lines can be gluons. Similarly, the all-collinear vertex with at least one gauge boson is proportional to the momentum flowing through the vertex, by Lorentz invariance:
\be
\fd{2.5cm}{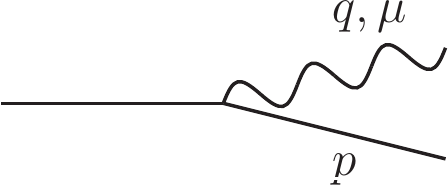} \;\;\propto\; p^\mu \propto q^\mu ,\qquad {\rm for}\; p\parallel q%
\label{CollVert}
\ee
Again, this is true irrespective of the spin of the particles in the straight lines and only when the lines are on-shell.

Now for the proof; we prove the Log Lemma (Lemma~\ref{lem:kappa0}) by induction on the number of loops. Tree-level diagrams trivially scale like $\kappa^0$, so Lemma~\ref{lem:kappa0} holds
for $n=0$. Then suppose it holds for $n-1$ loops and consider adding another loop. We will consider all possible ways to add a loop using 3- and 4-point vertices. 

For massless particles, propagators blow up when virtual lines are either soft or collinear. 
Let us begin with the soft case. According to the power-counting rules in Section~\ref{sec:prelim}, when the new line goes soft the measure associated with a soft line power counts as $d^4k \sim \kappa^8$  and the denominator of the propagator of the soft line counts as $k^{2} \sim \kappa^{4}$. If the soft line connects via 3-point vertices to two lines of momentum $p_\ccO^\mu$ and $p_\ccT^\mu$, then the new loop adds two more propagators with denominators $(p_\cci\pm k)^2$ for $\cci=1,2$. If $p_\cci^\mu$ is off-shell, this scales like $\kappa^0$; if $p_\cci^\mu$ is on-shell and not soft, it scales like $p_\cci \cdot k\sim \kappa^2$; and if
$p_\cci^\mu$ is soft it scales like $\kappa^4$.
The numerator of the propagators, combined with the 3-point vertices, power count the same as $p_\cci + k$. If $p_\cci^\mu$ is not soft, then $p_\cci + k \sim \kappa^0$; if 
$p^\mu_i$ is soft, then $p_\cci + k  \sim \kappa^2$. Thus when $p_\cci^\mu$ is off-shell, the numerator and denominator combine to $\kappa^0$; if $p_\cci^\mu$ is on-shell but not-soft, they combine
to $\kappa^0/\kappa^2 \sim \kappa^{-2}$ and if $p_\cci^\mu$ is soft, they combine to $\kappa^2/\kappa^4 \sim \kappa^{-2}$. The worst scaling is therefore when $p_\cci^\mu$ is on-shell, and then,
\be
\frac{p_\cci + k}{(p_\cci+k)^2} \sim \kappa^{-2} 
\quad\For p_\cci^2 = 0 \quad(\text{either soft or not-soft})
\ee
Thus, adding a soft loop with 3-point vertices only gives an enhancement if both lines it connects to are on-shell, in which case, the new loop power counts as
\be
	\fd{2.5cm}{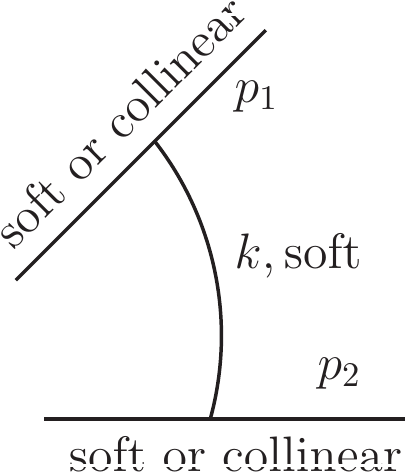}:   \qquad
d^4k \, \frac{1}{k^2}\, \frac{p_\ccO + k}{(p_\ccO+k)^2}\, \frac{p_\ccT +k}{(p_\ccT-k)^2} 
	\;\sim\; \kappa^8\, \frac{1}{\kappa^4}\, \frac{1}{\kappa^2}\, \frac{1}{\kappa^2} 
	\;\sim\; \kappa^0
\label{kappa0_Eq1}
\ee
on top of the original loop's power counting.

To be more precise, the lines with momenta $p_\ccO$ and $p_\ccT$ which connect to the soft momenta $k$ and go on-shell do not have to {\it directly} connect to $k$.
Even if there are some loops in the graph, as long as there are lines which go on-shell and connect to $k$ there will still be an enhancement.
We can simply think of these loops as producing a composite vertex:
\be
	\fd{2.5cm}{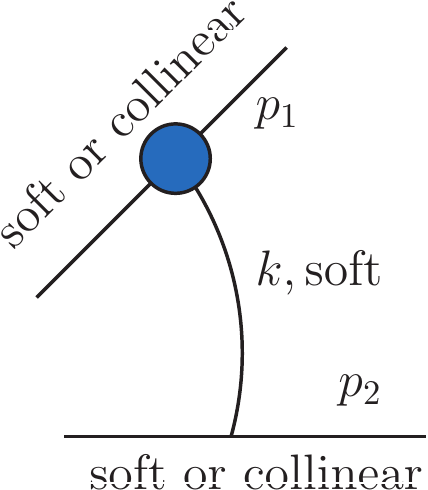}:   \qquad
d^4k \, \frac{1}{k^2}\, \frac{p_\ccO + k}{(p_\ccO+k)^2}\, \frac{p_\ccT +k}{(p_\ccT-k)^2} 
	\;\sim\; \kappa^8\, \frac{1}{\kappa^4}\, \frac{1}{\kappa^2}\, \frac{1}{\kappa^2} 
	\;\sim\; \kappa^0
\label{kappa0_Eq1}
\ee
Since there are no extra complications with such composite vertices, we will leave the composite case implicit in this proof.

Next suppose the new loop with the soft momentum connects via at least one 4-point vertex. This happens by the new gluon connecting to a 3-point vertex in the $n-1$ loop graph.
Again, the only way to get an enhancement is if the lines it connects to are on-shell. 
Due to the 4-point vertex,  the additional loop adds only two propagators rather than three. The new propagator denominators are $k^2$ and $(p_\cci + k)^2$. 
 The $n-1$-loop graph had a 3-point vertex, with either all three momenta collinear or one of them soft. Using \Eq{SoftVert} and \Eq{CollVert}, we see that the original 3-point vertex gave a contribution to the numerator of the graph of the form:
\be
\fd{2cm}{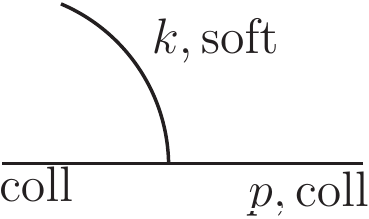} \;\;\propto\; p_\mu \Pi^{\mu\nu}(k) \sim \kappa^0 
\qquad\text{or} \qquad
\fd{2cm}{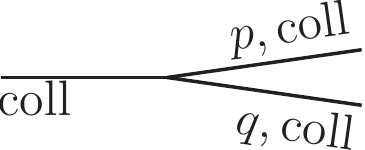} \;\;\propto\; p_\mu \Pi^{\mu\nu}(q) \sim \kappa, \text{ for } p\parallel q
\ee
Whereas, when we add the loop with the 4-point vertex, this becomes
\be
\fd{2cm}{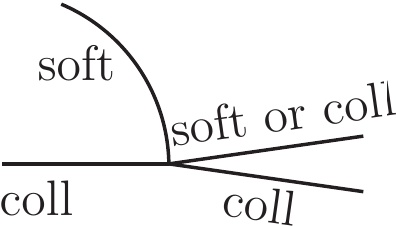} \;\;\propto\; g^{\mu\nu}g^{\rho\sigma} \sim \kappa^0
\ee
Thus, 
there is a  possible additional $\kappa^{-1}$ from killing the numerator suppression if the original graph had an all-collinear 3-point vertex. So, connecting a soft loop to a collinear line via a 4-point vertex adds a loop that power counts either as
\be
	\fd{2.5cm}{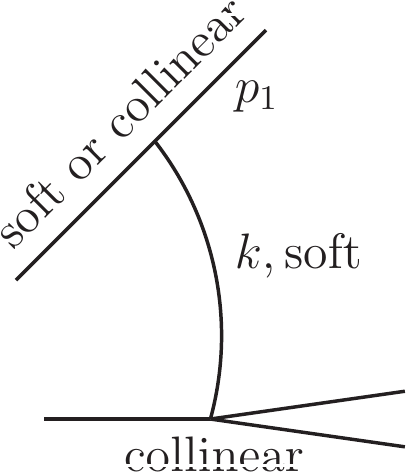}:   \qquad
d^4k \, \frac{1}{k^2}\, \frac{p_\ccO + k}{(p_\ccO+k)^2}\, \kappa^{-1}
	\;\sim\; \kappa^8\, \frac{1}{\kappa^4}\, \frac{1}{\kappa^2} \, \kappa^{-1}
	\;\sim\; \kappa
\label{kappa0_Eq2}
\ee
or as
\be
	\fd{2.7cm}{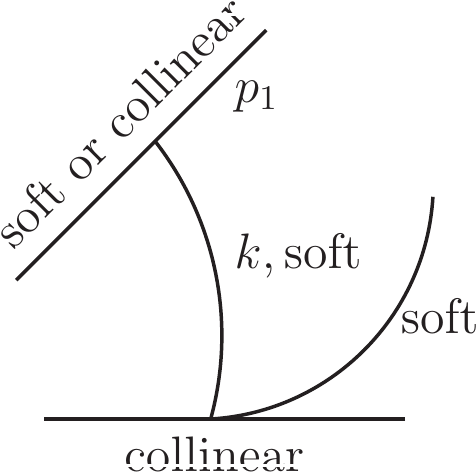}:   \qquad
d^4k \, \frac{1}{k^2}\, \frac{p_\ccO + k}{(p_\ccO+k)^2}
	\;\sim\; \kappa^8\, \frac{1}{\kappa^4}\, \frac{1}{\kappa^2}  
	\;\sim\; \kappa^2
\label{kappa0_Eq3}
\ee
In both cases, the new graph scales like a higher power of $\kappa$ than the graph it modified. By the same argument, adding a soft loop that connects to a collinear line on each end via a 4-point vertex will be (even more) IR finite. By the induction hypothesis, the rest of the graph scales at worst like $\kappa^0$, so any time we add a 4-point vertex with both soft and collinear momentum flowing through it, we get a $\kappa^{\geq1}$ scaling. Thus, we see that there cannot be a soft sensitivity when a soft gluon attaches to non-soft gluons through a 4-point vertex. This proves the \emph{4-point Lemma} (Lemma~\ref{lem:NoSC4pts}).

When {\it all} the relevant lines go soft, the 4-point vertices {\it can} contribute at leading-power.
To see this, consider the case where the soft loop connects to all-soft lines through a 4-point vertex and assume for now that the other end connects via a 3-point vertex. This case is just like the previous discussion in that the new loop adds only two new propagators of the form $k^{-2}$ and $(p_\cci+k)^{-2}$ and kills some of the suppression coming from the original 3-point vertex that became a 4-point vertex. However, in the all-soft case, the 3-point vertex suppression is a power of the soft momenta, which goes like $\kappa^2$ instead of $\kappa$ from the collinear case, so the new loop power counts as
\be
	\fd{2.5cm}{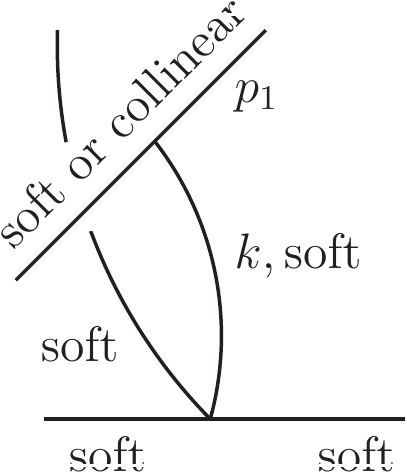}:   \qquad
d^4k \, \frac{1}{k^2}\, \frac{p_\ccO+ k}{(p_\ccO+k)^2}\, \kappa^{-2}
	\;\sim\; \kappa^8\, \frac{1}{\kappa^4}\, \frac{1}{\kappa^2} \, \kappa^{-2}
	\;\sim\; \kappa^0
\label{kappa0_Eq4}
\ee
Similarly, if the new soft loop connects to all-soft lines via a 4-point vertex on both ends, we only add one propagator, but we kill two $\kappa^2$-suppressed numerators, giving
\be
	\fd{2.5cm}{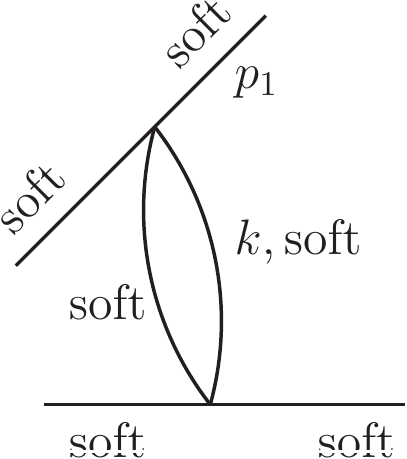}:   \qquad
d^4k \, \frac{1}{k^2}\, \kappa^{-2}\, \kappa^{-2}
	\;\sim\; \kappa^8\, \frac{1}{\kappa^4}\, \kappa^{-2}\, \kappa^{-2}
	\;\sim\; \kappa^0
\label{kappa0_Eq5}
\ee
Thus 4-point vertices involving all soft lines must be included. 
We have now exhausted all possible ways of adding a loop that can go soft and we have found that they all add a power counting of $\kappa^a$ for $a\geq0$ to the original graph.
This proves the Log Lemma as far as soft-scaling alone is concerned. 

 Now consider adding a line that can have a collinear sensitivity. As in the soft case, there are a number of ways that this can take place and we will systematically consider each possibility.
For the diagram to possibly be IR divergent the momentum in the line must be going collinear to the momenta of the lines it connects to on at least one end. Let us suppose first
that it is not also collinear to the line it connects to on the other end.
 Adding a line like this introduces two new on-shell propagators if it connects to the line to which it is collinear with a 3-point vertex, and only a single on-shell propagator if it connects with a 4-point vertex. In the first case, the all-collinear 3-point vertex will be proportional to the momentum flowing through it, as in \Eq{CollVert}, and this will give a suppression when contracted with any of the propagators (or external
polarization vectors) it connects to. This is because, in physical gauges, the propagator numerators are equal to the polarization-vector sum when the momentum in a propagator goes on-shell. Thus, $p_\mu \Pi^{\mu\nu}(q) \sim \kappa$ for $p\parallel q$ and we have
\be
	\fd{3cm}{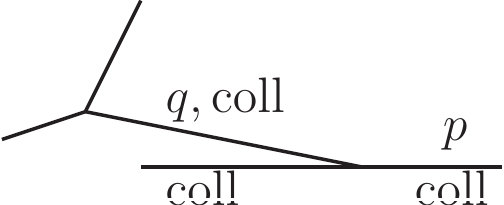}:   \qquad
	d^4q\, \frac{1}{q^2}\, \frac{1}{(p+q)^2}\, p_\mu \Pi^{\mu\nu}(q)
	\;\sim\; \kappa^4\,\frac{1}{\kappa^2}\,\frac{1}{\kappa^2}\, \kappa \;\sim\; \kappa
\label{kappa0_Eq6}
\ee
If the all-collinear vertex is a 4-point vertex, then we only get one new collinear propagator. However, going from an all-collinear 3-point vertex to a 4-point vertex kills the suppression that we just discussed, so we have
\be
	\fd{3cm}{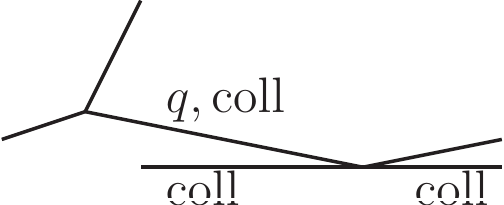}:   \qquad
	d^4q\, \frac{1}{q^2}\, \kappa^{-1}
	\;\sim\; \kappa^4\,\frac{1}{\kappa^2}\, \kappa^{-1} \;\sim\; \kappa
\label{kappa0_Eq7}
\ee
Finally, if the 4-point vertex has a soft line connecting to it, it will give a finite loop due to \Eq{kappa0_Eq2}. 
We conclude that unless the new line is collinear to the momenta on both ends, and in particular that all the relevant lines are on-shell,
the new diagram will have additional $\kappa$ suppression compared to the $n-1$ loop graph.

Combining \Eq{kappa0_Eq6} and \Eq{kappa0_Eq7}, we conclude that whenever a particle travels between two lines that could not originally go collinear, or that is
$\kappa$-suppressed if they do become collinear,
the resulting loop is $\kappa$-suppressed, and therefore, collinear insensitive. This proves the \emph{Collinear Lemma} (Lemma~\ref{lem:Collinear}).

It remains to show that when the momenta are all on-shell, the overall scaling is at worst $\kappa^0$. We have shown this already for soft singularities. So consider the 
remaining case when the new line goes collinear to all of the lines to which it connects. If both vertices are 3-point, we get three collinear propagators and two $\kappa$-suppressed products in the numerator:
\be
	\fd{2.5cm}{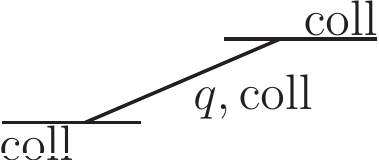}:   \qquad
	d^4q\, \frac{1}{q^2}\, \frac{1}{(p+q)^2}\, \frac{1}{p^2}\, p_\mu \Pi^{\mu\nu}(q) p_\nu
	\;\sim\; \kappa^4\,\frac{1}{\kappa^2}\,\frac{1}{\kappa^2}\,\frac{1}{\kappa^2}\, \kappa\, \kappa \;\sim\; \kappa^0
\label{kappa0_Eq8}
\ee
If only one of the vertices is a 3-point vertex, then adding the loop adds two propagators, one $\kappa$-suppressed product in the numerator due to the all-collinear 3-point vertex, and one $\kappa$ enhancement due to the removal of one of the original all-collinear 3-point vertices. Thus, graphs with one 3-point and one 4-point vertex power count as:
\be
	\fd{2.8cm}{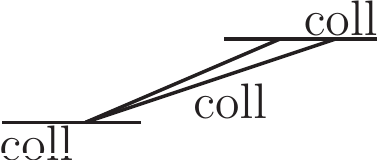}:   \qquad
	d^4q\, \frac{1}{q^2}\, \frac{1}{(p+q)^2}\, p_\mu \Pi^{\mu\nu}(q) \, \kappa^{-1}
	\;\sim\; \kappa^4\,\frac{1}{\kappa^2}\,\frac{1}{\kappa^2}\, \kappa\, \kappa^{-1} \;\sim\; \kappa^0
\label{kappa0_Eq9}
\ee
Finally, if the added loop connects on both ends to all-collinear 4-point vertices, then only one collinear propagator is added, but two 3-point vertices are removed causing two additional $\kappa^{-1}$ enhancements:
\be
	\fd{2.5cm}{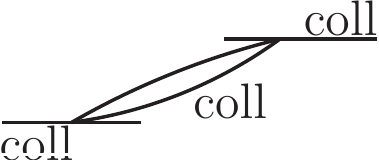}:   \qquad
	d^4q\, \frac{1}{q^2}\, \kappa^{-1}\, \kappa^{-1} \;\sim\; \kappa^4\,\frac{1}{\kappa^2}\, \kappa^{-1}\, \kappa^{-1} \;\sim\; \kappa^0
\label{kappa0_Eq10}
\ee
So, all possible additional loops that involve all-collinear vertices power count as $\kappa^0$ and are logarithmically collinear singular.

We have shown that any possible addition of a loop power counts as $\kappa^a$ for $a\geq0$.
 Therefore, by induction, every graph in physical gauges power counts like $\kappa^a$ for $a\geq0$ and is at most logarithmically divergent. This proves Lemma~\ref{lem:kappa0}.
\end{proof}

 Now, let us  define the term {\bf subdiagram} to mean a part of a larger diagram that could be cut out with an arbitrarily shaped (possibly 3D) cookie cutter. A subdiagram is considered as a function of the generic (not necessarily on-shell) momenta of the lines that the cookie cutter cut. These lines are considered to be external lines of the subdiagram, though they may have been internal in the original graph. Internal lines in a subdiagram are the complement of external lines. 
 
 With this definition, we can now make a useful observation about how IR-insensitive lines scale
 with $\kappa$ to establish how IR-insensitive graphs can {\it infect} any line they come in contact
 with, making it also IR insensitive. This observation is encapsulated by the following lemma:
\begin{lemma} 
 \ltag{Zombie Lemma}
\label{lem:zombie}
Consider adding a new internal line $L$ to a subdiagram with no IR-sensitive lines.  If at least one end of $L$ attaches to an internal line of the original subdiagram, then $L$ is IR insensitive.
\end{lemma}
\begin{proof}
Since no line in the subdiagram is IR sensitive, in any soft or collinear limit the subdiagram scales like $\kappa^a$ for some $a>0$. First, consider whether the line $L$ can have a soft sensitivity. When $L$ becomes soft, it produces a loop that scales like $\kappa^0$ at most. However, this only happens if the lines it connects to are on-shell (or it produces an on-shell line elsewhere in the subdiagram). By assumption, one of these lines is an internal line from  the original subdiagram, so there is a corresponding $\kappa^a$ suppression from the rest of the subdiagram. Thus, overall the subdiagram is still soft insensitive and so is the line $L$. That $L$ cannot be collinear sensitive follows directly from the \emph{Collinear Lemma} (Lemma~\ref{lem:Collinear}). Thus $L$ is IR insensitive and the Lemma is proven. 
\end{proof}

\subsection{IR insensitivity of the hard amplitude}
\label{sec:colorlemmas}
Two immediate consequences of the above lemmas  completely characterize the hard amplitude:
\begin{lemma}  \ltag{Hard-Blue Lemma}
\label{lem:hardblue}
Any all-blue 1PI subdiagram containing the hard-scattering vertex is IR insensitive.
\end{lemma}
\begin{proof}
Any 1PI subdiagram that contains the hard-scattering vertex must have momenta from two different collinear sectors piping through it. Consequently, there must be 
a line $L$ that connects between two lines that cannot simultaneously become collinear by momentum conservation.
The \emph{Collinear Lemma} (Lemma~\ref{lem:Collinear})  then implies that $L$ is not collinear sensitive. Since $L$ is blue (by hypothesis), it is soft insensitive as well, and hence IR-insensitive. Now, starting with the 1-loop graph containing $L$, we can build up the rest of the 1PI subdiagram by adding new lines (inserting vacuum loops in the middle of $L$ is allowed). Whenever a new line   connects to $L$, or to the network of lines previously connected to $L$, it is IR-insensitive by the \emph{Zombie Lemma} (Lemma \ref{lem:zombie}). Alternatively, a new line
 might connect to external lines of the subdiagram. If it connects two in the same sector, the graph cannot be 1PI. If it connects two in different sectors, the new line is IR-insensitive for the same reason $L$ is, and we can replace $L$ by this new line to continue our argument. Thus every line in the 1PI subdiagram is IR-insensitive, as was to be shown.
\end{proof}

\begin{lemma}  \ltag{Hard-Red Lemma}
\label{lem:nored}
Red lines cannot connect to internal lines of an all-blue 1PI subdiagram containing the hard-scattering vertex.
\end{lemma}
\begin{proof}
Any all-blue 1PI subdiagram containing the hard vertex is IR-insensitive by the \emph{Hard-Blue Lemma} (Lemma~\ref{lem:hardblue}). Any line connecting to an internal line of this subdiagram must also be IR-insensitive, by the \emph{Zombie Lemma} (Lemma \ref{lem:zombie}). Since red lines are soft sensitive, by definition, these lines cannot be red. 
\end{proof}

These two lemmas explain why some colored graphs are absent in physical gauges.
For example, as discussed in Section~\ref{sec:algex2}, the diagrams 
\be
\fd{1.5cm}{EikEx2_2.pdf} \qquad\text{and}\qquad \fd{1.5cm}{EikEx2_3.pdf}
\label{1PInsEx1}
\ee
are IR (in particular, soft) insensitive in generic-lightcone gauge and therefore, absent from the colored-graph decomposition. The diagrams
\be
\fd{1.5cm}{EikEx2_4.pdf} \;, \qquad
\fd{1.5cm}{EikEx2_1.pdf}  \qquad\text{and}\qquad  \fd{1.5cm}{EikEx2_5.pdf}
\label{1PInsEx2}
\ee
are present because the first two are IR divergent and the third is the IR-finite
``last'' graph in the decomposition. Note that  the second diagram in \Eq{1PInsEx2} does not satisfy the hypothesis of the \emph{Hard-Blue Lemma} (Lemma~\ref{lem:hardblue}) because without the red line, it is not a 1PI graph containing the hard vertex.

\subsection{Power-suppressed colored graphs}

\label{sec:powersupress}
So far, we have only characterized where the IR sensitivities are. Some diagrams, despite being IR sensitive contribute only at subleading power and can be dropped from a leading-power factorization theorem.  We have already seen an example of subleading diagrams in Section~\ref{sec:1loopEx2}. There, in particular in \Eq{Wils2} and \Eq{loglambda}, we found that for $q \parallel p_\ccO$, 
\be
 \fd{1.4cm}{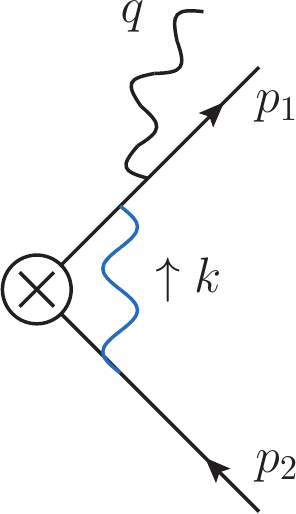} \sim \frac{1}{\lambda}
\quad\And\quad
 \fd{1.4cm}{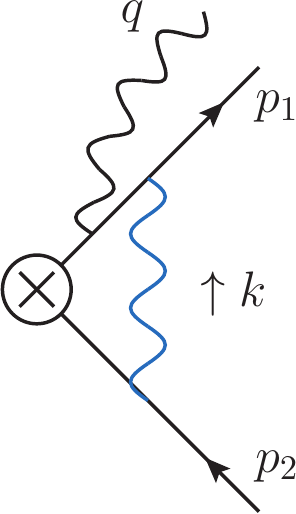} \sim \ln \lambda
\ee
In this example, the soft-insensitive loop in the first graph is IR-finite, so the $\lambda^{-1}$ comes from the tree-level splitting on external leg. In the second graph, the loop is tangled with the emission. At $\lambda=0$, the graph would be divergent, but for $\lambda>0$ it is not. Thus the graph scales like $\ln\lambda \ll \lambda^{-1}$.   The second graph is therefore subleading compared to the first and can be dropped. 
In a sense, the IR-insensitive loop {\it eats} the enhancement of the real emission. This is to be contrasted with IR-sensitive loops which do not eat emissions. For example, 
\be
\fd{1.4cm}{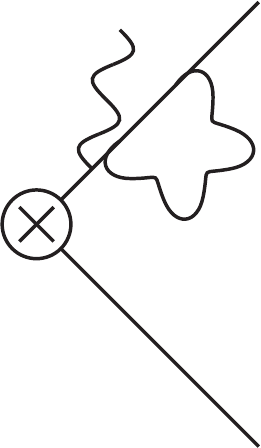} \;\sim\; \frac{\ln 0}{\lambda}
\;,\qquad
\fd{1.4cm}{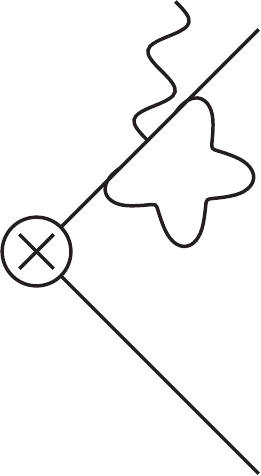} \;\sim\; \frac{\ln \lambda}{\lambda}
\quad\And\quad
\fd{1.4cm}{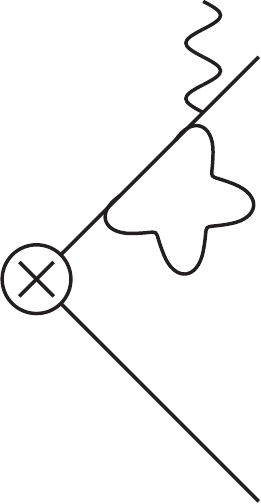} \;\sim\; \frac{\ln \lambda}{\lambda}
\ee
In each case, the graphs are divergent without the emission. In particular, the loop in the second graph cannot eat the emission.

The generalization of this example is embodied in the following lemma:
\begin{lemma} \ltag{Loop-emission Lemma}
\label{lem:EatPower}
Any diagram with an IR-insensitive 1PI subdiagram that has a real emission attached to an internal leg is power suppressed compared to a corresponding diagram where the emission comes off of an external leg. 
\end{lemma} 

\begin{proof}
An IR-insensitive subdiagram that is 1PI has at least one overall power of suppression when approaching the soft and collinear limits. That is, it scales like $\kappa^a$ for some $a>0$. Suppose some line in the loop has momenta $q+k$ in it, where $q$ is the external momenta and $k$ is the loop momenta. Adding an external collinear emission connecting inside the loop gives an additional propagator with momentum $p+q+k$ with $p$ the new external momenta. Since $(p+q)^2\sim \lambda^2$, when $k$ goes collinear to $q$, this propagator scales like 
 \be
 \frac{1}{(p+q)^2 + 2 (p+q) \cdot k + \kappa^2} \sim \frac{1}{\lambda^2 + \kappa^2}
 \ee
In physical gauges, the vertex contracted with the  polarization gives $(p+2q+2k)\cdot \epsilon \sim \lambda + \kappa$ when all of these momenta are collinear. The net effect is
therefore $\dfrac{\kappa + \lambda}{\kappa^2 + \lambda^2}$. 
So if of $n$ emissions, $m$ are inside the loop the diagram scales like
\be
\frac{1}{\lambda^{n-m}} 
\int d \kappa\, \kappa^{a-1} \left( \dfrac{\lambda + \kappa}{\lambda^2+ \kappa^2} \right)^m \;\sim\;
\begin{cases}
 \dfrac{\ln \lambda}{\lambda^{n-a}},\quad m \ge a\\[4mm]
\dfrac{1}{\lambda^{n- m} }, \quad m < a
\end{cases}
\ll\; \frac{1}{\lambda^n}
\ee
Thus the diagrams with any number of collinear emissions coming from within the loop are power suppressed compared to the diagram with $m=0$, where
all the emissions are outside the loop.

Soft emissions are similar. Adding a soft emission to an IR-insensitive subdiagram gives $(\lambda^2 + \kappa^2)^{-1}$
for the propagator, as before but now $(p+k+q)\cdot \epsilon \sim 1$ since although $k$ and $q$ are soft, $p$ is not. Thus each new
 emission from within the loop gives $(\lambda^2 + \kappa^2)^{-1}$ compared to $\lambda^{-2}$ from outside the loop, and becomes suppressed upon integration as above.

Thus, for either soft or collinear emissions, emissions coming out of an IR-finite loop (or an IR-finite, 1PI subdiagram) are power suppressed and can be dropped at leading power.
\end{proof}

A final lemma finishes the required ingredients for the advertised reduced diagram picture:

\begin{lemma} \ltag{Self-Collinear Lemma}
\label{lem:OnlySelfColl}
Graphs where a collinear gluon is emitted from a leg to which it cannot be collinear near an IR sensitivity are power suppressed compared to graphs where the gluon can be collinear to the leg it is emitted from
near an IR sensitivity.
\end{lemma}
\begin{proof}
This lemma was proven for tree-level graphs in~\tree, using that self-collinear emissions have an enhanced propagator compared to non-self-collinear ones.  The subdiagram to which a collinear emission is connected must be IR sensitive, by the previous Lemma (Lemma~\ref{lem:EatPower}), and therefore, cannot connect to 1PI subdiagram containing the hard vertex with only blue lines, by the \emph{Hard-Blue Lemma} (Lemma~\ref{lem:hardblue}). Thus the subdiagram to which the emission is connected can only contain external momenta associated with a single collinear sector before the emission is added. Thus, near an IR sensitivity all of the propagators in the subdiagram are either soft or collinear to the same direction and the lemma follows from the same reason it did at tree level.
\end{proof}

That completes the lemmas. As a reminder, all of these lemmas hold in physical gauges, as defined at the beginning of this section,
and are generally violated in Feynman or other covariant gauges.

\subsection{General reduced diagram}
With these lemmas we have all of the rules required to {\it reduce} the most general graphs that contribute to $N$-jet-like scattering in a physical gauge. We first expand the various loop momenta and soft external momenta in their soft limit to write a diagram as a sum of colored diagrams with soft-sensitive virtual particles and soft-insensitive ones. The lemmas guide the coloring, by
indicating where the soft sensitivities can be, they indicate which red or blue lines can have associated collinear sensitivities, and which colored diagrams are power suppressed (even if IR sensitive) compared to other diagrams with the same external states
at the same order in perturbation theory.

To draw the physical-gauge reduced diagram, first note that the \emph{Hard-Blue Lemma} (Lemma~\ref{lem:hardblue})
 tells us that each diagram has an IR-insensitive core, given by the largest-possible 1PI subdiagram containing the hard vertex
 which has only blue lines. 
 By the \emph{Loop-emission Lemma} (Lemma~\ref{lem:EatPower}), no real emissions can come out of this core. Thus the hard core connects
to the rest of the diagram only through a single line in each sector. 

Now let us temporarily ignore red lines. Then there are only collinear singularities. By the \emph{Collinear Lemma} (Lemma~\ref{lem:Collinear}), it is impossible for any IR-sensitive graph to involve external momenta from two different collinear sectors. Thus,
 outside of the IR-insensitive core, the only collinear-sensitive subdiagrams are self-energy-type corrections to each sector. No blue lines go between
 sectors, or they would remove the IR sensitivity, by Lemma~\ref{lem:hardblue}, and should have been included in the core. 
Moreover, all collinear emissions come from self-collinear sectors, by the
 \emph{Self-Collinear Lemma} (Lemma~\ref{lem:OnlySelfColl}).
 Now add  the red lines back in. These lines can connect anywhere, except to the IR-insensitive core by Lemma~\ref{lem:nored}.

We have therefore shown that  any colored diagram can be drawn as
\be
  \addtolength{\fboxsep}{5mm}
\boxed{
\bra{X_\ccO\cdots X_\rN;X_\scs} \cO \ket{0}
 \;\overset{\substack{\text{physical} \\[0.5mm] \text{gauges} }}\LPeq\; 
\sum_\text{diagrams}
\fd{6cm}{AllLoopPinch.pdf}
}
\label{pinches2}
\ee
This reduced diagram has all the properties claimed in Section~\ref{sec:proofoutline}.
We call the sum over soft-insensitive (blue) 1PI subdiagrams involving the hard vertex the {\bf hard amplitude}
 and the sum of all soft-insensitive (blue) corrections to each external leg the {\bf jet amplitude}.
All of the soft-sensitive (red) lines are in the {\bf soft amplitude}, which is not necessarily connected.
 Note that these are amplitudes, in contrast to the common use of hard jet and soft functions to refer to squares of the amplitudes. 
This reduced diagram displays hard factorization. We have not yet shown how the jet and soft amplitudes can be disentangled which requires soft-collinear factorization.

 In generic lightcone gauge, where there are no ghosts, every line in or exiting $S$ is soft sensitive and is colored red. 
Because all the lines entering $S$ are soft-sensitive, no momenta within $S$ can be dropped with respect to any other momenta. Thus, there is no expansion done
by the coloring algorithm applied to $S$ and the loops within $S$ are given by the full-QCD Feynman rules. The lines leaving $S$ connecting to the $J_\ccj$ blobs
have been expanded, and have eikonal interactions with the $J_\ccj$ blob. 
As we will see in the next section, in factorization gauge, there are ghosts in the $S$ blob. Ghosts are always IR-insensitive, thus they should be colored blue. Since the
ghosts are blue without any expansion, the $S$ blob still contains all the unmodified loops of full QCD. 
In summary, in any physical gauge, the $S$ amplitude connects to the rest of the diagram
through soft-sensitive (red) lines with eikonal interactions and all the internal loops of $S$ are the same as in full QCD.

Before moving on to soft-collinear factorization, we pause to discuss the physically rich structure of the reduced diagram in \Eq{pinches2}. The hard factorization displayed here is a consequence of the geometrical property that the jet and soft subdiagrams attach to the hard subdiagram by a single line. Moreover, near the IR sensitivities in the loops, this line is almost on-shell and carries the net momentum of the jet. The hard subdiagram is therefore a completely independent process that depends only on a single net momentum and the overall quantum numbers for each collinear sector. Since the hard subdiagram has a smooth $\lambda\to0$ limit, it is completely insensitive to corrections of order $\lambda$; namely, it is completely insensitive to the distribution of collinear momenta among the external states $\bra{X_\ccO\cdots X_\rN;X_\scs}$.

The IR-finiteness of the hard amplitude arises because, in physical gauges, there are additional suppression factors from numerators in regions where the virtual particles go on-shell.  Since the hard amplitude is IR-insensitive, all the dynamics it encapsulates takes place at short-distance. Only distances of order
  $(\Delta x)_H = (P_\cci \cdot P_\ccj)^{-1/2}$ are relevant. Since the hard diagram communicates with the rest of the process only through the single lines which are off-shell by of order $\lambda$, these
interactions take place at distances $(\Delta x)_J\sim\lambda^{-1} (\Delta x)_H$ away from the hard core.
The subsequent non-soft (i.e. collinear) interactions take place around $(\Delta x)_J$, but in different directions. These collinear particles can then only communicate with each other through the exchange of long-wavelength modes, at distances of order  $(\Delta x)_S = \lambda^{-2}(\Delta x)_H$. The single particle in each sector coming out of the hard vertex corresponds to the single partons in hard matrix elements which can be calculated first and then either showered through a Monte Carlo event generator or convolved against analytic jet and soft functions in an inclusive calculation.

It is important to note that the intuitive picture drawn in \Eq{pinches2} is only valid in physical gauges, such as generic-lightcone gauge. In Feynman gauge or non-generic-lightcone gauges with enhanced polarization vectors \Eq{pinches2} is totally destroyed and the factorization becomes completely opaque~\cite{Collins:1989gx}. Although this seems like an esoteric point, these unphysical gauges are often used in discussions of factorization, such as in the original formulation of SCET~\cite{Bauer:2000ew,Bauer:2000yr}. For more discussion of this point see \tree.

\section{Step 3: Factorization Gauge}
\label{sec:factorizationgauge}
We saw in the previous section that generic-lightcone gauge limits the types of
diagrams which can contribute at leading power. Let us temporarily imagine restricting the region of integration of the loop momenta so that the soft-sensitive lines are forced to be soft 
and the soft-insensitive lines are forced to be collinear to some direction (instead of integrating them over  $\mathbb{R}^{1,3}$ like we should).
Then each reduced diagram would just be some integrals over soft and collinear particles with the same topologies as discussed in~\tree, and it seems like the same proof of soft-collinear decoupling would apply nearly unchanged.
However, \tree~made heavy use of the freedom to choose different reference vectors for different external particles. In particular,  a different reference vector $r_\ccj^\mu$ is chosen for each distinct collinear sector as well as another, $r_\scs^\mu$, for the soft sector. For this to work at loop level, we need to be able to choose the reference vector for a lightcone-gauge propagator to depend on the direction that the virtual gluon is going. We call a gauge with this flexibility {\bf factorization gauge}. Factorization gauge is critical to our proof and will be useful even when the virtual phase space is unrestricted over $\mathbb{R}^{1,3}$.

This section introduces factorization gauge. In factorization gauge, ghosts do not completely decouple, as they do in lightcone gauge. However, we will show that ghosts do not give rise to additional IR sensitivities. 
The next section will use factorization gauge to rigorously prove soft-collinear factorization, following essentially the same procedure as in~\tree.

\subsection{Definition}
We would like to be able to choose a different lightcone-gauge reference vector for each sector in the reduced diagram, which is the loop-level equivalent of choosing different reference vectors for the polarizations of each sector which was done in \tree. That is, we would like to choose a gauge such that the numerator of the gluon propagator is given by:
\be
\Pi^{\mu\nu}(k) = -g^{\mu\nu} + \frac{r^\mu(k) k^\nu + r^\nu(k) k^\mu}{r(k)\cdot k}
\label{pimnk}
\ee
with
\be
r^\mu (k) = \left\{ \begin{array}{cc}
r_\scs^\mu \;, \quad &  k  \text{ soft} \\
r_j^\mu  \;,  \quad &  k \parallel p_\ccj \\
r_h^\mu \;, \quad & \text{otherwise}
\end{array} \right.
\label{fgdef}
\ee
We assume $r_\scs^\mu$, $r_\ccj^\mu$ and $r_h^\mu$ to be lightlike and take $r_\scs^\mu$ and $r_\ccj^\mu$ as the reference vectors for polarizations of soft and collinear external gluons. 
Given that for loop momenta $k$ being soft or collinear is equivalent to $-k$ being the same, we will further define
\be
r^\mu(k) = r^\mu(-k)
\ee
so we only need to specify $r^\mu$ for positive-energy momenta. In practice, we will only use two different reference vectors: $r_\ccj^\mu = r_\ccc^\mu$ for all $\ccj$ and $r_h^\mu=r_\scs^\mu$. Although our arguments will only use the freedom to choose $r_\ccc$ and $r_\scs$ separately, we define factorization gauge with the full $N+2$ different reference-vector choices since this is consistent with our freedom to choose the reference vectors for the external gluons separately.
\footnote{Lightcone gauges with different (constant) reference vectors for different sectors have appeared in the SCET literature \cite{GarciaEchevarria:2011md}.}

To be concrete, we can make  Eq.~\eqref{fgdef} precise by chopping up phase space. For example, we can draw a Euclidean ball of size $\lambda^2 \,Q$ around $k=0$ for the soft region, draw cones of angle $\lambda$ around each jet region, and let everything else count as hard.  
The precise partitioning will not matter for the proof of factorization. 

Note that both soft-sensitive and soft-insensitive gluons have unrestricted momenta. For example,
 soft-sensitive (red) lines can be collinear or hard in which case their propagator has $r_\ccj$ or $r_h$. Factorization gauge does not assign a different reference vector to different lines in the reduced diagram (which would not be gauge-invariant). 
The assignment of reference vector is based only on the gluons' momentum, which is a legitimate gauge choice.

To implement this gauge choice into the Lagrangian, we can use the following non-local gauge fixing term:
\be
\cL_\text{g-f}(x) 
	= -\frac{1}{2\xi} \Big(r^\mu(i\partial) \,A_\mu^\sua(x) \Big)^2
\ee
 and then take $\xi \to \infty$. This gives a Faddeev-Popov determinant of
 \be
 \det\Big(\frac1g r(i\partial)\cdot D^{\adj}\Big) = \int \cD c\, \cD \bar c\, \exp\bigg(-i\int d^4x\, \bar c^\sua\, \big( r(i\partial)\cdot D^{\sua \sub}\big) \, c^\sub \bigg)
 \ee
 Therefore, the ghosts couple to gluons via
 \be
 \fd{3cm}{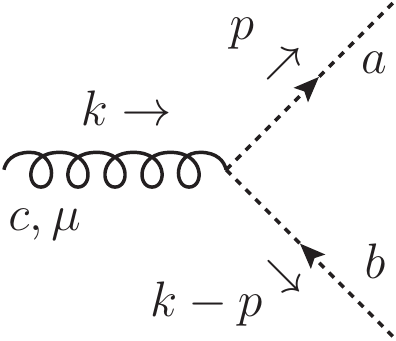} \quad\propto\; g\, f^{\sua \sub \suc}\, r^\mu (p)
 \ee
Thus, the vertex Feynman rule depends on $r^\mu(p)$ with $p^\mu$ the momentum of the ghost. 

The gluon propagator is
\be
\fd{2.5cm}{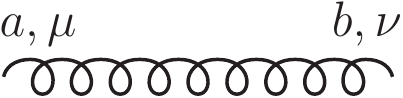} \quad = \quad \frac{i\,\delta^{\sua \sub}}{k^2\pie} \, \Pi^{\mu\nu}(k)
\ee
with $\Pi^{\mu\nu}(k)$ given in Eq.~\eqref{pimnk} which satisfies (for lightlike $r^\mu$)
\be
r_\mu(k)\, \Pi^{\mu\nu}(k)=0
\label{rpzero}
\ee
Recall that in lightcone gauge (where $r^\mu$ is constant), although the ghost-gluon vertex is still proportional to $r^\mu$, any graph
where a ghost couples to a virtual gluon is zero, due to Eq.~\eqref{rpzero}. If $r^\mu$ is also the reference vector of the external polarizations,
then the ghosts completely decouple diagram-by-diagram (for a different choice of external reference vector, individual diagrams with ghosts may not vanish
but their sum must due to the Ward identity, which guarantees reference-vector independence).
In factorization gauge, when a gluon of momentum $k$ couples to a ghost of momentum $p$, where $r(k) \neq r(p)$, the vertex will not be orthogonal to the gluon propagator or polarization.
Thus,  ghosts do not completely decouple in factorization gauge. Nevertheless, ghosts play a very small role in factorization, as we now show.

 \subsection{Ghosts Decoupling}
\label{sec:ghosts}
Although ghosts do not completely decouple, we will now show that ghosts cannot give rise to IR sensitivities. In particular, this means that ghost lines can never be red and can only contribute 
IR-insensitive loops internal to the hard, jet and soft blobs of \Eq{pinches2}. 

The fact that ghost loops do not give rise to IR sensitivities can be anticipated using unitarity. Independent of the gauge choice, we are always free to choose different reference vectors for the polarizations of external gluons in different IR sectors (as was extensively used in \tree). By unitarity, these on-shell soft and collinear gluons should be in one-to-one correspondence with cuts of loops near IR singularities. We then expect that in a gauge consistent with choosing different reference vectors for different IR sectors (i.e. factorization gauge) ghosts should not exist in IR-sensitive loops, since the ghosts cannot exist as external particles.

Ghosts cannot be part of IR-sensitive loops because near the IR-sensitive regions of integration, factorization gauge looks like a regular lightcone gauge in which ghosts decouple. That is, because the sum of soft momenta is soft and the sum of collinear momenta (to a single direction) is collinear, the all-soft and all-collinear ghost-ghost-gluon vertices vanish when contracted with the gluon propagator or external polarization exactly as they do in lightcone gauge. Therefore, 
ghosts will only modify the internal structure of the hard, jet and soft blobs by adding to them IR-insensitive loops.
  
What other types of vertices can give rise to IR sensitivities? Momentum conservation rules out the possibility of vertices with off-shell and two collinear momenta or off-shell and a soft and collinear momentum. The \emph{Collinear Lemma} (Lemma ~\ref{lem:Collinear}) says that a vertex with an off-shell momentum and two-hard on-shell momenta that are not collinear to each other cannot give rise to an IR sensitivity. So, we only need to consider ghost loops with singularities where the vertices in the loop have mixed-on-shell momenta.
There are then two
possibilities
\begin{enumerate}
\item Collinear ghost/soft ghost/collinear gluon, such as in
\be
\fd{5cm}{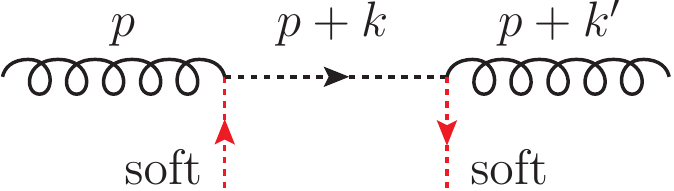}
\label{firsttype}
\ee
 where $p$ is a collinear and $k$ and $k'$ are soft. Or
\item Collinear ghost/collinear ghost/soft gluon,
\be
\fd{3.2cm}{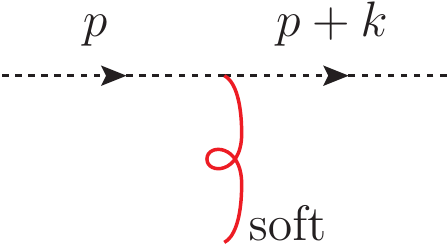}
\label{secondtype}
\ee
 with $p$ collinear and $k$  soft.
\end{enumerate}
 
In situations of the first type, one of the vertices is proportional to $r^\mu(k'-k)$ and the other vertex to $r^\mu(p+k)$. The $r^\mu(k'-k)$ is not orthogonal to the collinear-gluon propagator, $\Pi^{\mu\nu}(p)$, because $k$ and $k'$ are soft, so this vertex will not vanish. However, these non-vanishing vertices are always accompanied with the other vertex which is proportional to $r^\mu(p+k)$ which is equal to $r^\mu(p)$ since $p \parallel p+k$, and $r^\mu(p)$ is orthogonal to $\Pi^{\mu\nu}(p)$. Hence graphs with segments like in Eq.~\eqref{firsttype} always vanish near the singularity. A vertex of the second type, Eq.~\eqref{secondtype} does not automatically vanish on its own, since $r^\mu(p) \neq r^\mu(k)$. However, since there are no external ghosts, a ghost with a collinear momentum can only give rise to an IR sensitivity if it came from a gluon with collinear momentum. Thus there must be a vertex of the first type somewhere in the graph making the graph vanish in the IR sensitive region of the ghost.

That being said, we are not arguing that the soft gluon in \Eq{secondtype} cannot give rise to a soft sensitivity irrespective of the ghost momentum; we are only showing that the ghost lines themselves cannot give rise to IR sensitivities when they go on-shell. For example, we could have the following soft-sensitive graphs:
\be
G_1 \;=\; \fd{2.2cm}{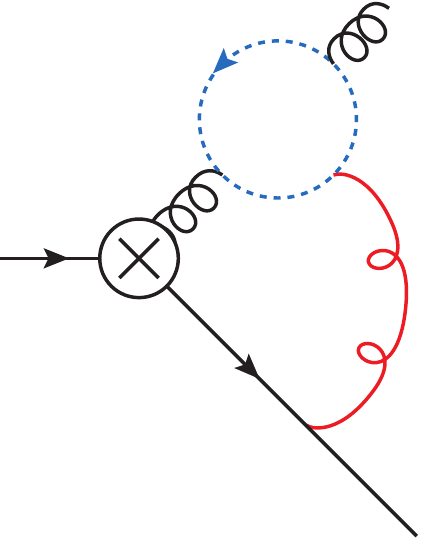}
\quad,\qquad\qquad
G_2 \;=\; \fd{2.3cm}{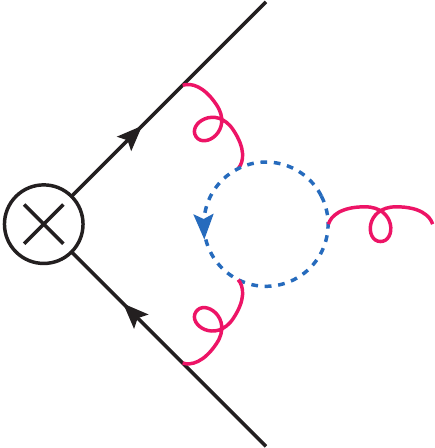}
\ee
In both cases, the integrand vanishes when the red-gluon(s) go soft and the ghost goes soft or collinear. However, when the ghost is off-shell, the red gluon(s) can go soft giving rise to a soft sensitivity of the same form as from the corresponding graphs where the ghost loop is contracted to a point.

The most important point that we use from this section is that ghost lines cannot be soft sensitive (red). Since we can treat ghosts as blue lines, any 1PI-blue subdiagrams that contain the hard vertex are IR insensitive by the \emph{Hard-Blue Lemma} (Lemma~\ref{lem:hardblue}), irrespective of whether or not they contain ghosts. Furthermore, the \emph{Loop-emission Lemma} (Lemma~\ref{lem:EatPower}) tells us that such subdiagrams do not have external emissions connecting to them. Hence, the reduced-diagram picture in \Eq{pinches2} is unchanged in factorization gauge, except for the fact that now the hard, jet and soft amplitudes may contain IR-insensitive ghost loops.

\section{Step 4: Soft-collinear factorization}
\label{sec:Step4}
The all-orders proof of soft-collinear factorization can now be built upon the skeleton of the tree-level proof from \tree. This is made possible by factorization gauge, in particular, our ability to choose a different reference vector for (real and virtual) soft momenta, $r_\scs$, and for (real and virtual) collinear momenta, $r_\ccj$. 
We will choose all of the $r_\ccj$'s to be a particular generic direction $r_\ccc$ not collinear to any of the collinear sectors; we call this the generic-$r_\ccc$ choice. For the soft reference vector we will go back and forth between choosing  $r_\scs$ in a particular collinear direction and $r_\scs$ generic, building up elements of soft-collinear factorization as we go.  We take $r_h = r_\scs$ for simplicity.

Before getting started, it is worth noting  how coloring works in matrix elements involving Wilson lines. One should color these diagrams just as with diagrams involving only local fields. Since emissions from Wilson lines already have eikonal vertices, they are exactly equal to their leading expansion in the soft-limit. 
Thus, in matrix elements involving only Wilson lines, such as $\bra{0} Y^\dg_\ccj W_\ccj  \ket{0}$, all the lines  are red. These lines interact with each other through an $S$ blob just like in \Eq{pinches2}. 
In matrix elements involving Wilson lines and fields, on the other hand, such as $\bra{0} \phi^\star W_\ccj  \ket{0}$, there can be both blue and red lines. As discussed in the previous section, in factorization gauge, the $S$ blob can also have blue lines if there are ghosts just like in the non-Wilson line matrix elements. 

 Although we use scalar QED notation,  operators in QCD look similar, with extra gauge and spin indices floating around. As far as hard-soft-collinear
factorization is concerned, the differences between scalar QED and QCD
are almost entirely notational. Thus we postpone the presentation of QCD matrix elements until Section~\ref{sec:QCD}.

\subsection{Soft and collinear factorization separately}
\label{sec:Cfact}

To begin, consider diagrams which only have red lines connecting to bare collinear sectors and call them $G_\text{pure red}$. 
Recall that  diagrams with red lines are derived from full theory diagrams by expanding to leading order around the limit where the momenta in all the red lines are small. This expansion is the same as the eikonal expansion. Equivalently we can expand by taking all the non-soft lines infinitely hard. This infinite-hard limit removes the dynamics from the non-soft lines, making them appear is classical sources which can be represented with Wilson lines. Thus, the sum of graphs of the form $G_\text{pure red}$ give matrix elements of Wilson lines:
\be
\sum G_\text{pure red} \equiv \fd{5.5cm}{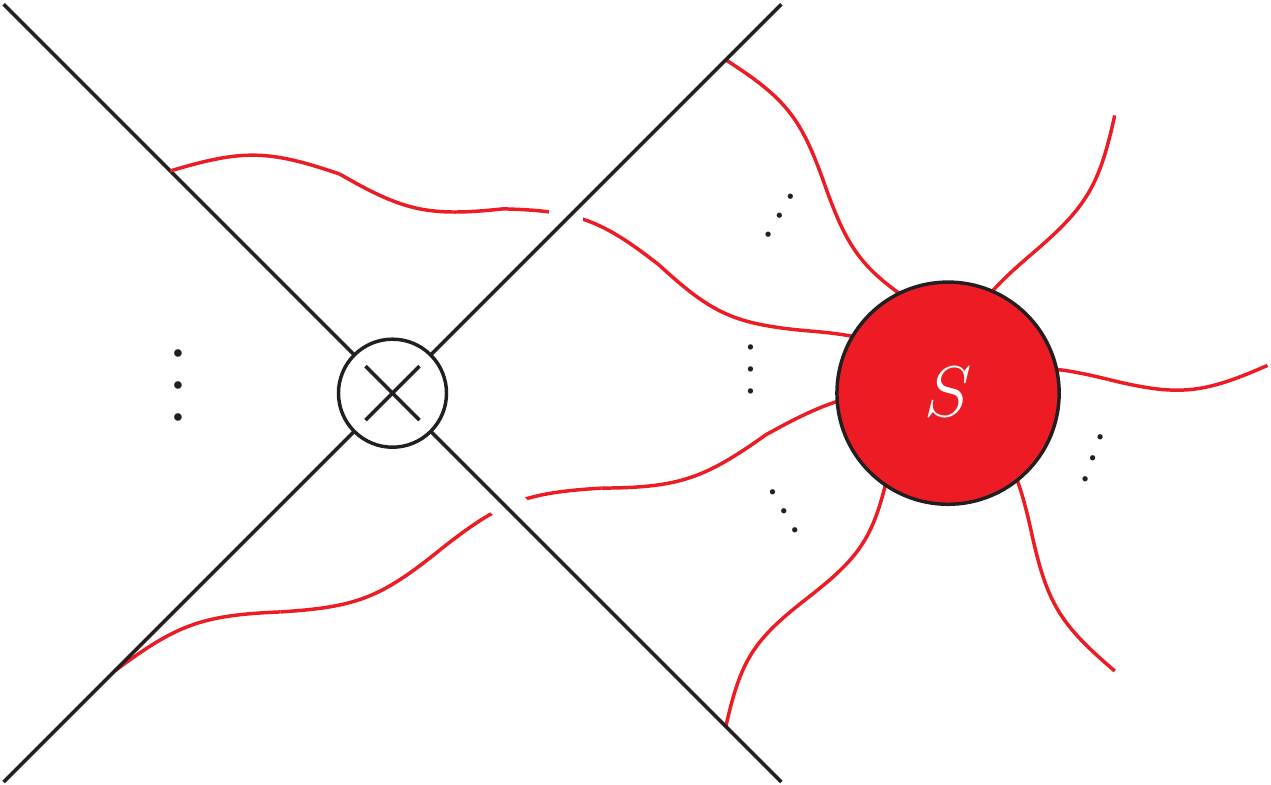} \;=\; \bra{X_\scs} Y_\ccO^\dg \cdots Y_\rN \ket{0}
\label{OnlySoftPinch}
\ee
where the sum over all possible diagrams of this topology is implicit. This equality holds in any gauge.

It is not hard to prove Eq.~\eqref{OnlySoftPinch} directly. The Wilson lines $Y_\ccj^\dg$ exactly give the eikonal Feynman rules, so doing the contraction-combinatorics just like in \tree, we see that the sum of the red lines connecting to the collinear ones is the same as the if the red lines connected to the soft Wilson lines. Since the $S$ blob gives all-possible QCD interactions (including ghosts in factorization gauge), we exactly get the matrix element of Wilson lines in \Eq{OnlySoftPinch} to all-loop order.

For  \Eq{OnlySoftPinch} to work the symmetry factors in the original uncolored loops must turn into the symmetry factors of the red loops. This is not hard to check. As discussed in Section~\ref{sec:algex3}, for every symmetry of an uncolored graph that is broken by the coloring, there are exactly as many different-but-equivalent soft sensitivities. So the symmetry factors work out correctly.

Pure collinear factorization is harder to discuss using colored diagrams. While diagrams with the maximal number of red lines are reproduced from a simple gauge-invariant Wilson line structure, diagrams with the maximal number of blue lines do not have any special simplifying property. Indeed, the Feynman rules for blue lines are a mess since they are given by differences between full QCD Feynman rules and eikonal Feynman rules. Moreover, graphs with all red lines are just as collinear sensitive as graphs with all blue lines.

Instead, it is perhaps useful to consider the following rather trivial diagrammatic identity, forgetting about the coloring altogether:
\be
\fd{3.5cm}{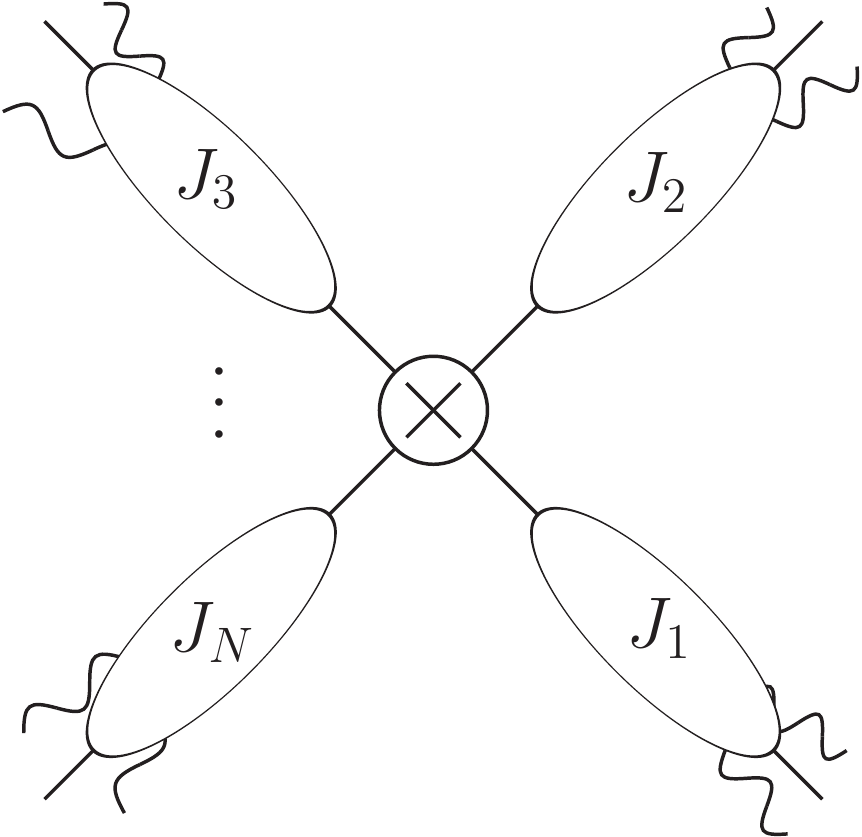} 
\;=\;
\bra{X_\ccO} \phi^\star  \ket{0} \cdots \bra{X_\rN} \phi \ket{0}
\label{SelfEnergyFact}
\ee
where, again, a sum over diagrams of the topology shown is implicit. In this equation, the right hand side is simply the sum over all graphs in scalar QCD with only self-energy corrections to each collinear sector. We saw such a structure emerge from the reduced diagram picture. Recall from the \emph{Hard-Blue Lemma} (Lemma~\ref{lem:hardblue}) that when a gluon connects between two different collinear sectors, there is no collinear sensitivity associated with it. Thus the diagrams on the right give the maximally collinear sensitive contributions to an amplitude at each order in perturbation theory in physical gauges.

\subsection{Soft-collinear factorization with a single collinear sector}
\label{sec:SingleCollSect}

We have seen that the sum of all graphs with the only red lines connecting to naked collinear sectors is reproduced by a matrix element
 of Wilson lines, as in Eq.~\eqref{OnlySoftPinch}, and that the self-energy type corrections to a single collinear sector are given by matrix elements of fields, as in Eq.~\eqref{SelfEnergyFact}. 
To prove soft-collinear factorization, the next step, as in \tree, is to factorize amplitudes containing both soft sensitivity and collinear sensitivity in one direction.

Let us define $\GJS$ as the sum of all colored diagrams that, when the red lines are removed, have collinear sensitivity to the $\ccj$ direction and no collinear sensitivity to any other direction. These are diagrams with any type of red or blue self-energy corrections to the $\ccj$-leg, any number of blue lines in the hard vertex, and any number of red lines connecting the $\ccj$-sector to other sectors.  
These diagrams all have the form
\be
\GJS \:=\:\fd{5.5cm}{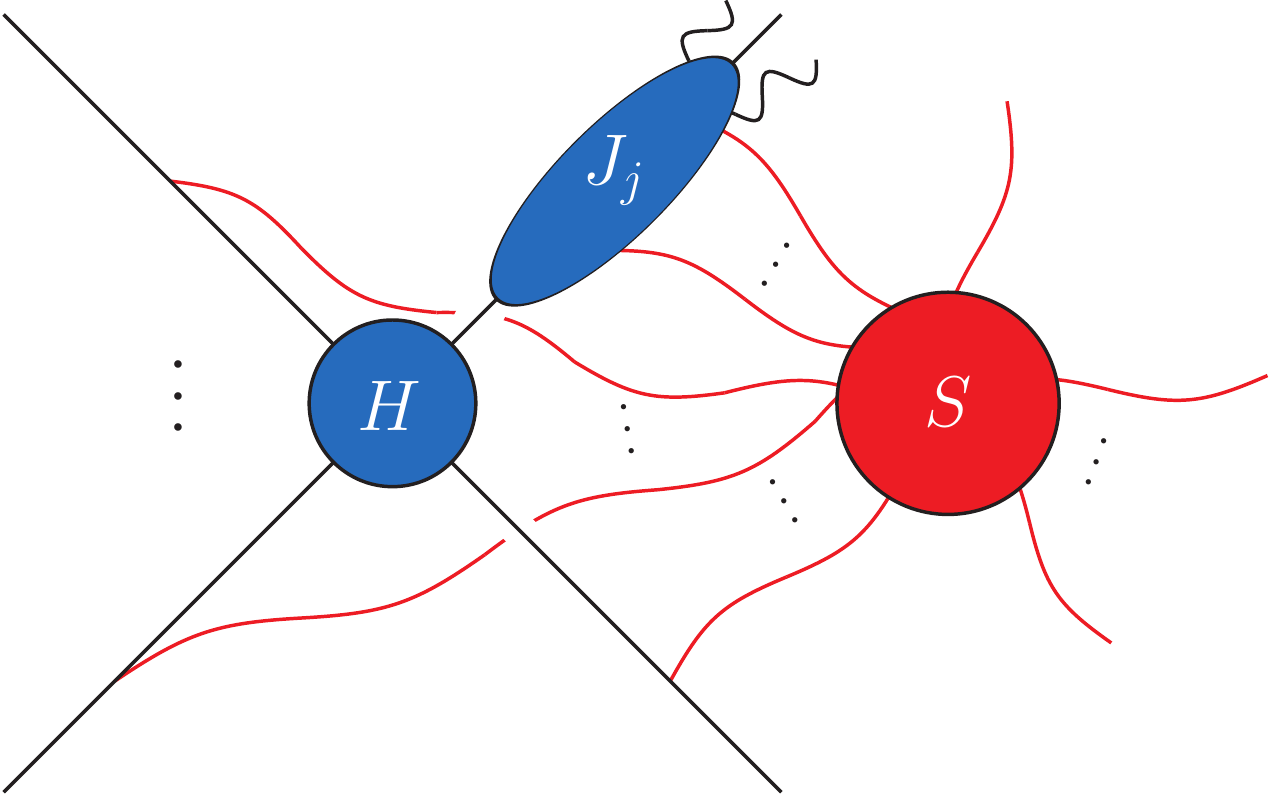}
\label{JjS1}
\ee
That $\GJS$ is a sum of such diagrams is left implicit. 
The $J_\ccj$ blob means all possible soft-insensitive loops (only blue lines) consistent with the external emissions in $\bra{X_\ccj}$ and the $S$ blob means all-possible graphs with only red lines (soft-sensitive lines or soft external lines) coming out. Note that the restriction that $J_\ccj$ have only blue lines is only a convention. It does not restrict the relevant subdiagrams, since any red self-energy contributions are simply absorbed into $S$. The $S$ blob does not have to be 1PI, planar or even connected. 

It is not hard to write down an operator definition of $\GJS$. As long is $r_\ccc$ is generic
\be
\bra{X_\ccj ;X_\scs} \, Y_\ccO^\dg \cdots \phi^\star \cdots  Y_\rN  \ket{0} 
	\overset{\substack{\text{any}~r_\scs\vspace{.5mm}\\ \text{gen.}~r_\ccc}}\LPeq
\GJS
\label{GJSdef}
\ee
There is an implicit choice of $H$ in this equation. 
The $Y_\cci$ Wilson lines for $\cci\ne \ccj$ provide the eikonal interactions between the red lines and the $\cci\ne \ccj$ collinear sectors. The
$\phi_j^\star$ allows for any possible self-energy type graphs in the $\ccj$ sector. Although the left-hand side is
gauge-invariant, in unphysical gauges (such as Feynman gauge or factorization gauge with a non-generic $r_\ccc \parallel p_\ccj$),  there will be collinear-sensitive diagrams with gluons going between different Wilson lines, or between a Wilson line and the $\ccj$ sector. The \emph{Hard-Blue Lemma} (Lemma~\ref{lem:hardblue}), which guarantees that such lines are only soft-sensitive, critically uses that a physical gauge
was chosen in the collinear sensitive region.

Now we will show that in factorization gauge with $r_h=r_\scs = p_\ccj$ there are no soft-sensitive graphs in $\GJS$
 with lines
connecting the $S$ blob to the $J_\ccj$ blob. This is the loop-level version of the tree-level result that when $r_\scs = p_\ccj$ any graph with soft external lines connecting to the $p_\ccj$-collinear sector is power suppressed.  
At tree level, the decoupling happens because the eikonal vertex gives a factor of $p_\ccj \cdot \epsilon(r_\scs)$ which is power suppressed when $r_\scs \LPeq p_\ccj$.  At loop level, we need to show that all the relevant graphs have a similar structure and are therefore similarly power suppressed. 

Although there is no restriction that red lines have soft momenta -- in general, red lines are integrated over all of $\mathbb{R}^{1,3}$ -- there is a restriction that red lines do have to be soft sensitive. Their soft sensitivity is inherent in the coloring, as discussed in Section~\ref{sec:coloring}. Thus, consider the soft-sensitive region of a subdiagram with red gluon emerging from the jet blob. It looks like
\be
S^\mu(k;\{p_\ccj\}, r_\scs,r_\ccc,r_h) \;=\;
\fd{4.5cm}{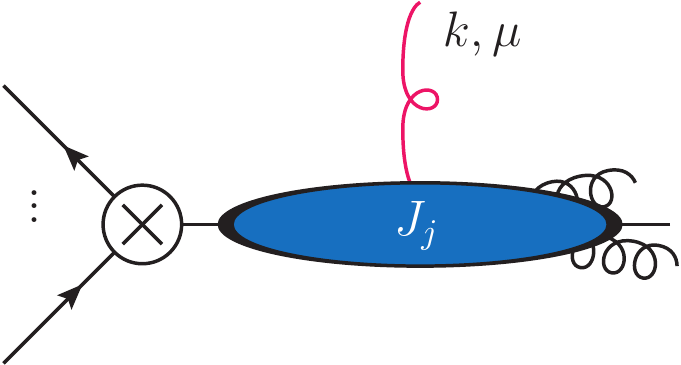}\;
\ee
where all the indices are suppressed except the Lorentz index on the soft line. Here, $S^\mu$ is a function of the momentum $k$, the external-collinear momentum $\{p_\ccj\}$ and the reference vectors associated with our gauge choice; that is, we imagine having done all of the loops in the collinear blob, $J_\ccj$. We now state a simple Lemma pertaining to  which Lorentz structures can carry the $\mu$ index in $S^\mu$:

\begin{lemma} \ltag{Soft-Attachment Lemma}
\label{lem:SoftAttach}
When $r_\scs=r_h$, the soft sensitivity can only come from the term in $S^\mu(k;\{p_\ccj\}, r_\scs,r_\ccj,r_h)$ proportional to $p_\ccj^\mu$.
\end{lemma}

\begin{proof}
The first step is to show that $S^\mu(k;\{p_\ccj\}, r_\scs,r_\ccj,r_h)$ has no term proportional to $r_\ccj^\mu$ at leading power. The only way to get an $r_\ccj^\mu$ term in $S^\mu$ is from the soft line connecting to a line that goes collinear to the $\ccj$-jet direction. However then, the leading power soft vertex is eikonal, namely, proportional to $p_\ccj^\mu$ instead of $r_\ccj^\mu$ as discussed in \Eq{SoftVert} (as discussed in Section~\ref{sec:ghosts}, the soft gluon cannot connect to a collinear ghost). So any terms proportional to $r_\ccj^\mu$ are $\kappa$ suppressed near the collinear sensitivity. Then, when the collinear region is integrated over, the $\kappa$-suppressed integrals give a finite value proportional to the volume of the collinear region, namely, $\lambda$ to some positive power. Thus $r_\ccj^\mu$ terms are power suppressed in loops and trees alike.

Now, since the $r_\ccj^\mu$ term is power suppressed, we are left with $r_\scs^\mu = r_h^\mu$, $k^\mu$ and $p_\ccj^\mu$. However, when the red line is contracted with a soft propagator or a soft external polarization, any term proportional to $r_\scs^\mu$ vanishes exactly and any term proportional to $k^\mu$ will be suppressed in $\kappa$. So these terms cannot contribute to a soft sensitivity in the red line.  This proves the lemma.
\end{proof}

Therefore, if we make the non-generic choice: $r_\scs = r_h = p_\ccj$, which we call {\bf collinear ${\mathbf r_\scs}$}, there will be no soft sensitivities connecting to the $\ccj$-collinear sector. We state this as a lemma:

\begin{lemma} \ltag{Collinear-$\mathbf{r_\scs}$ Lemma}
\label{lem:Collrs}
There are no soft-sensitive (red) lines connecting to the $\ccj$-collinear sector in factorization gauge in collinear-$r_\scs$ ($r_\scs^\mu=r_h^\mu=p_\ccj^\mu$).
\end{lemma}

\begin{proof}
The result is easy to see for a single soft line by Lemma~\ref{lem:SoftAttach}, since when $r_\scs=p_\ccj$ any soft propagator or external polarization vector will be orthogonal to $p_\ccj^\mu$. Now suppose we have many lines connecting to the $\ccj$-collinear sector. Working our way inwards towards the hard vertex, the outermost line must be soft-insensitive  by the argument for a single line, since it does not depend on the momentum of the other potentially-soft lines. If the outermost-red line connects to a different collinear sector, then by the \emph{Hard-Blue Lemma} (Lemma~\ref{lem:hardblue}) the rest of the lines must be blue and IR insensitive or, if any of the other lines are external-soft emissions, the whole graph is power-suppressed by \emph{Loop-emission Lemma} (Lemma~\ref{lem:EatPower}). So the lemma is proved in this case. On the other hand, if the outermost line connects back to the $\ccj$-collinear sector, because it is soft insensitive, it will just contribute to the blue-collinear blob and 
we can start the argument over again starting from the 
next-outermost line. 
In this way, we see that no soft-sensitive (red) lines can connect to the $\ccj$-collinear sector in collinear-$r_\scs$.
\end{proof}

For the rest of this paper we will take all of the collinear-reference vectors, $\{ r_\ccj \}$, to be the same generic direction, $r_\ccc$, that is not collinear to any of the collinear sectors. Furthermore, we will always take $r_h=r_\scs$. Neither of these choices is necessary, but they simplify the discussion. We have shown that if one chooses $r_\scs =p_\ccj$ there are actually no red lines connecting to the $J_\ccj$ blob in \Eq{JjS1}. This means that no expansion was done to the integrals in the the $J_\ccj$ blob and therefore the $J_\ccj$ blob is exactly the same as in the full theory. 
 Thus the set of relevant colored graphs contributing to $\GJS$ is somewhat different in generic-lightcone gauge from factorization gauge with $r_\scs=p_\ccj$:

\be
\GJS =
\underbrace{
\fd{5cm}{jcollPinch3.pdf}
}_{ \text{most physical gauges} }
,\hspace{8mm}
\GJS=
\underbrace{
\fd{5cm}{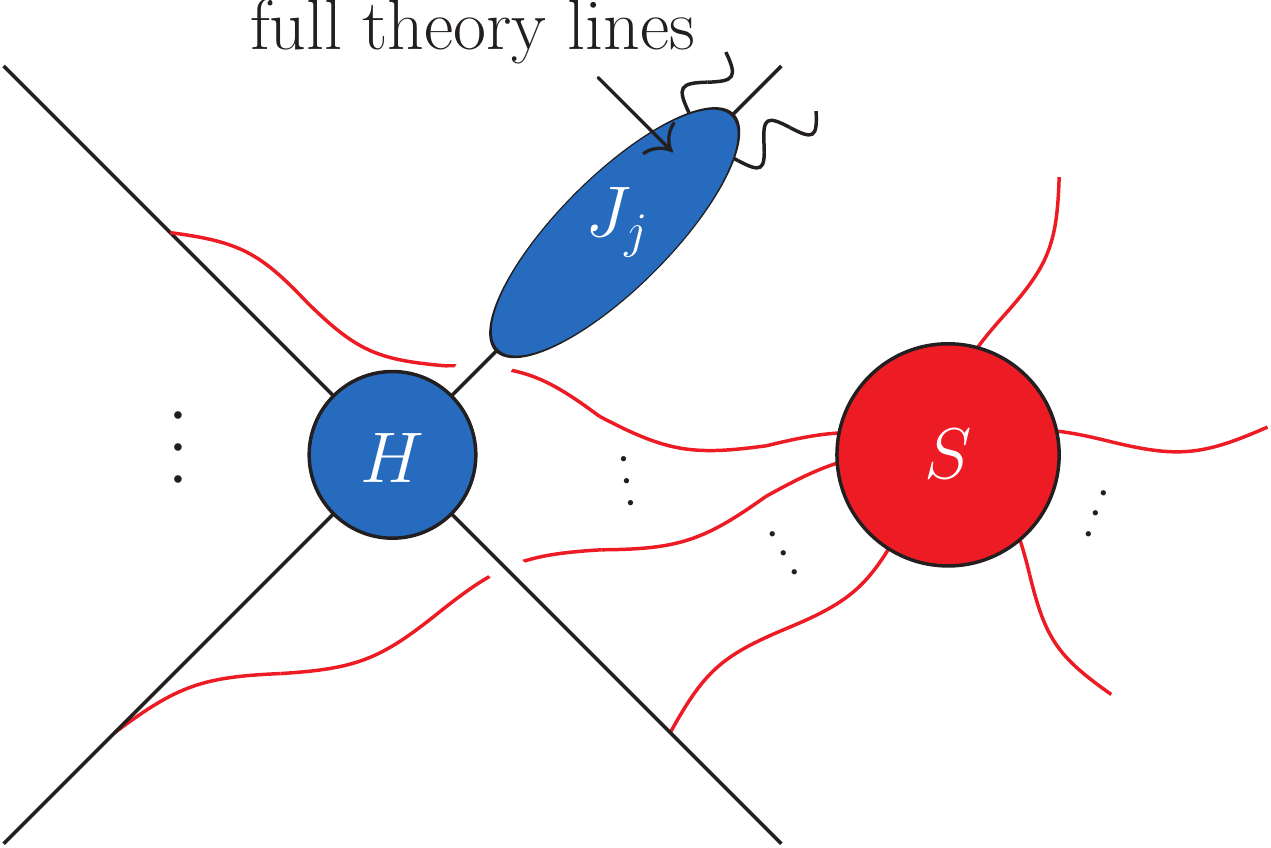}
}_{ \text{factorization gauge with}~r_\scs = p_\ccj\ne r_\ccc }
\label{JjS2}
\ee
In most physical gauges, there are blue self-energy bubbles in the $J_\ccj$ blob, red self-energy bubbles attaching to the $J_\ccj$ blob, as well as red lines leaving this blob and connecting to the other legs and to external-soft emissions. However, in factorization gauge with  $r_\scs^\mu=p_\ccj^\mu$, the $J_\ccj$ blob is unmodified from full QCD and no red lines connect to it. The $H$, $J$ and $S$ blobs are all different in the two cases.

Now, since there are no soft-sensitive lines connecting to the $J_\cci$ blob when $r_\scs=p_\ccj$, the amplitude from summing all the relevant graphs is closely related to the amplitude from a product of Wilson lines, as in Section~\ref{sec:Cfact}. More precisely,
\be
\fd{4.5cm}{jcollPinch4RsPj.pdf}
	\overset{\substack{r_\scs\, = \,p_\ccj \vspace{.5mm}\\ \text{gen.}~r_\ccc}}\LPeq
	C(\{n_\cci\cdot P_\ccj\}) \times
\fd{3cm}{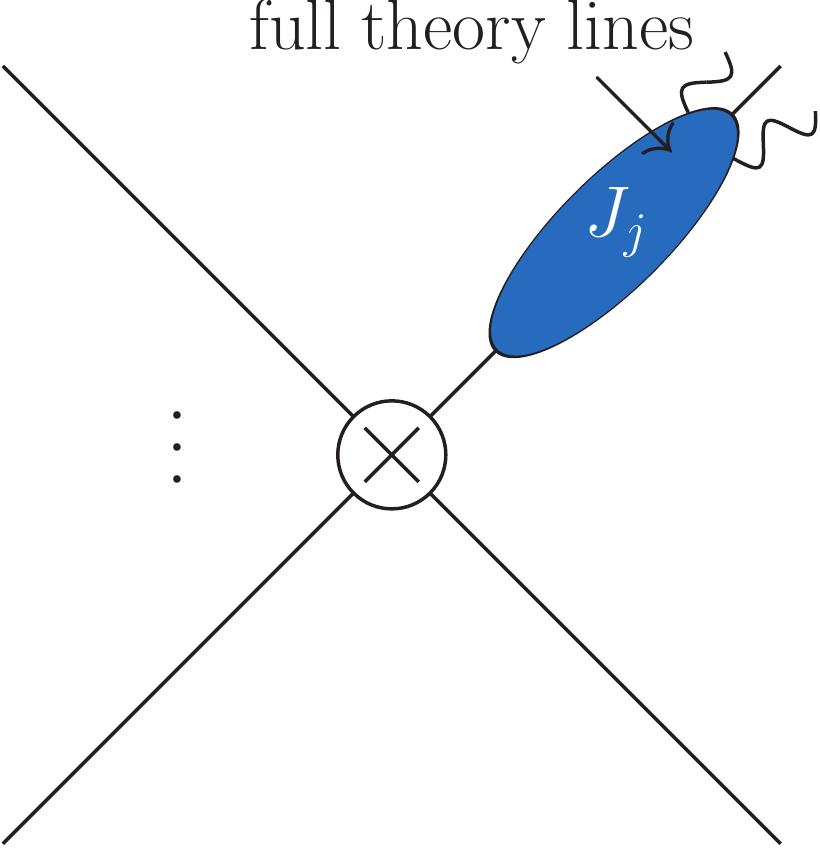} 
			\times \bra{X_\scs} Y_\ccO^\dg\cdots Y_{\ccj-\ccO} Y_{\ccj+\ccO} \cdots Y_\rN \ket{0}
\label{AllButOne} 
\ee
where $P_\ccj^\mu$ is the net collinear momentum in the $\ccj$ sector, $n_\cci^\mu$ is the lightlike direction of the $i$ sector and $\cC(\{n_\cci\cdot P_\ccj\})$ is an IR-finite function of $n_\cci\cdot P_\ccj$ for $i\neq j$. 

A subtle point is that $\cC(\{n_\cci\cdot P_\ccj\})$ does not have to equal the sum of the graphs in the hard amplitude $H(P_\ccj,k_i)$ evaluated at $k_i^\mu=0$ for all the soft loop momenta. 
To see where the difference comes from, 
recall that the $H$ blob is IR insensitive, so it is finite when any of the momentum from the red lines goes soft. Thus, we can write
\be
\int_{\{k_i\}}\!\! H(P_\ccj,k_i) \, \Eik(k_i,n_\cci)
=
H(P_\ccj,0) \int_{\{k_i\}}\!\!  \Eik(k_i,n_\cci) + 
\int_{\{k_i\}}\!\!  \big(H(P_\ccj,k_i) - H(P_\ccj,0)\big) \, \Eik(k_i,n_\cci)
	\label{HardSoftProc}
\ee
This allows us to extract the loops over the soft-sensitive red lines, $\Eik(k_i,n_\cci)$, from the soft-insensitive loops, $H(P_\ccj,k_i)$. Since the soft-sensitive loops are at most logarithmically divergent by
the \emph{Log  Lemma} (Lemma~\ref{lem:kappa0}), the second term is finite because $H(P_\ccj,k_i) - H(P_\ccj,0)$ vanishes when the $k_i \to 0$. Thus, we can pull out an overall IR-insensitive power series, $\cC(\{n_\cci\cdot P_\ccj\})$, times the pure-eikonal loops which are identically given by the matrix element of Wilson lines shown in \Eq{AllButOne}. Now, the second term on the right-hand side of \Eq{HardSoftProc} 
could either be power-suppressed (for example, if the $k_i\to0$ limit in question is tangled with a soft emission by Lemma~\ref{lem:EatPower}), or it could be some IR-finite integral 
multiplying  a lower-order IR-sensitive contribution from the soft Wilson-line matrix element. Thus, $\cC(\{n_\cci\cdot P_\ccj\})$ is not equal to $H(P_\ccj,0)$ in general. Instead, it is some IR insensitive power series in the perturbative coupling that starts at $1$. Despite the difference, $\cC(\{n_\cci\cdot P_\ccj\})$, like $H(P_\ccj,0)$, only depends on the net momenta in each collinear sector. The difference is from the subtraction
terms on the right-hand side of \Eq{HardSoftProc} which is subleading power when tangled with external emissions, by Lemma~\ref{lem:EatPower}.

Now, combining \Eq{GJSdef}, \Eq{JjS2} and \Eq{AllButOne}, and that, since the  $J_\ccj$ blob contains no red lines it is simply all the corrections to the $\ccj$-sector in full QCD, we have
\be
\bra{ X_\ccj ;X_\scs} \, Y_\ccO^\dg \cdots \phi^\star \cdots  Y_\rN  \ket{0} 
	\overset{\substack{r_\scs\, = \,p_\ccj \vspace{.5mm}\\ \text{gen.}~r_\ccc}}\LPeq
\cC(\{n_\cci\cdot P_\ccj\})  \bra{X_\ccj} \phi^\star \ket{0} \bra{X_\scs} Y_\ccO^\dg\cdots Y_{j-1} Y_{j+1} \cdots Y_\rN \ket{0}
\label{AllButOne_ME} 
\ee
In other words, $r_\scs =p_\ccj$ lets us disentangle a field from the product of Wilson lines.

\subsection{Bootstrapping in $Y_\ccj^\dg$ and $W_\ccj$}

 At this point, following \tree, we want to insert  $Y_\ccj^\dg$ into  $ Y_\ccO^\dg\cdots Y_{j-1} Y_{j+1} \cdots Y_\rN$ in  \Eq{AllButOne_ME} to make it gauge invariant. Recall that at tree-level choosing $r_\scs = p_\ccj$ for the external soft particles forces $Y_\ccj^\dg$ to contribute only power-suppressed terms. When loops are involved, it is not quite that simple, since the red lines are not restricted to be soft. Indeed,
self-contractions in $Y_\ccj^\dg$ (self-energy graphs on the $\ccj$-leg) are collinear sensitive, since in the collinear-sensitive region the gluon propagator has the collinear reference vector $r_\ccc$ instead of $r_\scs$. Thus it is true at tree-level but not at loop-level that inserting $Y_\ccj^\dg$ only gives a power-suppressed modification in collinear-$r_\scs$. 

When $r_\scs=p_\ccj$, contractions of $Y_\ccj^\dg$ with the other $Y_\cci$'s are soft insensitive by the \emph{Collinear-$r_\scs$ Lemma} (Lemma~\ref{lem:Collrs}) and must be blue. Then, by the \emph{Hard-Blue Lemma} (Lemma~\ref{lem:hardblue}), we know that any contractions of $Y_\ccj^\dg$ with the other $Y_\cci$'s are IR insensitive in physical gauges, as are any 1PI subdiagrams containing such contractions. 
So when $r_\scs=p_\ccj$, the {\it only} new IR sensitivities that  arise from adding in the $Y_\ccj^\dg$ are the collinear sensitivities in new self-energy type corrections to the $p_\ccj$ sector, namely from purely self-contractions of the $Y_\ccj^\dg$ operator. The sum of the purely self-contractions of $Y_\ccj^\dg$ is trivially given by $\bra{0} Y_\ccj^\dg \ket{0}$. Therefore, if we not only add the $Y_\ccj^\dg$ into the product of $Y_\ccO^\dg\cdots Y_{\ccj-\ccO} Y_{\ccj+\ccO} \cdots Y_\rN$ but also divide by $\bra{0} Y_\ccj^\dg \ket{0}$, the new 
collinear-sensitive contributions from $Y_\ccj^\dg$ will be completely removed, and this addition does not change the IR sensitivities.

The net effect of adding $Y_\ccj^\dg$ to the product of Wilson lines and dividing by $\bra{0} Y_\ccj^\dg \ket{0}$ is not nothing. There are graphs from this modification with gluons going between $Y_\ccj^\dg$ and one of the other legs. These contributions are soft insensitive (in factorization gauge with $r_\scs=r_p$) and collinear insensitive (since they connect different legs, by Lemma~\ref{lem:hardblue}), thus they are IR insensitive. Using the same procedure as outlined in \Eq{HardSoftProc}, we can absorb the IR-insensitive difference into a modification of the Wilson coefficient, which means that \Eq{AllButOne_ME} becomes
\be
\bra{ X_\ccj ;X_\scs} \, Y_\ccO^\dg \cdots \phi^\star \cdots  Y_\rN  \ket{0} 
	\;\overset{\substack{r_\scs\, = \,p_\ccj \vspace{.5mm}\\ \text{gen.}~r_\ccc}}\LPeq\;
\cC'(\{n_\cci\cdot P_\ccj\})
\, \bra{X_\ccj} \phi^\star \ket{0} \,
\frac{ \bra{X_\scs} Y_\ccO^\dg\cdots Y_\rN \ket{0} }{ \bra{0} Y_\ccj^\dg \ket{0} }
\label{SOTteenEq}
\ee
for some new IR-insensitive function $\cC'(\{n_\cci\cdot P_\ccj\})$.

This is the second time we find two objects with the same leading-power IR sensitivities differing by an IR-insensitive set of loops. Rather than modifying the Wilson coefficient, $\cC(\{n_\cci\cdot P_\ccj\})$, in each step for the IR-insensitive part, let us introduce the symbol $\LPFeq$ to mean that the IR-sensitivities on both sides agree at leading power. For example, with this notation, \Eq{SOTteenEq} becomes:

\be
\bra{ X_\ccj ;X_\scs} \, Y_\ccO^\dg \cdots \phi^\star \cdots  Y_\rN  \ket{0} 
	\;\overset{\substack{r_\scs\, = \,p_\ccj \vspace{.5mm}\\ \text{gen.}~r_\ccc}}{~\LPFeq}\;
	 \bra{X_\ccj} \phi^\star \ket{0}
\,
\frac{ \bra{X_\scs} Y_\ccO^\dg\cdots Y_\rN \ket{0} }{ \bra{0} Y_\ccj^\dg \ket{0} }
\label{SOTteenEqIR}
\ee
An $\LPFeq$ equivalence implies that a $\LPeq$ equivalence holds if some IR finite Wilson coefficient, $\cC(\{P_\cci\cdot P_\ccj\})$, is multiplied on one side. That is  
\be
A ~\LPFeq B \quad \Longleftrightarrow \quad \frac{A}{B} \LPeq \cC(\Sij)
\ee
for some IR-insensitive function $\cC(\Sij)$, where $\Sij =(P_\cci + P_\ccj)^2$.

Next, we show that collinear Wilson lines can be added without changing the IR structure.
Recall that collinear Wilson lines $W_\ccj$ have the same definition as soft Wilson lines $Y_\ccj$, but while the $Y_\ccj$ point along the jet direction $p_\ccj$, the $W_\ccj$ lines point in some direction
 $t_\ccj$ which is only restricted not to be collinear to $p_\ccj$.
 In lightcone gauge, if we choose $t_\ccj=r$, then $W_\ccj$ simply decouples since the gluons all have $t^\ccj_\mu \, \Pi^{\mu\nu}(k)=0$ for any
$k$ and $W_\ccj=1$ effectively. In factorization gauge with $r_\scs = p_\ccj$ and $t_\ccj$ and $r_\ccc$ generic, the Wilson lines do not decouple completely. However, it is still true that
\be
\frac{\bra{X_\ccj} \phi^\star W_\ccj \ket{0} }{ \bra{0} Y_\ccj^\dg W_\ccj \ket{0} }
\;
	\overset{\substack{r_\scs\, = \,p_\ccj \vspace{.5mm}\\ \text{gen. }r_\ccc}}{~\LPFeq}
\;
\frac{\bra{X_\ccj} \phi^\star \ket{0} }{ \bra{0} Y_\ccj^\dg \ket{0} } 
\label{Wbootrt}
\ee
This is true for exactly the same reason that we could bootstrap $Y_\ccj^\dg$ into \Eq{SOTteenEq}: when $r_\scs=p_\ccj$, any lines connecting to $\phi^\star$ and $Y_\ccj^\dg$ are blue by Lemma~\ref{lem:Collrs}. This means, by Lemma~\ref{lem:hardblue}, that the only new IR sensitivities introduced on the left-hand side of \Eq{Wbootrt} are those coming from purely self contractions of $W_\ccj$ which cancel in the ratio, proving Eq.~\eqref{Wbootrt}.

Now, since no red lines can connect to $\phi^\star$ or to $Y_\ccj^\dg$ when $r_\scs =p_\ccj$, the right-hand side of \Eq{Wbootrt} must be soft insensitive. This implies
that the left-hand side is soft insensitive too. Since the left-hand side is gauge invariant, it is soft insensitive in any gauge. In other words,
 $\bra{X_\ccj} \phi^\star W_\ccj \ket{0} \big/ \bra{0} Y_\ccj^\dg W_\ccj \ket{0}$ contains only blue lines. Moreover, these all-blue-line graphs cannot come from 
 $\bra{0} Y_\ccj^\dg W_\ccj \ket{0}$  or $\bra{X_\ccj} W_\ccj \ket{0}$ since these matrix elements, involving Wilson lines only, always have red lines attaching to the Wilson lines (with an arbitrary $S$ blob connecting them). Thus the blue lines
can come from $\bra{X_\ccj} \phi^\star \ket{0}$ or from contractions between $W_\ccj$ and $\phi^\star$. However, blue contractions between $W_\ccj$ and $\phi^\star$ are IR insensitive by the \emph{Hard-Blue Lemma} (Lemma~\ref{lem:hardblue}). Therefore, we have
\be
\frac{\bra{X_\ccj} \phi^\star W_\ccj \ket{0} }{ \bra{0} Y_\ccj^\dg W_\ccj \ket{0} }
\;\overset{\text{gen. }r_\ccc}{~\LPFeq}\;
\bra{X_\ccj} \phi^\star \ket{0} \Big|_\text{blue only} 
\; = \; \fd{3cm}{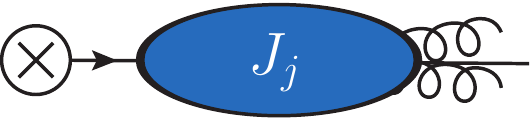} 
= \text{soft insensitive}
\label{blueonly}
\ee
where the $J_\ccj$ blob has only blue lines. We use this result below to strip the red lines off of a general matrix element.

Let us pause briefly to give an interpretation of $ \bra{0} Y_\ccj^\dg W_\ccj \ket{0}$. Note that $\bra{X_\ccj} \phi^\star W_\ccj \ket{0} $ has both collinear
and soft sensitivities, but $\bra{X_\ccj} \phi^\star W_\ccj \ket{0}/\bra{0} Y_\ccj^\dg W_\ccj \ket{0}$ has only blue lines so it is soft-insensitive. Thus  $ \bra{0} Y_\ccj^\dg W_\ccj \ket{0}$
is subtracting off the contribution which is both soft and collinear sensitive. 
Dividing by it implements the subtraction procedure known as the zero-bin subtraction in SCET. We will discuss this further in Section~\ref{sec:SCET} where we contrast our matrix-element definition with that used in the SCET literature.

Returning to Eq.~\eqref{Wbootrt}, if we combine it with Eq.~\eqref{SOTteenEqIR}, we find
\be
\bra{X_\ccj ;X_\scs} \, Y_\ccO^\dg \cdots \phi^\star \cdots  Y_\rN  \ket{0}
\;\LPFeq\;
\frac{\bra{X_\ccj} \phi^\star W_\ccj \ket{0} }{ \bra{0} Y_\ccj^\dg W_\ccj \ket{0} }
\times\bra{X_\scs} Y_\ccO^\dg\cdots Y_\rN \ket{0} 
\label{pulloneout1}
\ee
Although we only showed this IR-equivalence in collinear $r_\scs$ ($r_\scs =r_h= p_\ccj$, generic $r_\ccc$) since both sides of this equation are gauge invariant, it must hold for any choice of $r_\scs$ or $r_\ccc$ and more generally in any gauge (including Feynman gauge). Thus, Eq.~\eqref{pulloneout1} is not restricted to a particular gauge.

Note that Eq.~\eqref{pulloneout1} holds for any number of soft Wilson lines. As a special case, when there are two sectors:
\be
\bra{X_\ccj ;X_\scs} \phi^\star \, Y_\cci \ket{0}
\;\LPFeq\;
\frac{\bra{X_\ccj} \phi^\star W_\ccj \ket{0} }{ \bra{0} Y_\ccj^\dg W_\ccj \ket{0} }
\times\bra{X_\scs} Y_\ccj^\dg \, Y_\cci \ket{0} 
\label{pulloneout2}
\ee
which holds for any $i$ and $\ccj$.

\subsection{Sprig of thyme}
Eq.~\eqref{pulloneout1} (or more simply, \Eq{pulloneout2}) establishes soft-collinear factorization for a single non-minimal collinear sector. When multiple sectors are non-minimal, we clearly cannot choose $r_\scs= p_\ccj$ for all $\ccj$
simultaneously to repeat the above derivation. However, since Eq.~\eqref{pulloneout1} is gauge-independent, this is not necessary, as we will see. 

When $r_\scs^\mu$ is not collinear to $p_\ccj^\mu$,  Eq.~\eqref{pulloneout1} still holds, since it is gauge invariant. For generic choices of $r_\scs$,
there are soft-sensitive diagrams with red lines connecting to the $J_\ccj$ blob contributing to Eq.~\eqref{pulloneout1}. Although there is no diagram-by-diagram correspondence in Eq.~\eqref{pulloneout1}, the sum of diagrams
with a $J_\ccj$ blob and a fixed number $n$ of red lines attaching to it do correspond. In fact, as we will now show,
\begin{align}
\sum_{\text{perms on~$\ccj$}} \;
\fd{4.6cm}{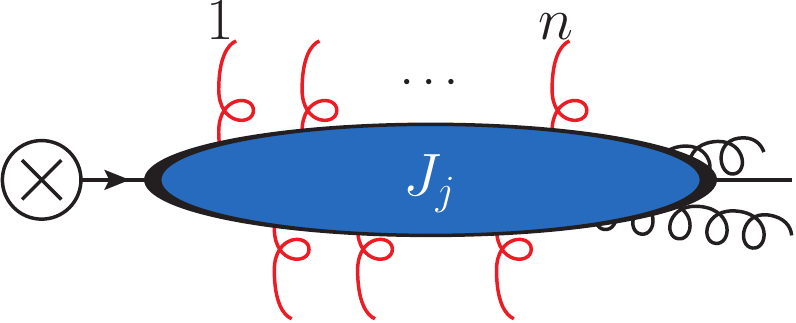} 
&\;\overset{\text{gen. }r_\ccc}{~\LPFeq}\;\;
\frac{\bra{X_\ccj} \phi^\star W_\ccj \ket{0} }{ \bra{0} Y_\ccj^\dg W_\ccj \ket{0} }
	\times
Y_\ccj^\dg
\,  \Big|_{n~\scred}
\label{SOT}
\end{align}
On the left-hand side, the usual $J_\ccj$ blob is defined to have only blue (soft-insensitive) lines and to have all such lines summed over and their integrals evaluated. 
We are considering diagrams which have $n$ generically off-shell red lines attaching to this $J_\ccj$ blob. In a full diagram the red lines can be closed into a loop, contracted with polarizations for external soft particles, or connect to a $J$ blob in another sector (not shown); we simply slice them close
to their attachment to the $J_\ccj$ blob and treat them as off-shell. The $\sum_{\text{perms on~$\ccj$}}$ means the sum over permutations of all possible ways of connecting the red lines to $J_\ccj$ blob on the left-hand side.
The right side has these same red lines now connecting to a $Y_\ccj^\dg$ Wilson line; the $Y_\ccj^\dg$ on the right-hand side is meant to be taken at the same order as the number of red lines on the left, as indicated by the $|_n$.  

\Eq{SOT} is the loop-level equivalent of the tree-level Eq. (94) in \tree. It shows that red lines can be stripped off of arbitrarily complicated jet amplitudes, like
leaves off a sprig of thyme, independent of where those red lines connect in the rest of the diagram.

\begin{proof}[Proof of \Eq{SOT} \ltag{Sprig-of-Thyme}] 
We will prove \Eq{SOT} by induction on the number of red lines $n$ leaving the $J_\ccj$ blob. The key, as in \tree, is to cancel all diagrams which contribute to both sides of \Eq{SOT} but have fewer than $n$ red lines attaching to the $J_\ccj$ blob using \Eq{pulloneout2} and the induction hypothesis.
The remaining diagrams will have all $n$ red lines connecting to the $J_\ccj$ blob so that \Eq{SOT} follows from \Eq{pulloneout2}.To avoid the notational quagmire of an algebraic induction proof as was done in \tree, in this paper we take a diagrammatic approach.

To begin  note that both sides of \Eq{pulloneout2} can be decomposed into colored diagrams. 
We will thus consider all of the blue diagrams in \Eq{pulloneout2} with a fixed number of red lines emerging from the $J_\ccj$ blob. 

{\bf n=0:}
 With no red lines coming out of $Y_\ccj^\dg$, this Wilson line is simply 1 and \Eq{SOT} follows from \Eq{blueonly} exactly:
\be
\fd{3cm}{SingleCollSect.pdf} 
\;\overset{\text{gen. }r_\ccc}{~\LPFeq}\;\;
\frac{\bra{X_\ccj} \phi^\star W_\ccj \ket{0} }{ \bra{0} Y_\ccj^\dg W_\ccj \ket{0} }
\label{n=1Step4}
\ee

{\bf n=1:}
Consider \Eq{pulloneout2} with one red end attached anywhere. Since there is only one red end attached, the red line must be part of $\bra{X_\scs} = \bra{k}$. 
Then the left-hand side of \Eq{pulloneout2}, at this order, is given by
\be
\bra{X_\ccj ;k} \phi^\star\, Y_\cci  \ket{0} \Big|_{1~\scred}
\;\overset{\text{gen. }r_\ccc}{~\LPFeq}\;\;
\sum_{\text{perms on~$\ccj$}}\;\fd{2.8cm}{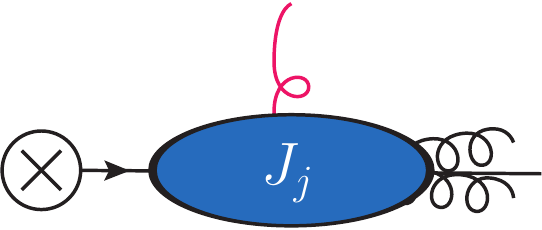} \;\;+\;\; 
\fd{3.3cm}{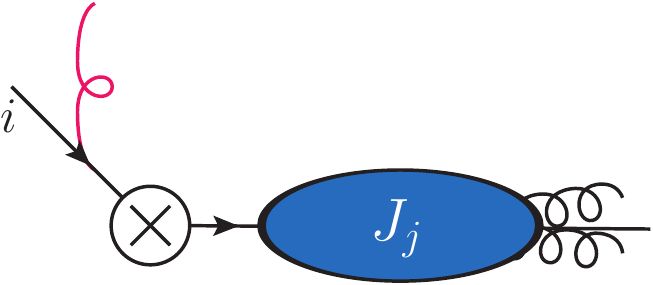} 
\label{n=1Step1}
\ee
On the right-hand side of \Eq{pulloneout2} the red line can only come from one of the Wilson lines in $\bra{k} Y_\ccj^\dg Y_\cci \ket{0}$ (since the other factor is all blue), so
\be
\left(
\frac{\bra{X_\ccj} \phi^\star W_\ccj \ket{0} }{ \bra{0} Y_\ccj^\dg W_\ccj \ket{0} }
\bra{X_\scs} Y_\ccj^\dg Y_\cci \ket{0} \right) \Big|_{1~\scred}
=
\frac{\bra{X_\ccj} \phi^\star W_\ccj \ket{0} }{ \bra{0} Y_\ccj^\dg W_\ccj \ket{0} }
\bigg( \bra{k} Y^\dg_j \ket{0}\Big|_{1~\scred}+  \bra{k} Y_\cci \ket{0}\Big|_{1~\scred}
\bigg)
\label{n=1Step2}
\ee
By \Eq{pulloneout2},  Eqs~\eqref{n=1Step1} and \eqref{n=1Step2} are equal. By \Eq{n=1Step4}, the second term on the right-hand sides of  Eqs~\eqref{n=1Step1} and \eqref{n=1Step2} are separately
equal. This leaves
\be
\fd{3.3cm}{IndSOT2.pdf} 
\;\overset{\text{gen. }r_\ccc}{~\LPFeq}\;\;
\frac{\bra{X_\ccj} \phi_j^\star W_\ccj \ket{0} }{ \bra{0} Y_\ccj^\dg W_\ccj \ket{0} }  \times \bra{k} Y_\ccj^\dg \ket{0}\Big|_{1~\scred}
\label{SOTstep1}
\ee
We  can now strip off the polarization vector (the contraction with the external state) because the vertex Feynman rule is the same for a red line in a loop connecting to another sector or for a real-emission, as discussed in \Eq{SoftVert} and also in the \emph{Soft-Attachment Lemma} (Lemma~\ref{lem:SoftAttach}). Thus, \Eq{SOTstep1} establishes \Eq{SOT} for $n=1$.

{\bf n=2:} At $n=2$, if the red lines are all external, \Eq{pulloneout2},  gives
\begin{align}
\sum_{\text{perms on~$\ccj$}}\; \fd{2.8cm}{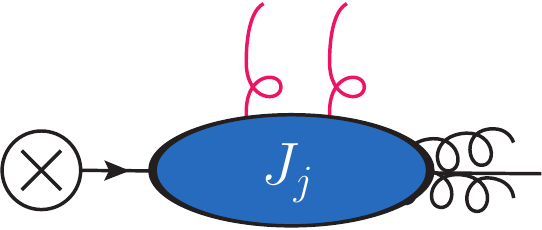}
&\;+\;
\sum_{\text{perms on~$\ccj$}}\; \fd{3.3cm}{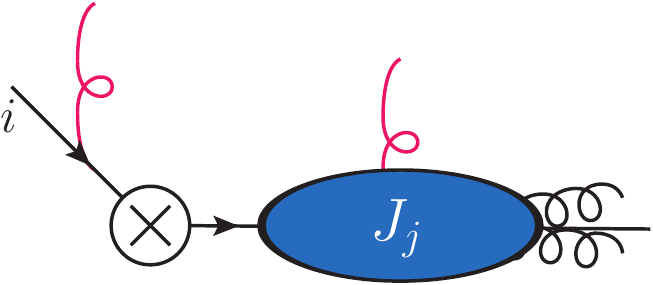}
\;+\;
\sum_{\text{perms on~$i$}}\; \fd{3.3cm}{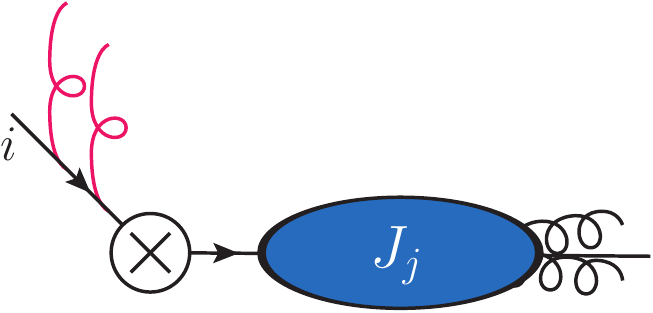}
\notag\\&
\;\overset{\text{gen. }r_\ccc}{~\LPFeq}\;\;
\frac{\bra{X_\ccj} \phi_j^\star W_\ccj \ket{0} }{ \bra{0} Y_\ccj^\dg W_\ccj \ket{0} } 
\left( \bra{k_1,k_2} Y_\ccj^\dg Y_\cci \ket{0}\right)\Big|_{2~\scred}
\end{align}
Using \Eq{n=1Step4}, the $\sum_{\text{perms~on~$\cci$}}$ terms cancel term-by-term with the $\cO(g^2)$ contractions of the external states with the $Y_\cci$ Wilson line. The middle term cancels with the $\cO(g)$ contractions of the external states with the $Y_\cci$ and $Y_\ccj^\dg$ operators using the previous induction step, \Eq{SOTstep1}. We are left with
\be
\sum_{\text{perms on~$\ccj$}}\fd{2.8cm}{IndSOT4.pdf}
\;\overset{\text{gen. }r_\ccc}{~\LPFeq}\;\;
\frac{\bra{X_\ccj} \phi^\star W_\ccj \ket{0} }{ \bra{0} Y_\ccj^\dg W_\ccj \ket{0} }
\times \bra{k_1,k_2} Y_\ccj^\dg \ket{0}\Big|_{2~\scred}
\label{SOTstep2}
\ee
This and the previous case are almost identical to the tree-level proof since there are as many external emissions as orders, $n$. That is, there are no red loops and we simply cancel off emissions off of the $i\neq j$ sector term-by-term using the previous induction hypotheses.

If the red lines are in a loop, then all cases where the red lines do not both come off the $\ccj$ line still cancel by the previous induction steps (which already have the polarization vectors stripped off).
Thus, after canceling these terms off in \Eq{pulloneout2}, we are left with:
\be 
 \sum_{\text{perms on~$\ccj$}} \; \fd{2.8cm}{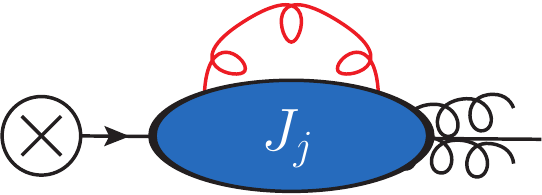}
\;\overset{\text{gen. }r_\ccc}{~\LPFeq}\;\;
\frac{\bra{X_\ccj} \phi^\star W_\ccj \ket{0} }{ \bra{0} Y_\ccj^\dg W_\ccj \ket{0} } 
\contraction[.7ex]{\bra{0}}{\!\!Y}{}{\!\!Y} 		
		\bra{0} Y_\ccj^\dg \ket{0} \Big|_{2~\scred}
\label{SOTn=2Check}
\ee
The indicated contraction is superfluous, since $\bra{0} Y_\ccj^\dg \ket{0}$ only has red lines and we are restricting it to only 2 red vertices. The combination of \Eq{SOTstep2} and \Eq{SOTn=2Check} mean that \Eq{SOT} holds for $n=2$.

{\bf Arbitrary n:} It should now be clear how the induction step works: at every step, all of the diagrams in \Eq{pulloneout2} cancel except those with all of the red lines on the $\ccj$th sector. That is, using all of the previous induction steps, this cancellation occurs between all of the contractions of the Wilson lines except those that only involve $Y_\ccj^\dg$. After canceling the terms off, we are left with the result for any $n$. Hence, \Eq{SOT} is proved.

\end{proof}

\subsection{Final steps}
\Eq{SOT} implies that we can strip red lines off sector-by-sector of the general reduced-diagram in \Eq{pinches2}: 
\be
\fd{6cm}{AllLoopPinch.pdf}  
	 \;\overset{\substack{\text{any~$r_\scs$} \\ \text{gen.~$r_\ccc$} }}{\LPFeq}\; 
\fd{6cm}{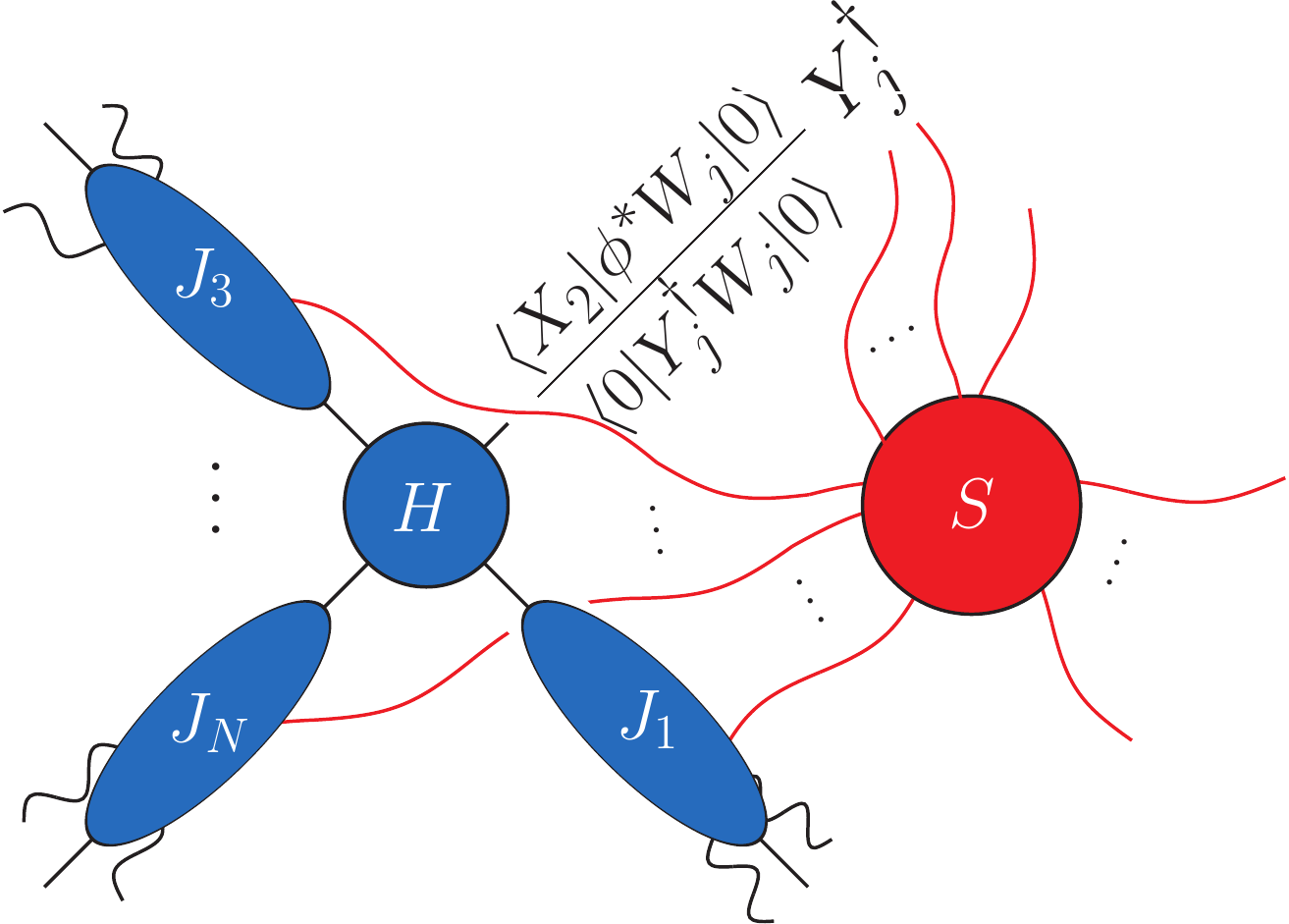}  
\ee
Once the red lines are stripped off of every collinear sector, they connect from the soft Wilson lines, through the $S$ blob, to the external emissions. The $S$ blob gives all possible interactions with the full QCD Lagrangian Feynman rules, so the red lines are exactly described by the matrix element $\bra{X_\scs} Y_\ccO^\dg \cdots Y_\rN \ket{0}$
in QCD. Thus,
\begin{align}
\bra{X_\ccO\cdots X_\rN;X_\scs} \cO \ket{0} \;&
 \;\overset{\substack{\text{any~$r_\scs$} \\ \text{gen.~$r_\ccc$} }}{\LPeq}\; 
	\fd{6cm}{AllLoopPinch.pdf}  
	\hspace{1cm} 
\notag\\[4mm]  \;&
 \;\overset{\substack{\text{any~$r_\scs$} \\ \text{gen.~$r_\ccc$} }}{
~\LPFeq}
~~
\underbrace{
\frac{\bra{X_\ccO} \phi^\star W_\ccO \ket{0} }{ \bra{0} Y_\ccO^\dg W_\ccO \ket{0} } 
}_{\text{blue only, $J_1$}}
\cdots
\underbrace{
\frac{\bra{X_\rN} W_\rN^\dg \phi \ket{0} }{ \bra{0} W_\rN^\dg Y_\rN \ket{0} } 
}_{\text{blue only, $J_N$}}
\times 
\underbrace{
\bra{X_\scs} Y_\ccO^\dg \cdots Y_\rN \ket{0}
 }_{\text{red only, $S$}}
\label{SoftCollFact}
\end{align}
The braces describe which parts of the reduced diagram the indicated quantities reproduce, in physical gauges. Since both sides are gauge invariant, this factorization formula holds in any gauge, even covariant ones.

This completes the proof of hard-soft-collinear factorization. To clean things up, we can drop the  $\LPFeq$ sign in favor of the leading-power equality, $\LPeq$, by adding in the Wilson coefficient. At every stage that we have dropped IR-insensitive loops, they have not contained external emissions by Lemma~\ref{lem:EatPower}, so the Wilson coefficient is still independent of the states, $\bra{X_\ccj}$ and $\bra{X_\scs}$, and only depends on the net momentum in each collinear sector (using the procedure of \Eq{HardSoftProc}). Therefore, we have our final factorization formula:
\be
\boxed{
\bra{X_\ccO\cdots X_\rN;X_\scs} \cO \ket{0} 
	\;\LPeq\; 
C(\Sij)\, \frac{\bra{X_\ccO} \phi^\star W_\ccO \ket{0}}{\bra{0} Y_\ccO^\dg W_\ccO \ket{0}} \,\cdots\, \frac{\bra{X_\rN} W_\rN^\dg\phi \ket{0}}{\bra{0} W_\rN^\dg Y_\rN \ket{0}}\, \bra{X_\scs} Y_\ccO^\dg \cdots Y_\rN \ket{0}
}
\label{finalfact}
\ee

\section{General scattering amplitudes}

\label{sec:GenHardScat}
So far, we have discussed factorization for matrix elements of local operators. None of the arguments given to derive the structure of the reduced diagram in \Eq{pinches2}
actually require the scattering to be mediated by a single operator.
In calculating a general scattering matrix element, any line that cannot go on-shell cannot be IR sensitive. Thus off-shell lines can be included in
the hard amplitude of the reduced diagram and absorbed into the Wilson coefficient. 

For example, we have already shown that matrix elements for the operator $|\phi|^2$ between the vacuum and final states $\bra{X_\ccTH X_\ccF; X_\scs}$ factorize as
\be
\bra{X_\ccTH X_\ccF;X_\scs} \phi^\star \phi \ket{0}
 \LPeq 
\cC_{| \phi |^2}(S_{\ccTH\ccF}) 
\; \frac{\bra{X_\ccTH} \phi^\star W_\ccTH \ket{0}}{ \bra{0} Y_\ccTH^\dg W_\ccTH \ket{0} } 
\; \frac{\bra{X_\ccF} W_\ccF^\dg \phi \ket{0}}{ \bra{0} W_\ccF^\dg Y_\ccF\ket{0} } 
\; \bra{X_\scs} Y_\ccTH^\dg Y_\ccF \ket{0}
\ee
where $S_{\ccTH\ccF}=(P_\ccTH+P_\ccF)^2$ and $\cC_{| \phi |^2}(S_{\ccTH\ccF})=1$ at tree level. Let us compare this to $\gamma\gamma \to \phi\phi^\star$ in scalar QED. At tree level, three diagrams contribute:
\be
\fd{1.8cm}{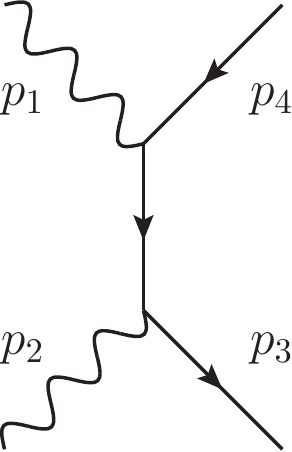}
 \quad+\quad 
\fd{1.8cm}{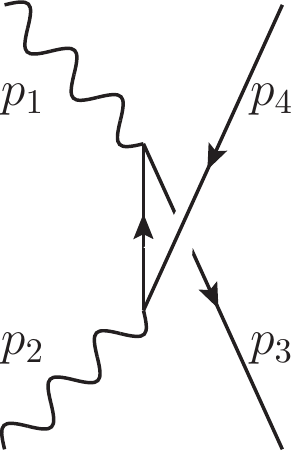}
 \quad+\quad 
\fd{1.8cm}{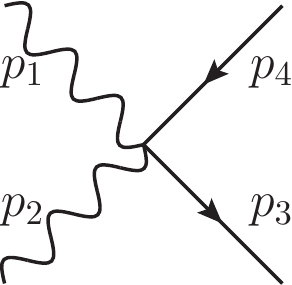}
\label{somehardproc}
\ee
Due to the off-shell lines, this  amplitude cannot be written exactly as the matrix element of a local operator. On the other hand, since
the lines are off-shell, we can still factorize the amplitude for $\gamma \gamma \to \bra{X_\ccTH X_\ccF; X_\scs}$ as   
\be
\braket{X_\ccTH X_\ccF;X_\scs |
\epsilon^\mu(p_\ccO); \epsilon^\nu(p_\ccT)}
 \LPeq
\epsilon_\mu^\ccO \epsilon_\nu^\ccT \, \cC_{\gamma\gamma\phi\phi^\star}^{\mu\nu}(S_{\cci\ccj}) 
\; \frac{\bra{X_\ccTH } \phi^\star W_\ccTH  \ket{0}}{ \bra{0} Y_\ccTH^\dg W_\ccTH\ket{0} } 
\; \frac{\bra{X_\ccF} W_\ccF^\dg \phi \ket{0}}{ \bra{0} W_\ccF^\dg Y_\ccF\ket{0} } 
\; \bra{X_\scs} Y_\ccTH^\dg Y_\ccF \ket{0}
\ee
with
\be
\cC^{\mu\nu}_{\gamma\gamma\phi\phi^\star } (s_{\cci\ccj}) =  
e^2\left[\frac{(2p_\ccF^\mu - p_\ccO^\mu)(p_\ccT^\nu - 2 p_\ccTH^\nu)}{(p_\ccO - p_\ccF)^2}
 +  \frac{(2p_\ccTH^\mu - p_\ccO^\mu)(p_\ccT^\nu - 2 p_\ccF^\nu)}{(p_\ccTH - p_\ccO)^2} 
+2 i g^{\mu\nu}  \right]
\ee
at tree level.

At higher orders, the Wilson coefficients $\cC_{| \phi |^2}$ and $\cC^{\mu\nu}_{\gamma\gamma\phi\phi^\star }$ will get different radiative corrections, 
but the jet and soft sectors of the factorized processes are identical.  The all-orders definitions of the Wilson coefficients are
\be
\quad\; \cC_{| \phi |^2}(Q) 
= \frac{\bra{\phi,p_\ccTH;\phi^\star,p_\ccF} \phi^\star\phi \ket{0}}
{
\dfrac{\bra{\phi,p_\ccTH} \phi^\star W_\ccTH \ket{0}}{ \bra{0} Y_\ccTH^\dg W_\ccTH\ket{0} } 
\, \dfrac{\bra{\phi^\star,p_\ccF} W_\ccF^\dg \phi \ket{0}}{ \bra{0} W_\ccF^\dg Y_\ccF\ket{0} } 
\, \bra{0} Y_\ccTH^\dg Y_\ccF \ket{0}
}
\ee
and
\be
\cC_{\gamma\gamma\phi\phi^\star}(Q) 
= \frac{\braket{\phi,p_\ccTH;\phi^\star,p_\ccF |
\epsilon^\mu(p_\ccO); \epsilon^\nu(p_\ccT)}}
{
\dfrac{\bra{\phi,p_\ccTH} \phi^\star W_\ccTH \ket{0}}{ \bra{0} Y_\ccTH^\dg W_\ccTH\ket{0} } 
\, \dfrac{\bra{\phi^\star,p_\ccF} W_\ccF^\dg \phi \ket{0}}{ \bra{0} W_\ccF^\dg Y_\ccF\ket{0} } 
\, \bra{0} Y_\ccTH^\dg Y_\ccF \ket{0}
}
\ee
In either case, the Wilson coefficient only depends on the type of scattering and not on distribution of soft and collinear radiation in 
the external states $\bra{X_\ccTH X_\ccF;X_\scs}$.
Thus, we see that the factorization arguments given in this paper apply to any type of scattering process in any gauge theory as long as the external states 
contain only soft and collinear degrees of freedom.

Factorization holds with identical arguments when there are collinear particles
in the initial state, with the only change that the Wilson lines become incoming (see \tree). The situation where particles in the initial state are collinear
to particles in the final state are explicitly {\it excluded} from our formulation. In particular, general hadron-hadron scattering is not described if there
are spectator partons with significant energy. The formula does apply to the special case of threshold hadron-hadron scattering, where the partonic center-of-mass
is close to the machine energy so the spectator partons are necessarily soft. Expanding around this limit has proved useful in both total-cross-section calculations~\cite{Ahrens:2008nc,Ahrens:2010zv} and jet
shape calculations at hadron colliders~\cite{Kidonakis:1998bk,Laenen:1998qw,Becher:2006mr,Becher:2009th,Kelley:2010qs,Chien:2012ur}.

\section{QCD}
\label{sec:QCD}

All of the arguments in the proof of hard-soft-collinear factorization are completely general. They apply to any renormalizable Abelian or non-Abelian gauge theory with any
matter content. The change in going from scalar QED  to QCD  essentially amounts to pinpointing where the color indices go. 
We will use $\suh_\cci$
for fundamental color indices and $\sua,\sub,\cdots$ for adjoint indices,
with $\cci$ and $\ccj$ still denoting jet directions.

\subsection{Jet amplitudes}

To add in the color contractions, we trace back through the soft-collinear factorization discussion, 
replacing scalars with quarks.
Eq.~\eqref{blueonly} becomes
 \be
\fd{3cm}{SingleCollSect.pdf} =
\bra{X_\ccj} \bar{\psi} \ket{0}^\suh \Big|_\text{blue only} 
\label{blue2}
\ee
Here the $\suh$ color index comes from the net color of the state $\bra{X_\ccj}$ that exits the jet blob on the left. Now, recall that in factorization gauge with $r_\scs=p_\ccj$ no soft sensitive lines can attach to the $\ccj$-collinear sector, which led to Eqs. \eqref{Wbootrt} and \eqref{blueonly}. In QCD these equations become
 \be
\fd{3cm}{SingleCollSect.pdf} 
	\overset{\substack{r_\scs\, = \,p_\ccj \vspace{.5mm}\\ \text{gen. }r_\ccc}}{~\LPFeq}
 \bra{X_\ccj} \bar{\psi} \ket{0}^{\suhp} \left[\frac{1}{ \bra{0} Y_\ccj^\dg  \ket{0} }\right]^{\suhp \suh}
 \overset{\text{gen. }r_\ccc}{~\LPFeq}
 \bra{X_\ccj} \bar{\psi}W_\ccj \ket{0}^{\suhp} \left[\frac{  1}{ \bra{0} Y_\ccj^\dg W_\ccj \ket{0} }\right]^{\suhp \suh}
\ee
One can think of $W_\ccj$ as bringing color $\suhp$ in from infinity to the origin along the $t_\ccj$ direction. 
Now the vacuum is gauge invariant, so 
\be
\bra{0} Y_\ccj^\dg W_\ccj\ket{0}^{\suh \suhp} = \frac{1}{N_c} \tr \bra{0} Y_\ccj^\dg W_\ccj \ket{0} \; \delta^{\suh \suhp}
\ee
and therefore
\be
\fd{4cm}{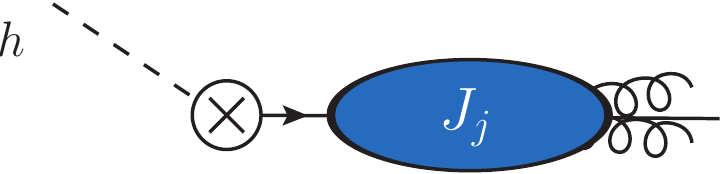} 
 \;\overset{\text{gen. }r_\ccc}{~\LPFeq}\;
N_c \frac{  \bra{X_\ccj} \bar{\psi}W_\ccj \ket{0}^{\suh}}{\tr \bra{0} Y_\ccj^\dg W_\ccj \ket{0} }
\ee

Similarly, the sprig-of-thyme equation, Eq. \eqref{SOT} for a quark jet becomes
\be
\sum_{ \text{perms on \ccj} } \;
\fd{4.3cm}{AllLoopThyme.pdf} 
\;\overset{\text{gen-}r_\ccc\;}{\;\LPFeq}\;
 \frac{  \bra{X_\ccj} \bar{\psi}W_\ccj \ket{0}^{\suhp}}{\tr \bra{0} Y_\ccj^\dg W_\ccj \ket{0} }
\times \big(Y_\ccj^\dg\big)^{\suhp \suh} \Big|_{n}
\label{SOTquark}
\ee
With the $N_c$ factor implicitly absorbed into the Wilson coefficient (by the definition of $\LPFeq$).
Pulling $n$ gluons out of the soft Wilson line gives a series of $\suTT^\sua$ matrices which multiply through to convert $\suhp$ to $\suh$.
The color indices on the soft Wilson line represent a matrix which transforms the color coming out of the hard process due to the soft radiation. It
is, of course, highly nontrivial that the color within the jet is manipulated only by $\bar\psi$ and $W_\ccj$ and the color of the soft radiation is
manipulated only by $Y_\ccj$, with the two not interacting. It is also true, since the soft radiation only senses the net color charge of the collinear radiation. This follows
from our proof because in $r_\scs= p_\ccj$ the soft radiation comes from everywhere else in the event (which has the opposite color charge as the jet). All of the manipulations
we did to prove soft-collinear factorization used only gauge invariance and that in the soft limit, gluon emissions are reproduced by the matrix element of a path-ordered Wilson line 
(a fact both well-known and proven in \tree).

The sprig-of-thyme for gluon jets is similar, but involves adjoint Wilson lines, $\cY_\ccj$ and $\cW_\ccj$ defined in \Eq{Yadjdef}. The equivalent of Eq.~\eqref{SOTquark} with adjoint vector fields is
\be
\sum_{ \text{perms on \ccj} } \;
\fd{4.3cm}{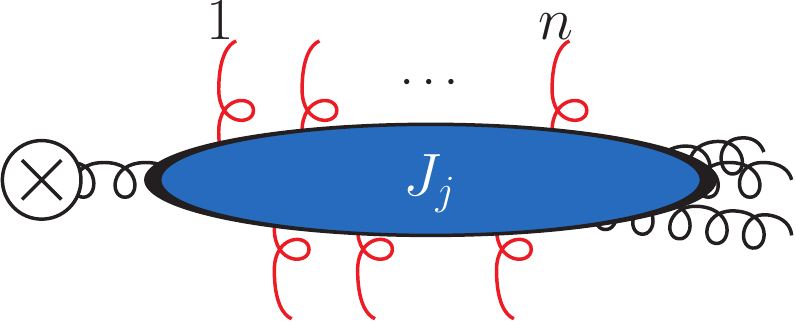} 
\;\overset{\text{gen-}r_\ccc\;}{\;\LPFeq}\;
 \frac{\bra{X_\ccj} A^\mu \cW_\ccj \ket{0}^{\sub}}
{ \tr  \bra{0} \cY_\ccj^\dg \cW_\ccj \ket{0}  }
\times \big(\cY_\ccj^\dg \big)^{\sub \sua} \,\Big|_{n}
\label{SOTglue1}
\ee
where $\tr\, \delta^{\sua \sub} = d(\adj) = N_c^2-1$ is again dropped.
Note that adjoint Wilson lines are not themselves Hermetian, despite the fact that the adjoint representation is real. Conjugating a path-ordered Wilson lines
reverses the order of the matrices. Thus, the correct relation between an adjoint Wilson line and its conjugate is $(\cY^\dg)^{\sua \sub} = \cY^{\sub \sua}$.

Although $A_\mu^\suc \cW_\ccj^{\suc\sub}$ is the obvious adjoint-version of $\bar{\psi} W_\ccj$, it is somewhat jarring to see  an operator
with a raw gauge field instead of covariant derivatives. Of course, since any matrix element of a color-singlet operator will satisfy the Ward identity, any factorized expression
containing  $A_\mu^\sua$ will also satisfy the Ward identity. It is nevertheless sometimes useful to rewrite the gluon jet function in terms of covariant derivatives.

If the original operator has $A_\mu^\sua$ in a covariant derivative in the fundamental representation, such as $\cO = \bar{\psi} \slashed{A} \psi$, then $A_\mu^\sua$ will come accompanied
by a $\suTT^\sua$. Thus there will be a $\suTT^\sua_{\suh \suhp}$ contracted with the $\sua$ index in \eqref{SOTglue1}, with $\suh$ and $\suhp$ contracted elsewhere in the factorized
expression. Now, use $Y_\ccj^\dg  \suTT^\sua\, Y_\ccj =\cY_\ccj^{\sua\sub}\, \suTT^\sub$, as in  \Eq{YYrel}, $(\cY^\dg)^{\sua\sub} = \cY^{\sub\sua}$, and $\tr [\suTT^\sua \suTT^\sub] = T_F \delta^{\sua\sub}$,
we find
\begin{align}
 \cW_\ccj^{\sua\sub}\big(\cY_\ccj^\dg\big)^{\sub\suc}  T_{\suh \suhp}^\suc 
&= \cW_\ccj^{\sua\sub} \big( Y_\ccj \suTT^\sub Y_\ccj^\dg \big)_{\suh\suhp}\nn
\\
&= T_F^{-1} \tr[ \suTT^\sua \suTT^\suc] \cW_\ccj^{cb}  \big( Y_\ccj \suTT^\sub Y_\ccj^\dg \big)_{\suh\suhp}
\notag\\&= T_F^{-1} \tr[  \suTT^\sua W_\ccj \suTT^\sub W_\ccj^\dg ]  \big( Y_\ccj \suTT^\sub Y_\ccj^\dg \big)_{\suh\suhp}
\label{ConvertRep}
\end{align}
Since the Ward identity must be satisfied in any process we consider, replacing $A_\mu \to \partial_\mu$ gives zero. Thus,
we can replace $i g_s A_\mu^\sua\tr[  \suTT^\sua W_\ccj \suTT^\sub W_\ccj^\dg ] \to \tr[ W_\ccj^\dg  D_\mu W_\ccj \suTT^\sub ]$. Therefore,
 converting the denominator with similar manipulations to those in \Eq{ConvertRep} and absorbing $i g_s$ into the Wilson coefficient, we can write
\be
\sum_{ \text{perms on \ccj} } 
\fd{4.3cm}{AllLoopThyme_glu.pdf} 
\times \suTT^\sua_{\suh \suhp}
\overset{\text{gen-}r_\ccc\;}{\;\LPFeq} 
\frac{\tr \bra{X_\ccj} W_\ccj^\dg D_\mu W_\ccj \suTT^\sub \ket{0} }
{ \tr \bra{0} W_\ccj^\dg (Y_\ccj \suTT^\sua Y_\ccj^\dg) W_\ccj \suTT^\sua \ket{0}  }
\big( Y_\ccj \suTT^\sub Y_\ccj^\dg \big)_{\suh \suhp} \,\Big|_{n}
\label{SOTglue2}
\ee
Jet amplitudes in this form are occasionally useful since they manifest gauge invariance and only have Wilson lines in the fundamental representation.

\subsection{Example factorization formulas}
To write down the factorization formula in QCD for some process, we simply combine
copies of Eqs. \eqref{SOTquark} and  \eqref{SOTglue2} for each quark or gluon jet direction
and contract the loose soft-Wilson lines  with the soft-sector final state.
For example, a vector boson decaying to 3 jets can be mediated by a  hard-scattering
operator of the form
\be
\cO = \bar\psi \slashed{D} \psi 
\ee
The associated factorization formula is, in gluon-jet notation
\begin{multline}
\bra{X_\ccO X_\ccT X_\ccTH;X_\scs} \bar\psi \slashed{D} \psi \ket{0}
\\[2mm]
\LPeq
\cC(\Sij) \; \gamma^\mu_{\spincol \alpha \beta}
\;
\dfrac
{\bra{X_\ccO} \bar\psi\, W_\ccO\ket{0}^{ {\spincol \alpha}\, {\suh_\ccO}} }
{ \tr\bra{0} Y_\ccO^\dg W_\ccO \ket{0}}
\;
\dfrac{\bra{X_\ccj} A^\mu \cW_\ccj \ket{0}^{\sua}}
{ \tr  \bra{0} \cY_\ccj^\dg \cW_\ccj \ket{0}  }
\;
\dfrac
{\bra{X_\ccTH} W_\ccTH^\dg\,\psi \ket{0}^{ {\spincol \beta} \, {\suh_\ccTH}}}
{\tr\bra{0} W_\ccTH^\dg Y_\ccTH \ket{0}}
\bra{X_\scs} Y_\ccO^\dg \cY_\ccT^{\dg \sua \sub} {\suncol T}^\sub \, Y_\ccTH \ket{0}^{\suncol {h_\ccO   h_\ccTH}} \hspace{5mm}
 \end{multline}
or, representing the gluons with covariant derivatives,
 \begin{multline}
\bra{X_\ccO X_\ccT X_\ccTH;X_\scs} \bar\psi \slashed{D} \psi \ket{0}
\\[1mm] \hspace{-.5cm}\LPeq
\cC(\Sij) \; \gamma^\mu_{\spincol \alpha \beta}
\;
\dfrac
{\bra{X_\ccO} \bar\psi\, W_\ccO\ket{0}^{ {\spincol \alpha}\, {\suncol h_\ccO}} }
{ \tr\bra{0} Y_\ccO^\dg W_\ccO \ket{0}}
\;
\dfrac
{\tr \bra{X_\ccT} W_\ccT^\dg D_\mu W_\ccT {\suncol T^a} \ket{0} }
{ \tr \bra{0} W_\ccT^\dg (Y_\ccT {\suncol T^b} Y_\ccT^\dg) W_\ccT {\suncol T^b} \ket{0}  }
\;
\dfrac
{\bra{X_\ccTH} W_\ccTH^\dg\,\psi \ket{0}^{ {\spincol \beta} \, {\suncol h_\ccTH}}}
{\tr\bra{0} W_\ccTH^\dg Y_\ccTH \ket{0}}
\\[3mm]
 \times
\bra{X_\scs} Y_\ccO^\dg \,Y_\ccT {\suncol T^a} Y_\ccT^\dg \, Y_\ccTH \ket{0}^{\suncol {h_\ccO h_\ccTH}} \hspace{5mm}
\end{multline}
where ${\spincol \alpha}$ and ${\spincol \beta}$ are  Dirac spin indices, ${\suncol a}$ and ${\suncol b}$ are adjoint color indices and ${\suncol h}_\cci$ are fundamental color indices. 
To reduce clutter, the $N_c$ and $N_c^2-1$ factors from the traces have been absorbed into the Wilson coefficient; to put them back one only needs to divide each zero bin by the dimension of the representation of that sector.

There may be multiple operators contributing to a single hard process. For example, in $u d \to u d$ scattering, there are two relevant hard operators~\cite{Kelley:2010fn}:
\be
\cO_1 = (\bar{u} {\suncol T^a} \gamma^\mu u)( \bar{d} {\suncol T^a} \gamma^\mu d) , 
\qquad
\cO_2 =( \bar{u} \gamma^\mu u)( \bar{d} \gamma^\mu  d), 
\ee
where the parentheses indicate color contractions. 
For $u d \to u d$ at tree level in QCD, only a single-gluon exchange is relevant and so $\cO_2$ is not. At 1-loop and beyond, both operators are important to correctly reproduce the hard scattering. 
As in this paper we have avoided configurations where incoming and outgoing partons can be collinear, the factorization formula has only
been shown to hold in threshold kinematical regimes where there is no phase space for hard initial state radiation to end up in the final state~\cite{Kidonakis:1998bk,Laenen:1998qw,Becher:2006mr,Becher:2009th,Kelley:2010qs,Chien:2012ur}.
Alternatively, one could think of the factorization formula in this case mediating a decay, like
$h \to \bar{u} u \bar{d} d$ rather than a scattering process. Factorization for 4-parton scattering
was also studied in~\cite{Sen:1982bt}.

To study $u d \to u d$ near threshold is helpful to have somewhat more general notation.
Labeling the hard partons as $\ccO$, $\ccT$, $\ccTH$, and $\ccF$, the relevant operators are
\begin{equation}
 \mathcal{O}_{\suI\gGG \gGG' } =
 ( \bar{q}_\ccF {\suncol T_I} \gamma_{\mu} \gG q_\ccT ) 
 ( \bar{q}_\ccTH {\suncol T_I} \gamma^{\mu} \gG'q_\ccO )  \,.
\end{equation}
Here, $\suI$ indexes the color structure ($\suTT_\suO= \suTT^\sua$ or $\suTT_\suT= {\suncol \mathbf{1}}$),
and $\gG$ and $\gG'$ index the helicity (e.g. $\gG =\gG'= {\spincol P_L} = {\spincol P_+} =  \frac{1}{2}\left( 1 -\gamma_5\right)$).
Helicity and flavor is preserved in QCD, so the helicity of the $u$ fixes the helicity of the $\bar{u}$.
There are thus eight relevant operators, since $\suI = 1,2$, $\gG = \gpm$ and $\gG' = \gpm$. 
Each set of helicities has a separate factorization, but the color structures can mix.

So the matrix element for a 4 quark-jet decay factorizes as
\begin{multline}
\cM_{\gpm \gpm} 
= \cM_{\gpm \gpm}( p_\ccO  + p_\ccT \to X_\ccTH+  X_\ccF + X_\scs)
 \\[2mm]
 \LPeq \sum_{\suI} \cC^{\suI}_{ \gpm \gpm}(S_{\cci\ccj})
\dfrac
{\bra{X_\ccF} \bar\psi_{\ccF} W_\ccF\ket{0}^{\gpm {\suh_\ccF} } }
{ \tr\bra{0} Y_\ccF^\dg W_\ccF \ket{0} }
\;
\dfrac
{\bra{p_\ccT} {\overline W}^\dg_\ccT \psi_{\ccT}\ket{0}^{\gpm {\suh_\ccT} } }
{ \tr\bra{0}  {\overline W}^\dg_\ccT {\overline Y}_\ccT \ket{0} }
\;
\dfrac
{\bra{X_\ccTH} \bar\psi_{\ccTH} W_\ccTH\ket{0}^{\gpm {\suh_\ccTH} } }
{ \tr\bra{0} Y_\ccTH^\dg W_\ccTH \ket{0} }
\;
\dfrac
{\bra{p_\ccO} {\overline W}^\dg_\ccO \psi_{\ccO}\ket{0}^{\gpm {\suh_\ccO} } }
{ \tr\bra{0}  {\overline W}^\dg_\ccO {\overline Y}_\ccO \ket{0} }\\[2mm]
\times
\bra{X_\scs} ( Y_\ccF^\dg {\suncol T_I} {\overline Y}_\ccT )^{ {\suh_\ccF} {\suh_\ccT} }( Y_\ccTH^\dg {\suncol T_I} {\overline Y}_\ccO )^{ {\suh_\ccTH} {\suh_\ccO} } \ket{0}
\hspace{15mm}
\end{multline}
where ${\overline W_i}$ and ${\overline Y_\cci}$ are incoming Wilson lines (see \tree). 
Note that we only write explicitly the color and spin indices of the partons which emerge from the hard scattering. There are many implicit color and spin indices in the states $\bra{X_\ccj}$ and $\bra{X_\scs}$. These colors and spins are important when computing scattering amplitudes, but are usually summed over in computing resummed distributions. 

\subsection{QCD factorization formula}
In summary, a general factorization formula in QCD can be written as
\begin{equation}
  \addtolength{\fboxsep}{3mm}
   \boxed{
   \begin{gathered}
\cM_{ \{ \gpm \} } \LPeq \sum_{\suI} \cC_{\suI, \{\gpm\} } (S_{\cci\ccj} ) 
\hspace{95mm}
\\
\hspace{5mm}
\times \cdots
\dfrac
{\bra{X_\cci} \bar\psi_{ \cci } W_\cci \ket{0}^{\gpm {\suh_\cci}}  }
{ \tr\bra{0} Y_\cci^\dg W_\cci \ket{0} }
\cdots
\dfrac{\bra{X_\ccj} A^\mu \cW_\ccj \ket{0}^{\gpm \sua_\ccj}}
{ \tr  \bra{0} \cY_\ccj^\dg \cW_\ccj \ket{0}  }
\cdots
\dfrac
{\bra{X_{\cck}} W_\cck^\dg\psi_{ \cck } \ket{0}^{\gpm {\suh_\cck}}}
{ \tr\bra{0} W_\cck^\dg Y_\cck \ket{0} }
\cdots\\[2mm]
\hspace{10mm}
\times 
\bra{X_\scs}\cdots  (Y_\cci^\dg {\suncol T_I^i})^{{\suh_\cci}{\suncol l_\cci}}
 \cdots 
( \cY^\dg_\ccj {\suncol T}_\suI^{\suj} )^{ {\suncol l}_{ \ccol j-1} \sua_\ccj {\suncol l}_{ \ccol j+1} }
 \cdots
 ({\suncol T^k_I}  Y_\cck)^{{\suncol l_\cck}{\suh_\cck}}
\cdots \ket{0} 
   \end{gathered}
   }
\label{genfactQCD}
\end{equation}
where the $\gpm$ indexes the helicities.
The ${\suncol l}_\cci$ indices are contracted within the soft Wilson line matrix element, while the $\suh_\cci$ and $\sua_\cci$ indices contract with the colors of the jets.

\section{Splitting functions and soft currents}
\label{sec:applic}

One application of factorization is that it can provide gauge-invariant and regulator-independent definitions of the collinear-sensitive or soft-sensitive
parts of scattering amplitudes. Such definitions may be useful in  perturbative QCD calculations if they help simplify or clarify the structure of the infrared divergences.
We therefore consider the soft and collinear limits of our formulas separately, deriving definitions of
splitting functions and soft currents and thereby proving their universality.

\subsection{Splitting Functions}
\label{sec:SplitFunc}

Suppose we have a state $\bra{X^0} = \bra{X_\ccO^0 \cdots X_\ccN^0; X_\scs^0}$ containing soft and collinear particles and a matrix element $\cM_0$ for producing that state. We want to know how $\cM_0$ is modified into $\cM$ 
by the addition of extra collinear particles to the $\ccj$-collinear sector, turning $\bra{X_\ccj^0}$ into $\bra{X_\ccj}$,
while leaving the net momenta in the $\ccj$ sector unmodified at leading power $P_\ccj^\mu \LPeq P_\ccj^{0\,\mu}$. 
 Let us write the modified
matrix element formally as some operator acting on the original matrix element
\be
\cM = \Sp \cdot \cM_0 
\ee
The distribution of the soft radiation in $\bra{X_\scs^0}$ is completely independent of the splitting. The only
modification from the addition of collinear particles to $\bra{X_\ccj^0}$ is in the matrix element associated with the $\ccj$-collinear sector. 

The factorization formulas for $\cM_0$ and $\cM$ are  almost identical. The relevant parts of the factorization formulas are:
 \begin{align}
 \cM_0 \LPeq
\dfrac
{\bra{X^0_\ccj} \bar\psi \, W_\ccj \ket{0}^{\gpm {\suh_\ccj}}  }
{ \tr\bra{0} Y^\dg_\ccj W_\ccj \ket{0} } \cdot 
   \cM_{{\black {\text{rest}}}}^{\suh_\ccj}
\;,
\qquad
 \cM \LPeq
\dfrac
{\bra{X_\ccj} \bar\psi \, W_\ccj \ket{0}^{\gpm {\suh_\ccj}}  }
{ \tr\bra{0} Y^\dg_\ccj W_\ccj \ket{0} }\cdot 
   \cM_{{\black {\text{rest}}}}^{\suh_\ccj}
\;
 \end{align}
Now, the spin of each collinear sector, that is, the helicity of the nearly-on-shell particle coming out of the hard vertex,
in $\cM$ must be the same as in $\cM_0$ for the two to be related. 
So let us fix this helicity $\gpm$ and drop the spin indices. 
Then we can write
\be
\Sp (X_\ccj, X_\ccj^0)^{ {\suh \suhp}} = 
\dfrac
{\bra{X_\ccj} \bar\psi \, W_\ccj \ket{0}^{\suh} }
{\bra{X^0_\ccj} \bar\psi \, W_\ccj \ket{0}^{\suh'} }
\label{QuarkSplit}
\ee
The notation here indicates that the splitting functions are operators in color space. Note that the zero-bin subtractions from the denominator of the general factorization formula have dropped out. These denominators are $1$ in dimensional regularization, but here we see that they play no role with any regulator. As we will see, this is also true for soft currents.

To  convert  Eq.~\eqref{QuarkSplit} into something more practical, let us work out a simple example,
following Section 9.1 of \tree.
We take
$\bra{X^0}$ to have a single right-handed antiquark in it with momentum $P^\mu$ and color ${\suh}$: 
 $\bra{X^0} = \bra{\bar{u}_{\suh}(P)}$. In terms of spinor helicities,  this state is $[ P$ and at tree level and
\be
  \cM_0^{ {\suh} {\spincol R}} \LPeq
   [P
   \cM_{{\black {\text{rest}}}}^{\suh} ]
\ee
 We take  take $\bra{X}$ to have a right-handed antiquark of momentum $p^\mu \LPeq z P^\mu$ and a single gluon with  momentum $q^\mu \LPeq (1-z) P^\mu$ with color $\sua$ and helicity ${\gpm}$. If the gluon helicity is ${\spincol -}$, the modified amplitude is (see \tree)
\begin{align}
\cM^{ {\suh \sua}{\spincol R - }}
&= g_s\, \frac{ \sqrt{2}}{[qp]} \frac{z}{\sqrt{1-z}} \,\,
[P \suTT^\sua_{\suh \suhp}
   \cM_{{\black {\text{rest}}}}^{\suhp} ] 
\end{align}
Thus the tree-level splitting function for a ${\spincol -}$ helicity gluon is
\be
\Sp_{{\spincol R-} }
(p,q) =  g_s\, \frac{ \sqrt{2}}{[qp]} \frac{z}{\sqrt{1-z}}  
\TT_\ccj
\ee
For a  ${\spincol +}$ helicity  gluon, the tree-level splitting function is also extractable from \tree:
\be
\Sp_{{\spincol R+} }^{{\suh \suhp}}(p,q) =  g_s  \, \frac{ \sqrt{2}}{\l pq\r}\bigg( \frac{z}{\sqrt{1-z}} + \sqrt{1-z} \bigg)
\suTT^{{\sua}}_{\suh \suhp}
\ee
These splitting functions can be calculated to higher order using Eq.~\eqref{QuarkSplit}. 

The gluon splitting functions are similar
\be
\Sp^{\mathbf g} (X_\ccj, X_\ccj^0)^{ {\suncol a b}} = 
\dfrac
{\bra{X_\ccj} W_{\ccj}^\dg A_\mu W_\ccj \ket{0}^{\suncol a} }
{\bra{X^0_\ccj} W_{\ccj}^\dg A_\mu W_\ccj \ket{0}^{\suncol b} }
\label{GluonSplit}
\ee
The universality of \Eq{QuarkSplit} and \Eq{GluonSplit} to all orders for any process is proven by our factorization theorem.

\subsection{Soft currents}
\label{sec:SoftCurr}
The equivalent of splitting functions for soft radiation are often called soft currents~\cite{Bassetto:1984ik}. Extracting their matrix-element definition from the general
factorization formula proceeds in the same way as for collinear splittings. 

Suppose we have a state $\bra{X^0} = \bra{X_\ccO^0 \cdots X_\ccN^0;X_\scs^0}$ containing soft and collinear particles and a matrix element $\cM_0$ for producing that state. We want to know how $\cM_0$ is modified into $\cM$ 
by the addition of extra soft particles $\bra{X_\scs}$. The modified
matrix element can be formally written as
\be
\cM = \JJ \cdot \cM_0 
\ee
where $\JJ$ is an operator acting in color space. Isolating the part of the factorization formula involving soft radiation, it follows that
\be
\JJ
= \frac{ 
\bra{ X_\scs} 
 Y_\ccO^{\dg} \cdots \TT_\suI \cdots  Y_\ccN \ket{0} }
{ \bra{X_\scs^0} Y_\ccO^{\dg} \cdots \TT_\suI \cdots Y_\ccN\ket{0} }
\ee
Here $\suI$ indexes the color structures of the relevant operators.

 $\JJ$ has implicit indices which also act on the color of the particles in $\bra{X_\ccO \cdots X_\ccN}$.
It is standard to write $\JJ$ as a function of color-charge operators $\TT_\ccj^\sua$ which act in color space as the $\mathrm{SU}(3)$ generator in the representation
of net color flowing in direction $\ccj$. This representation is of course the same as the representation of the $Y_\ccj$ Wilson line. When using color-charge operators,
one never needs to perform a color sum, and so there is, trivially, no dependence of $\JJ$ on the color structure $\suI$. That the matrix element for soft emission
{\it only} depends on the net color in each collinear sector, and not how that color is distributed, is a nontrivial consequence of factorization. It was proven to
1-loop by direct computation in~\cite{Catani:2000pi}, and now we have show that it holds to all orders in $g_s$, for an arbitrarily complicated collinear sector
and any number of hard particles. 

In the simplest case, $\bra{X_\scs^0} = \bra{0}$ and  $\bra{X_\scs}$ has only one gluon, with momentum $q$, polarization $\epsilon^\mu(q)$ and color $\sua$. Then $\JJ = \epsilon_\mu \JJ_\sua^\mu$.  
At tree level, $\JJ$ is:
\be
\JJ^{\mu (0) } = g_s  \sum_{\ccj=1}^m \TT_\ccj
 \frac{p_\ccj^\mu}{p_\ccj \cdot q} 
\ee
where $\TT_\ccj$ is the color-charge operator in the $\ccj$ direction. 
To be more concrete, if there is only a quark and anti-quark jet, then
\be
\JJ^\mu  = J^\mu_{\sua \suh \suhp}
= \frac{ 
\bra{ \epsilon^{ { \spincol \mu}} (p); \sua} 
 Y_\ccO^{\dg} Y_\ccT \ket{0}^{ \suh \suhp}}
{ C_A \tr \bra{0} Y_\ccO^{\dg} Y_\ccT\ket{0} }
=g_s \suTT^{\sua}_{\suh\suhp} \left( \frac{p_\ccO^\mu}{p_\ccO \cdot q} - \frac{p_\ccT^\mu}{p_\ccT \cdot q}\right) + \cdots
\ee
The $\suh$ and $\suhp$ color indices act on the jets, ${\bra{X_\ccO} \bar\psi \, W_\ccO \ket{0}^\suh }{\bra{X_\ccO} W_\ccT^\dg \psi \ket{0}^\suhp }$.

In dimensional regularization in $4-2\varepsilon$ dimensions, with outgoing particles only, the 1-loop current is~\cite{Catani:2000pi}:
\begin{multline}
\JJ^{\mu (1) } =-\frac{1}{16\pi^2} \frac{1}{\varepsilon^2} \frac{\Gamma^3(1-\varepsilon) \Gamma^2(1+\varepsilon)} { \Gamma(1-2\varepsilon)} 
\\
\times i f_{ \sua\sub \suc} \sum_{\cci\ne \ccj} \TT^\sub_\cci \TT^\suc_\ccj 
\left( \frac{p_\cci^\mu}{p_\cci \cdot q} - \frac{p_\ccj^\mu}{p_\ccj \cdot q} \right)
\left[ \frac{- 4 \pi p_\cci \cdot p_\ccj}{2 (p_\cci \cdot q) (p_\ccj \cdot q)} \right]^\varepsilon
\end{multline}
In calculating this current, Catani and Grazzini were able to prove that it it is independent of the momenta and color-flow of the process at 1-loop. 
As noted above, our proof generalizes this observation to all orders. Of course, the factorization formula does not help in actually calculating the soft
current in dimensional regularization. The current for one soft gluon emission at 2 loops can be found in~\cite{Duhr:2013msa,Li:2013lsa}.

Another familiar result that can be deduced from our all-orders definition of the soft current is that of Abelian exponentiation. Namely, that in an Abelian gauge theory, the soft current is exact at tree level. This follows simply from the fact that in an Abelian theory, the contraction of a Wilson line with the external state can be pulled out of the rest of the matrix element. Since the Wilson lines are exponentials, pulling out a contraction leaves behind the same Wilson line (just like taking a derivative), so that, to all orders in perturbation theory:
\be
J^\mu_\text{Abelian}
= \frac{ 
\bra{ q} 
 Y_\ccO^{\dg} \cdots Y_\ccN \ket{0} }
{ \bra{0} Y_\ccO^{\dg} \cdots Y_\ccN\ket{0} }
= 
\frac{ \bra{0} Y_\ccO^{\dg} \cdots Y_\ccN \ket{0} } { \bra{0} Y_\ccO^{\dg} \cdots Y_\ccN\ket{0} }
\sum_{\ccj=\ccO}^\ccN \bra{q} Y_\ccj \ket{0}_\text{tree}
= \sum_{\ccj=\ccO}^\ccN Q_\ccj \frac{n_\ccj^\mu}{n_\ccj\cdot q}
\ee
where $Q_\ccj$ is the QED charge: $Q_\ccj = e$ if it comes from a $Y_\ccj^\dg$ and $Q_\ccj = - e$ if it comes from a $Y_\ccj$. Gauge invariance implies that $\sum_{\ccj=\ccO}^\ccN Q_\ccj = 0$.

\section{Effective Field Theory}
\label{sec:SCET}
In this paper, our emphasis has been on factorization in QCD at the amplitude level.
In our view, working at the amplitude level, rather than at the amplitude-squared level as is often done, makes some elements of factorization more transparent.
It also elucidates some aspects of Soft-Collinear Effective Theory (SCET).

Consider \Eq{genfactQCD}, which we have proven to leading power in $\lambda$. Let us assign particles in each collinear sector $\bra{X_\ccj}$ the quantum number  $\ccj\in\{\ccO,\ldots,{\ccol{ N}}\}$ and each particle in  $\bra{X_\scs}$ the quantum number $\scs$. Let us also write an effective Lagrangian
that is $N+1$ copies of the QCD Lagrangian
\be
\cL_\text{eff} = \cL_{\softcol \text{soft}} + \sum_{\ccj=1}^{N} \cL_\ccj
\ee
with fields in each sector only creating and annihilating states with the appropriate quantum numbers. 
Then we can combine the numerator matrix elements in \Eq{genfactQCD} into a single matrix element in a trivial way. 

For example,  with two collinear sectors, the factorization formula becomes
\be
\bra{X_\ccO X_\ccT;X_\scs} \bar\psi \gamma^\mu \psi \ket{0} \LPeq \cC_2
\bra{X_\ccO X_\ccT;X_\scs}
\frac{\bar\psi_\ccO \, W_\ccO}{\tr\bra{0} Y_\ccO^\dg W_\ccO \ket{0}/N_c} 
Y_\ccO^\dg \gamma^\mu Y_\ccT
\frac{ W_\ccT^\dg \, \psi_\ccT }{\tr\bra{0} W_\ccT^\dg Y_\ccT \ket{0}/N_c}
\ket{0}_{\cL_\text{eff}}
\label{SCETEq1}
\ee
if computed with an effective Lagrangian
\be
\cL_\text{eff} = \cL_{\softcol \text{soft}} + \cL_\ccO  + \cL_\ccT
\ee
The Wilson coefficient $\cC_2$ depends only on the net momenta $P^\mu_1$ and $P_2^\mu$ in each sector, not on the detailed distribution of momenta in  $\bra{X_\ccO X_\ccT;X_\scs}$.
Since
 $\cC_2$
 depends on the hard-scattering operator and not the states, it is a legitimate Wilson coefficient from matching onto an effective field theory.

It is possible to clean up the effective field theory operator a little. Let us define
 \be
\IRZ_\cci \equiv \frac{1}{N_c} \tr\bra{0} W^{\dg}_\cci Y_\cci \ket{0} 
\ee
For other color representations, $\IRZ_\cci$ is defined similarly with the Wilson lines in the appropriate representation and  $N_c$ replaced by dimension of the representation.
The $\IRZ_\cci$ factors are both UV and IR divergent. They are, however, independent of $\lambda$ and any momenta in the process. That is, for given UV and IR regulators, they are power series in 
$\alpha_s$. Thus, they can play the role of a kind of field-strength renormalization for jets.
Indeed, it is natural to define jet fields as
\be
\chi_\cci \equiv  \frac{1}{\IRZ_\cci} W_\cci^\dg \, \psi_\cci
\ee
These composite fields are gauge invariant (up to a global rotation associated with the net color charge of the jet) and are soft insensitive and collinear sensitive only in the $\cci$ direction. In terms of the jet fields, \Eq{SCETEq1} becomes simply
\be
\bar\psi \gamma^\mu \psi \,\LPeq\, \cC_2 \,
\big(\bar\chi_\ccO  Y_\ccO^\dg\big)	 \gamma^\mu \big(Y_\ccT \chi_\ccT\big)
\label{ValidMatchingEq}
\ee
which is a valid leading-power matching equation in an effective theory describing dijet-like states because the Wilson coefficient is IR-insensitive and independent of which external states are used to compute it. Of course, this matching must be done within the r\'egime of validity of the effective theory, which in this case is justified by the factorization theorem that is proved for $N$-jet-like \emph{final} states.\footnote{See \cite{Bauer:2010cc} for an interesting discussion of how this matching equation can break down when certain initial states are used to perform the matching.}

The effective theory that naturally arises from our factorization formula is very different from the traditional formulation of SCET. Consequently, had we started from the traditional formulation of SCET and derived a factorization formula, it would look very different from the one we have proven. In particular, the Lagrangian and Feynman rules would not be those of full QCD and would not give rise to an all-orders full-QCD definition of the soft current and splitting functions.

Transitioning to the effective field theory language is particularly useful when discussing subleading power corrections in $\lambda$. Recent progress has been made toward describing collider-physics observables at subleading power using the formulation of SCET discussed in this section \cite{Freedman:2013vya}.

In \tree,  the tree-level version of this formulation of SCET (without the vacuum-matrix element denominators) was shown to be equivalent to that discussed by Freedman and Luke \cite{Freedman:2011kj}. However, with the all-loop factorization theorem in hand we naturally see arise an all-orders matrix-element definition of the zero-bin subtraction (similar to what was shown in \cite{Idilbi:2007yi,Idilbi:2007ff}). In Freedman and Luke's approach to SCET, the zero-bin is subtracted off using an ad-hoc procedure applied on an integral-by-integral basis that essentially comes from mimicking the procedure of the traditional approach to SCET~\cite{Manohar:2006nz}. In the traditional approach, the zero-bin subtraction arises naturally from the SCET Lagrangian. It instructs us to apply  a soft subtraction to every single collinear line in each Feynman diagram. This is arguably  a more complicated algorithm than dividing by a single gauge-invariant color-coherent vacuum matrix element, as in our factorization formula.

Before moving on, we point out that our factorization formula is derived with fixed external states that come designated as  soft or collinear.  This was the goal of our paper. For particles which power-count as soft or collinear, the factorization theorem holds if they are put in either sector.
 However, to perform phase space integrals in the factorized expression without chopping up phase space, it would be convenient not to  place a hard cutoff between sectors. To achieve this, in the language of Section~\ref{sec:coloring}, the algorithm in Section~\ref{sec:genalg} would need to be modified to color external-collinear particles blue or red. Then when calculating cross sections, we would be able to integrate the collinear states over their entire phase space, including the soft region. Our expectation is that this would be a simple step using the tools at our disposal, and would give a zero-bin of the form of the eikonal-cross-section subtraction used in the QCD literature (see the discussion in \cite{Collins:1989bt,Lee:2006nr}). However, this is outside the goal of our paper and we leave it for future work.

Another feature of our approach to factorization is that we did not have to choose the power counting of the soft emissions to be the same as that of the collinear emissions. For example, we could have used a separate $\lambda_s$ and $\lambda_c$:
\be
k^\mu \text{ soft } \iff k^\mu \sim \lambda_s^2 Q 
\quad\And\quad
q^\mu \parallel p^\mu \iff q\cdot p \sim \lambda_c^2 Q^2 ,\;\; q^0,p^0 \sim Q 
\ee
The factorization theorem holds at leading power  in both $\lambda_s$ and $\lambda_c$.
In fact, one could even take a different $\lambda_c$ in each sector.  Taking $\lambda_s = \lambda_c = \lambda$ and transitioning to an effective theory implies the factorization theorem that is appropriate to what is referred to as $\text{SCET}_{\text I}$ in the literature. If we take instead $\lambda_s^2 = \lambda_c = \lambda$ the factorization theorem still holds. This power-counting is equivalent $k_\text{soft} \sim (\lambda,\lambda,\lambda)$ and $q_\text{coll} \sim (\lambda^2,1,\lambda)$ in lightcone coordinates, which in the SCET literature is considered to be a different effective field theory, known as $\text{SCET}_{\text{II}}$.  The traditional derivation of $\text{SCET}_{\text{II}}$ involves rather involved intermediary matching through $\text{SCET}_{\text{I}}$ \cite{Bauer:2002nz}. 
The factorization theorem presented in this paper is general enough to unify these two SCETs into a 
single framework.

\section{Conclusions}
\enlargethispage{10pt}
In this paper we have formulated and proven to all orders in perturbation theory a precise statement of factorization for scattering amplitudes in QCD, given in \Eq{genfactQCD}.
This formula applies to states with $N$ well-separated jets with any number collinear particles in each jet, $\bra{X_\ccj}$ for $\ccj=\ccO,\ldots, {\ccol N}$, and any amount of soft radiation in any direction, $\bra{X_\scs}$.
Suppressing color and spin indices, the formula for quark jets reads:
\be
\braket{X_\ccO\cdots X_\rN;X_\scs | i } \,\LPeq\, \cC(P_\cci) \,
\frac{ \bra{X_\ccO} \bar\psi W_\ccO \ket{0} }{ \tr \bra{0} Y_\ccO^\dg W_\ccO \ket{0}} \cdots
\frac{ \bra{X_\rN} W_\rN^\dg \psi \ket{0} }{ \tr \bra{0} W_\rN^\dg Y_\rN \ket{0}}
\bra{X_\scs} Y_\ccO^\dg \cdots Y_\rN \ket{0}
\label{ConcResult}
\ee
where $\ket{i}$ is, say, some uncolored initial state and $\cC(P_\cci)$ is an IR-finite function  depending only on the net momenta in each sector $P_\cci^\mu$.
The symbol $\LPeq$ means equality at leading power in $\lambda$, a physical power counting parameter that constrains only the external momenta in the amplitude. The factorization formula actually holds  to leading power in different power counting parameters $\lambda_\scs$ and $\lambda_\ccc^\ccj$ in each sector. It also holds if there are collinear particles in the initial state, as long as no initial state and final state particles are collinear to each other. 

The proof of \Eq{ConcResult} was broken into two steps, which essentially correspond
to hard factorization and soft-collinear factorization. The first step was to determine the structure of the possible graphs that contribute to each type of infrared sensitivity (soft or $\ccj$-collinear) in the matrix element. The structure of the diagrams relevant at leading power are encoded in the reduced diagram (see \Eq{pinches2}), which represents hard factorization in physical gauges. This reduced diagram is similar to reduced diagrams used in the literature to represent the pinch
surface. Indeed, our derivation of hard factorization exploits essentially the same observations
as these traditional approaches. However, the reduced diagrams traditionally used in the literature  are usually defined only for momenta which are exactly $k^\mu =(0,0,0,0)$ or exact proportional to one of the external momenta. In contrast, our reduced diagram represents a precise set of Feynman integrals, defined for all values of external and loop momenta with rules that describe how they are to be calculated. This generalization of the reduced diagram  allows for a clean transition to an amplitude-level factorization formula.

The second step in the proof is to factorize the soft-sensitive from the collinear-sensitive contributions to  matrix elements. This step builds upon the reduced diagram picture and coloring rules which established hard factorization. The all-orders proof of soft-collinear factorization uses the same logic as was used in \tree~for the tree-level proof. In particular, the use of different reference-vector choices used in \tree~is critical also at loop-level. For loops, the reference-vector flexibility must be generalized  to momentum-dependent lightcone-gauge reference-vector choices. We call a gauge with this flexibility {\it factorization gauge}. Within factorization gauge, different choices for the reference vector in the soft region slosh the soft sensitivities around among different colored diagrams within the reduced diagram structure. This  lets us see how  soft sensitivities factorize from collinear sensitivities for any value of the soft and collinear power-counting parameters, $\lambda_\scs$ and 
$\lambda_\ccc$. Once appropriate Wilson lines are added, the final factorization formula is gauge-invariant and applies even in covariant gauges like Feynman gauge. 

\enlargethispage{10pt}

There are many practical applications of factorization, from the universality of splitting functions and soft currents in QCD~\cite{Altarelli:1977zs,Altarelli:1979ub,Bern:1995ix,Kosower:1999xi,Catani:2000pi}, to regulating infrared divergences in fixed-order calculations~\cite{Catani:1999ss,Bern:1998sc,GehrmannDeRidder:2007jk,Currie:2013vh,DelDuca:1999ha}, to the computation of resummed distributions in jet substructure\cite{Feige:2012vc, Dasgupta:2012hg, 
Chien:2012ur,Jouttenus:2013hs,Dasgupta:2013ihk}.
For example, having  gauge-invariant and regulator independent definitions for objects which contain universal soft or collinear singularities may be useful as the basis of subtractions for fixed-order calculations in QCD.
 In many cases, assuming factorization is enough for phenomenological purposes. Having a rigorous proof of factorization of course puts many approximations on firmer footing. But it may also point the way to understanding subtleties of where factorization may break down, such as in the context of forward scattering~\cite{Collins:1988ig,Catani:2011st, Bauer:2008qu, Bauer:2010cc} or non-global logarithms~
 \cite{  Dasgupta:2001sh,  Banfi:2002hw,  Kelley:2011tj,  KhelifaKerfa:2011zu,  Walsh:2011fz,  Kelley:2011aa,  Hornig:2011iu, Kelley:2011ng, Kelley:2012kj,Schwartz:2014wha}.
In both of these cases, our expectation is not that the factorization theorem proven in this paper will immediately resolve the confusions. Instead, we envisage that the physical picture on which the factorization is based,  with an intuitive reduced diagram picture and matrix-element zero-bin subtractions, should be a practical scaffold on which to build a more sophisticated and nuanced picture of factorization in scattering amplitudes.

\section{Acknowledgements}
\label{sec:ackn}
We thank S. Catani, G. Sterman, J. Collins and K. Yan for detailed feedback and H. Georgi, M. Luke and D. Neill for helpful discussions.
The authors are supported in part by grant DE-SC003916 from the Department of Energy.
I.F. is also supported in part by the Natural Sciences and Engineering Research Council of Canada and the Center for the Fundamental Laws of Nature at Harvard University. We  would also like to thank the Erwin Schr\"odinger Institute and the organizers of the ``Jets and Quantum Fields for LHC and Future Colliders'' workshop for their hospitality and the engaging environment. All diagrams in this paper were drawn using Jaxodraw \cite{Vermaseren:1994je,Binosi:2003yf}.

\bibliography{onshellLOOPSv3}
\bibliographystyle{utphys}

\end{document}